\documentclass[a4paper,USenglish,cleveref,autoref,thm-restate,numberwithinsect]{lipics-arxiv-v2021}

\usepackage{mathtools}
\usepackage{tikz}
\usepackage{pgfplots}
\usepackage{float}
\usepackage{hyperref}

\numberwithin{equation}{section}
\numberwithin{figure}{section}

\usetikzlibrary{patterns,arrows,decorations.pathreplacing,decorations.pathmorphing,calc}

\tikzstyle{normalvertex}=[circle,fill=white,draw=black]
\tikzstyle{tinyvertex}=[circle,fill=white,draw=black,minimum size=5pt,inner sep=0pt]

\definecolor{darkpastelgreen}{rgb}{0.01, 0.75, 0.24}

\newtheorem{assumption}[theorem]{Assumption}

\newcommand{\notleftright}{\mathrel{\ooalign{$\Leftrightarrow$\cr\hidewidth$/$\hidewidth}}}

\DeclareMathOperator{\dist}{dist}
\DeclareMathOperator{\ISOTYPE}{ISOTYPE}

\DeclareMathOperator{\BP}{BP}
\DeclareMathOperator{\width}{width}
\DeclareMathOperator{\tw}{tw}

\DeclareMathOperator{\CFI}{CFI}
\DeclareMathOperator{\twCFI}{{CFI^{\textsf{x}}}}

\DeclareMathOperator{\dimWL}{dim_{\sf WL}}

\title{The Power of the Weisfeiler-Leman Algorithm to Decompose Graphs}

\titlerunning{Weisfeiler-Leman and Graph Separators} 

\author{Sandra Kiefer}{RWTH Aachen University, Aachen, Germany}{kiefer@cs.rwth-aachen.de}{https://orcid.org/0000-0003-4614-9444}{}
\author{Daniel Neuen}{CISPA Helmholtz Center for Information Security, Saarbrücken, Germany}{daniel.neuen@cispa.saarland}{https://orcid.org/0000-0002-4940-0318}{}

\authorrunning{S. Kiefer and D. Neuen} 

\Copyright{Sandra Kiefer and Daniel Neuen} 


\nolinenumbers 

\PlatformAcronym{arXiv}

\begin{document}

\maketitle

\begin{abstract}
 The Weisfeiler-Leman procedure is a widely-used technique for graph isomorphism testing that works by iteratively computing an isomorphism-invariant coloring of vertex tuples.
 Meanwhile, a fundamental tool in structural graph theory, which is often exploited in approaches to tackle the graph isomorphism problem, is the decomposition into $2$- and $3$-connected components.
 
 We prove that the $2$-dimensional Weisfeiler-Leman algorithm implicitly computes the decomposition of a graph into its $3$-connected components.
 This implies that the dimension of the algorithm needed to distinguish two given non-isomorphic graphs is at most the dimension required to distinguish non-isomorphic $3$-connected components of the graphs (assuming dimension at least~$2$). 
 
 To obtain our decomposition result, we show that, for $k \geq 2$, the $k$-dimensional algorithm distinguishes $k$-separators, i.e., $k$-tuples of vertices that separate the graph, from other vertex $k$-tuples.
 As a byproduct, we also obtain insights about the connectivity of constituent graphs of association schemes.
 
 In an application of the results, we show the new upper bound of $k$ on the Weisfeiler-Leman dimension of the class of graphs of treewidth at most $k$.
 Using a construction by Cai, F\"urer, and Immerman, we also provide a new lower bound that is asymptotically tight up to a factor of $2$.
\end{abstract}

\section{Introduction}

Originally introduced in \cite{WeisfeilerL68}, the Weisfeiler-Leman (WL) algorithm has become a -- if not the -- fundamental combinatorial subroutine in the context of isomorphism testing for graphs.
It is used both in theoretical and in practical approaches to tackle the graph isomorphism problem (see, e.g., \cite{BabaiCSTW13,GroheN19,McKay81,McKayP14,Neuen16}), 
among them is also Babai's recent quasipolynomial-time isomorphism test \cite{Babai16}.
For every $k \geq 1$, there is a $k$-dimensional version of the algorithm ($k$-WL) which colors the vertex $k$-tuples of the input graph and iteratively refines the coloring in an isomorphism-invariant manner.

There are various characterizations of the algorithm, which link it to other areas in theoretical computer science (see also \hyperref[ref:frw]{\textit{Further Related Work}}).
For example, recent results in the context of machine learning show that $1$-WL is as expressive as graph neural networks with respect to distinguishing graphs \cite{MorrisRFHLRG18}.
Following Grohe \cite{Grohe17}, an indicator to investigate the expressive power of the algorithm is the so-called \emph{WL dimension} of a graph, defined as the minimal dimension of the WL algorithm required in order to distinguish the graph from every other non-isomorphic graph.

There is no fixed dimension of the algorithm that decides graph isomorphism in general, as was proved by Cai, F\"urer, and Immerman \cite{CaiFI92}.
Still, when focusing on particular graph classes, often a bounded dimension of the algorithm suffices to identify every graph in the class.
This proves that for the considered class, graph isomorphism is solvable in polynomial time, since $k$-WL can be implemented to run in time $\mathcal{O}(n^{k+1} \log n)$ \cite{IL90}.
For example, it suffices to apply $3$-WL to identify every planar graph \cite{kieponschwei19}.
Also, the WL dimension of graphs of treewidth at most $k$ is bounded by $k+2$ \cite{GroheM99}.
More generally, by a celebrated result due to Grohe, for all graph classes with an excluded minor, the WL dimension is bounded \cite{Grohe12}.
Recent work provides explicit upper bounds on the WL dimension, which are linear in the rank width \cite{GroheN19} and in the Euler genus \cite{GroheK19}, respectively, of the graph.

Regarding combinatorial techniques, to handle graphs with complex structures, decompositions into connected, biconnected, and triconnected components provide an important tool from structural graph theory.
They can be computed in linear time (see, e.g., \cite{HopcroftT73,Tutte84}).
Hopcroft and Tarjan used the decomposition of a graph into its triconnected components to obtain an algorithm that decides isomorphism for planar graphs in quasi-linear time \cite{HopcroftT71,HopcroftT72,HopcroftT73-a}, which was improved to linear time by Hopcroft and Wong \cite{HopcroftW74}. 

Also, in \cite{kieponschwei19}, to prove the bound on the WL dimension for the class of planar graphs, the challenge of distinguishing two arbitrary planar graphs is reduced to the case of two arc-colored triconnected planar graphs, by exploiting the fact that $3$-WL is able to implicitly compute the decomposition of a graph into its triconnected components.
Similarly, the bound on the WL dimension for graphs parameterized by their Euler genus from \cite{GroheK19} relies on an isomorphism-invariant decomposition of the graphs into their triconnected components.

\subparagraph{Our Contribution}

For every $k \geq 2$, we show that $k$-WL implicitly computes the decomposition into the triconnected components of a given graph.
More specifically, we prove that already $2$-WL distinguishes separating pairs, i.e., pairs of vertices that separate the given graph, from other vertex pairs.
This improves on a result from \cite{kieponschwei19}, where an analogous statement was proved for $3$-WL.
Using the decomposition techniques discussed there, we conclude that, for $k$-WL with $k \geq 2$, to identify a graph, it suffices to determine vertex orbits on all arc-colored $3$-connected components of it.
Since it is easy to see that $k = 1$ does not suffice to distinguish vertices contained in $2$-separators from others (see, e.g., \cite[Chapter~7]{kiefer:phd}), our upper bound of~$2$ is tight.

The expressive power of $k$-WL corresponds to definability in the logic ${\sf C}^{k+1}$, the extension of the $(k+1)$-variable fragment of first-order logic by counting quantifiers \cite{CaiFI92,IL90}.
Exploiting this correspondence, our results imply that, for every $n \in \mathbb{N}$, there is a formula $\varphi_n(x_1,x_2) \in {\sf C}^{3}$ such that for every $n$-vertex graph $G$, it holds that $G \models \varphi_n(v,w)$ if and only if $\{v,w\}$ is a $2$-separator in $G$ (where 
$G\models\varphi(v,w)$ means that $G$ satisfies $\varphi$ if
$x_1$ is mapped to $v$ and $x_2$ is mapped to $w$). To obtain such a formula with only three variables at our disposal, it is not possible to take the route of \cite{kieponschwei19} by comparing certain numbers of walks between different pairs of vertices.
Instead, the formulas obtained from our proof are essentially a disjunction over all $n$-vertex graphs and subformulas for two distinct graphs may look completely different, exploiting specific structural properties of the graphs.
While this makes the proof rather involved, it also stresses the power of $2$-WL and equivalently, the expressive power of the logic ${\sf C}^{3}$. We actually show that for all $n,s \in \mathbb{N}$, there is a formula $\varphi_{n,s}(x_1,x_2,x_3) \in {\sf C}^{3}$ such that, for every $n$-vertex graph $G$, it holds that $G \models \varphi_{n,s}(u,v,w)$ if and only if $s = |C|$, where~$C$ is the vertex set of the connected component containing $u$ after removing $v$ and $w$ from the graph $G$.

Our result can also be viewed in a combinatorial setting.
In 1985, Brouwer and Mesner~\cite{BrouwerM85} proved that the vertex connectivity of a strongly regular graph equals its degree and that, in fact, the only minimal disconnecting vertex sets are vertex neighborhoods.
Later, Brouwer conjectured this to be true for every constituent graph of an association scheme (i.e., every graph consisting in a single color class of the association scheme) \cite{Brouwer96}.
While some progress has been made on certain special cases \cite{EvdokimovP06}, most prominently distance-regular graphs \cite{BrouwerK09}, the general question is still open.
Our results imply that every connected constituent graph of an association scheme is either a cycle or $3$-connected.
Such a statement was previously only known for symmetric association schemes \cite{KodalenM17}, which are far more restricted than the general ones.

A natural use case of these results is to determine or to improve upper bounds on the WL dimension of certain graph classes.
As a first application in this direction, we obtain a new upper bound of $k$ on the WL dimension for graphs of treewidth at most $k$.
Based on~\cite{DawarR07}, we also provide a new lower bound for this graph class, thus delimiting the value of the WL dimension of the class of graphs of treewidth bounded by $k$ to the interval $\left[\lceil\frac{k}{2}\rceil-2, k\right]$.

\subparagraph{Further Related Work}\label{ref:frw}

Apart from its correspondence to counting logics, the WL algorithm has surprising links to other areas.
For example, the algorithm has a close connection to Sherali-Adams relaxations of particular linear programs \cite{AtseriasM13,GroheO15} and captures the same information as certain homomorphism counts \cite{Dvorak10}.
It can also be characterized via winning strategies in so-called pebble games~\cite{Hella96}, which are a particular family of Ehrenfeucht-Fraissë games.
The survey column \cite{Kiefer20} discusses the results of this article in the context of a general treatment of the WL algorithm and its dimension.

As mentioned above, $1$-WL essentially corresponds to graph neural networks.
In order to make them more powerful, the authors of \cite{MorrisRFHLRG18} propose an extension of graph neural networks based on $k$-WL (see also \cite{MorrisRM20}).

Towards understanding the expressive power of the algorithm, in a related direction of research, it has been studied which graph properties the WL algorithm can detect, which may become particularly relevant in the graph-learning framework.
In this context, F\"urer~\cite{Furer17} as well as Arvind et al.\ \cite{ArvindFKV20} have obtained results concerning the ability of the algorithm to detect and count certain subgraphs. 

\section{Preliminaries}
\label{sec:prelim}

\subsection{Graphs}

A \emph{graph} is a pair $G = \big(V(G),E(G)\big)$ of a \emph{vertex set} $V(G)$ and an \emph{edge set} $E(G) \subseteq \big\{\{u,v\} \, \big\vert \, u,v \in V(G)\big\}$.
Unless explicitly stated otherwise, every graph considered in this paper is finite, undirected, and simple (i.e., it contains no loops or multiple edges).
For $v,w \in V(G)$, we also write $vw$ as a shorthand for $\{v,w\}$.
The \emph{neighborhood} of~$v$ is denoted by~$N(v)$, and the \emph{closed neighborhood} of $v$ is $N[v] \coloneqq N(v) \cup \{v\}$.
The \emph{degree} of~$v$, denoted by $\deg(v)$, is the number of edges incident with $v$.
For $A \subseteq V(G)$ we define $N(A) \coloneqq \left(\bigcup_{v \in A}N(v)\right) \setminus A$.

A \emph{walk of length $k$} from $v$ to $w$ is a sequence of vertices $v = u_0,u_1,\dots,u_k = w$ such that $u_{i-1}u_i \in E(G)$ for all $i \in \{1,\dots,k\}$.
A \emph{path of length $k$} from $v$ to $w$ is a walk of length $k$ from $v$ to $w$ for which all occurring vertices are pairwise distinct.
We refer to the \emph{distance} between two vertices $v,w \in V(G)$ by $\dist(v,w)$.
For a set $A \subseteq V(G)$, we denote by $G[A]$ the \emph{induced subgraph} of $G$ on vertex set $A$.
Also, we denote by $G - A$ the subgraph induced by the complement of $A$, that is, the graph $G - A \coloneqq G[V(G) \setminus A]$.
A set $S \subseteq V(G)$ is a \emph{separator} of $G$ if $G - S$ has more connected components than $G$.
A \emph{$k$-separator} of $G$ is a separator of $G$ of size $k$.
A vertex $v \in V(G)$ is a \emph{cut vertex} if $\{v\}$ is a separator of $G$.
The graph $G$ is \emph{$k$-connected} if it is connected and has no separator of size at most $k-1$.

An \emph{isomorphism} from $G$ to another graph $H$ is a bijection $\varphi\colon V(G) \rightarrow V(H)$ that respects the edge relation, that is, for all~$v,w \in V(G)$, it holds that~$vw \in E(G)$ if and only if $\varphi(v)\varphi(w) \in E(H)$.
Two graphs $G$ and $H$ are \emph{isomorphic} ($G \cong H$) if there is an isomorphism from~$G$ to~$H$.
We write $\varphi\colon G\cong H$ to denote that $\varphi$ is an isomorphism from $G$ to $H$.

A \emph{vertex-colored graph} is a tuple $(G,\lambda)$, where $G$ is a graph and $\lambda\colon V(G) \rightarrow \mathcal{C}$ is a mapping into some set $\mathcal{C}$ of colors.
Similarly, an \emph{arc-colored graph} is a tuple $(G,\lambda)$, where $G$ is a graph and $\lambda\colon \{(v,v) \mid v \in V(G)\} \cup \{(u,v) \mid \{u,v\} \in E(G)\} \rightarrow \mathcal{C}$ is a mapping into some color set $\mathcal{C}$.
Typically, $\mathcal{C}$ is chosen to be an initial segment $[n] \coloneqq \{1,\dots,n\}$ of the natural numbers.
Isomorphisms between vertex- and arc-colored graphs have to respect the colors of the vertices and arcs.
Usually, we will not distinguish between a colored graph $(G,\lambda)$ and its uncolored version $G$, and also use $G$ to denote the colored graph $(G,\lambda)$.

We recall the definition of the treewidth of a graph.
For more background on tree decompositions and treewidth, we refer the reader to \cite{Kloks94}.
Let $G$ be a graph. A \emph{tree decomposition} of $G$ is a pair $(T,\beta)$ where $T$ is a tree and $\beta\colon V(T) \rightarrow 2^{V(G)}$ is a function (where $2^{V(G)}$ denotes the power set of $V(G)$) such that
\begin{enumerate}
 \item for every $v \in V(G)$, the set $\{t \in V(T) \mid v \in \beta(t)\}$ is non-empty and induces a connected subgraph in $T$, and
 \item for every $e \in E(G)$, there is a $t \in V(T)$ such that $e \subseteq \beta(t)$.
\end{enumerate}
The sets $\beta(t)$ for $t \in V(T)$ are the \emph{bags} of the tree decomposition. The \emph{width} of a tree decomposition $(T,\beta)$ is
\[\width(T,\beta) \coloneqq \max_{t \in V(T)} |\beta(t)| - 1.\]
The \emph{treewidth} of $G$, denoted by $\tw(G)$, is the minimum width of a tree decomposition of $G$.

\subsection{The Weisfeiler-Leman Algorithm}

Let~$\chi,\chi'\colon V^k \rightarrow \mathcal{C}$ be colorings of the~$k$-tuples of vertices of~$G$, where~$\mathcal{C}$ is some finite set of colors. 
We say $\chi'$ \emph{refines} $\chi$ if for all $\bar{v},\bar{w} \in V^k$, we have that $\chi'(\bar{v}) = \chi'(\bar{w})$ implies $\chi(\bar{v}) = \chi(\bar{w})$.
The $k$-dimensional WL algorithm is a procedure that, given a graph~$G$ and a coloring $\chi$ of its~$k$-tuples of vertices, computes an isomorphism-invariant coloring that refines~$\chi$.
(Here, \emph{isomorphism invariance} means that every isomorphism from the graph to a second graph preserves the colors of vertex $k$-tuples in the computed colorings on the graphs.)

We describe the mechanisms of the algorithm in the following. For an integer~$k > 1$ and a vertex-colored graph~$(G,\lambda)$, we let $\chi_{G,k}^0\colon V^{k} \rightarrow \mathcal{C}$ be the coloring where each $k$-tuple is colored with the isomorphism type of its underlying ordered colored subgraph. 
More formally, $\chi_{G,k}^0(v_1,\dots,v_k) = \chi_{G,k}^0(w_1,\dots,w_k)$ if and only if for all $i\in [k]$, it holds that $\lambda(v_i)= \lambda(w_i)$ and for all $i,j\in [k]$, it holds that $v_i = v_j \Leftrightarrow w_i =w_j$ and $v_iv_j \in E(G) \Leftrightarrow w_iw_j \in E(G)$.
If~$G$ is arc-colored, the arc colors must be respected accordingly.

We then recursively define the coloring~$\chi_{G,k}^{i}$ obtained after $i$ rounds of the algorithm. Let $\chi_{G,k}^{i+1}(v_1, \dots, v_k) \coloneqq (\chi_{G,k}^{i} (v_1, \dots, v_k); \mathcal{M})$,
where~$\mathcal{M}$ is a multiset defined as
\[\big\{\!\!\big\{\big(\chi_{G,k}^{i}(\bar v[w/1]),\chi_{G,k}^{i}(\bar v[w/2]), \dots, \chi_{G,k}^{i}(\bar v[w/k])\big) \, \big\vert \, w\in V \big\}\!\!\big\}\]
where $\bar v[w/i] \coloneqq (v_1,\dots,v_{i-1},w,v_{i+1},\dots,v_k)$.

For $1$-WL (i.e., $k=1$), the definition is similar, but we iterate only over the neighbors of~$v_1$, that is, the multiset $\mathcal{M}$ equals $\{\!\! \{ \chi_{G,1}^i(w) \mid w\in N(v_1) \}\!\!\}$.

By definition, every coloring~$\chi_{G,k}^{i+1}$ induces a refinement of the partition of the~$k$-tuples of vertices of the graph~$G$ into the color classes of~$\chi_{G,k}^{i}$.
Thus, there is a minimal~$i$ such that the partition of the vertex $k$-tuples induced by~$\chi_{G,k}^{i+1}$ is not strictly finer than the one induced by~$\chi_{G,k}^i$.
For this value of~$i$, we call the coloring~$\chi_{G,k}^i$ the \emph{final} coloring of~$G$ and denote it by $\chi_{G,k}$. If $G$ is clear from the context or irrelevant, we also use the notation $\chi_k$ for $\chi_{G,k}$.

The original WL algorithm is its $2$-dimensional variant \cite{WeisfeilerL68}.
Since that version is the central algorithm of this paper, we omit the index $2$ and write $\chi_{G}$ instead of $\chi_{G,2}$.

For each~$k \in \mathbb{N}$, the \emph{$k$-dimensional WL algorithm} ($k$-WL) takes as input a (vertex- or arc-)colored graph $(G, \lambda)$ and returns the final coloring~$\chi_{G,k}$.

For two graphs~$G$ and~$H$, we say that~$k$-WL \emph{distinguishes}~$G$ and~$H$ if there is a color~$c$ such that the sets
$\{\bar{v} \mid \bar{v} \in \big(V(G)\big)^{k}, \chi_{G,k}(\bar{v}) = c\}$ and $\{\bar{w} \mid \bar{w} \in \big(V(H)\big)^{k}, \chi_{H,k}(\bar{w}) = c\}$
have different cardinalities. We write~$G \simeq_k H$ if $k$-WL does not distinguish~$G$ and~$H$.
The algorithm \emph{identifies} $G$ if it distinguishes $G$ from every non-isomorphic graph $H$.

We also regularly exploit the fact that $k$-WL can detect certain graph properties.
Let $f_G\colon (V(G))^\ell \rightarrow S$ be a function defined on $\ell$-tuples of vertices for a graph $G$ (for example, $f_G$ might be the distance function for pairs of vertices).
For $k \geq \ell$, we say that $k$-WL \emph{knows $f_G$} if the implication
\begin{align*}
                 &f_G(v_1,\dots,v_\ell) \neq f_H(w_1,\dots,w_\ell) \\
 \Rightarrow\;\; &\chi_{G,k}(v_1,\dots,v_\ell,\dots,v_\ell) \neq \chi_{H,k}(w_1,\dots,w_\ell,\dots,w_\ell)
\end{align*}
holds for all graphs $G,H$ and vertices $v_1,\dots,v_\ell \in V(G)$ and $w_1,\dots,w_\ell \in V(H)$.
For a property $P_G \subseteq (V(G))^\ell$, we say that $k$-WL \emph{knows $P_G$ (in $G$)} if $k$-WL knows the characteristic function of $P_G$.
For example, the statement that $k$-WL knows whether two vertices are at distance $d$ from each other and that $k$-WL knows the set of pairs of vertices of distance $d$ from each other can be used equivalently.

Finally, if $V(H) \subseteq V(G)$, we say that $k$-WL \emph{knows $H$ (in $G$)} if $k$-WL knows $V(H)$ and~$E(H)$.

\subparagraph{Pebble Games}

For further analysis, it is often cumbersome to work with the WL algorithm directly and more convenient to use the following characterization via pebble games, which is known to capture the same information.
Let $k \in \mathbb{N}$.
For graphs $G$ and $H$ on the same number of vertices and with vertex colorings~$\lambda_G$ and~$\lambda_H$, respectively,
we define the \emph{bijective $k$-pebble game} $\BP_{k}(G,H)$ as follows:
\begin{itemize}
 \item The game has two players called \emph{Spoiler} and \emph{Duplicator}.
 \item The game proceeds in rounds, each of which is associated with a pair of positions
 $(\bar v,\bar w)$ with~$\bar v \in \big(V(G)\big)^\ell$ and~$\bar w \in \big(V(H)\big)^\ell$, where $0 \leq \ell \leq k$.
 \item The initial position of the game is a pair of vertex tuples of equal length $\ell$ with $0 \leq \ell \leq k$. If not specified otherwise, the initial position is the pair $\big((),()\big)$ of empty tuples.
 \item Each round consists of the following steps. Suppose the current position of the game is $(\bar v,\bar w) = ((v_1,\ldots,v_\ell),(w_1,\ldots,w_\ell))$.
  First, Spoiler chooses whether to remove a pair of pebbles or to play a new pair of pebbles.
  The first option is only possible if $\ell > 0$, and the latter option is only possible if $\ell < k$.
  
  If Spoiler wishes to remove a pair of pebbles, he picks some $i \in [\ell]$ and the game moves to position
  $(\bar v\setminus i,\bar w\setminus i)$ where $\bar v \setminus i \coloneqq (v_1,\dots,v_{i-1},v_{i+1},\dots,v_\ell)$, and the tuple ($\bar w \setminus i)$ is defined in the analogous way.
  Otherwise, the following steps are performed.
  \begin{itemize}
   \item[(D)] Duplicator picks a bijection $f\colon V(G) \rightarrow V(H)$.
   \item[(S)] Spoiler chooses $v \in V(G)$ and sets $w \coloneqq f(v)$.
    Then the game moves to position $\big((v_1,\dots,v_\ell,v),(w_1,\dots,w_\ell,w)\big)$.
  \end{itemize}

  Relabel vertices so that~$\big((v_1,\dots,v_\ell),(w_1,\dots,w_\ell)\big)$ is the current position. If mapping each $v_i$ to $w_i$ does not define an isomorphism of the induced subgraphs of $G$ and $H$, Spoiler wins the play.
  More precisely, Spoiler wins if there is an~$i\in [\ell]$ such that $\lambda_G(v_i)\neq \lambda_H(w_i)$, or there are~$i,j\in [\ell]$ such that~$v_i = v_j\notleftright w_i =w_j$ or~$v_iv_j \in E(G)\notleftright w_iw_j \in E(H)$.
  If there is no position of the play such that Spoiler wins, then Duplicator wins.
\end{itemize}

We say that Spoiler (and Duplicator, respectively) \emph{wins $\BP_k(G,H)$} if Spoiler (and Duplicator, respectively) has a winning strategy for the game.
Also, for a position $(\bar v,\bar w)$ with $\bar v \in \big(V(G)\big)^\ell$ and $\bar w \in \big(V(H)\big)^\ell$, we say that Spoiler (and Duplicator, respectively) \emph{wins $\BP_k(G,H)$ from position $(\bar v,\bar w)$} if Spoiler (and Duplicator, respectively) has a winning strategy for the game started at position $(\bar v,\bar w)$.

The next theorem describes the correspondence between the WL algorithm and the introduced pebble games.\footnote{The pebble games in \cite{CaiFI92,Grohe17} are defined slightly differently. Still, a player has a winning strategy in the game described there if and only if they have one in our game and thus, Theorem \ref{thm:eq-wl-pebble-tuples} holds for both versions of the game.}

\begin{theorem}[see, e.g., \cite{CaiFI92, Grohe17}]
 \label{thm:eq-wl-pebble-tuples}
 Let $k \geq 1$.
 Let $G$ and $H$ be two graphs and let $\bar v \in (V(G))^k$ and $\bar w \in (V(H))^k$.
 Then $\chi_{G,k}(\bar v) = \chi_{H,k}(\bar w)$ if and only if Duplicator wins the game $\BP_{k+1}(G,H)$ from the initial position $(\bar v,\bar w)$.
\end{theorem}

\begin{corollary}
 \label{cor:eq-wl-pebble}
 Let $G$ and $H$ be two graphs.
 Then $G \simeq_{k} H$ if and only if Duplicator wins the game $\BP_{k+1}(G,H)$.
\end{corollary}

\subparagraph{Association Schemes}

Let $V$ be a set. An \emph{association scheme} on $V$ is an ordered partition $(R_0,\dots,R_d)$ of $V^{2}$ such that
\begin{enumerate}
 \item $R_0 = \{(v,v) \mid v \in V\}$, and
 \item for every $i \in [d]$, there is a $j \in [d]$ such that the set $R_i^{\intercal} \coloneqq \{(w,v) \mid (v,w) \in R_i\}$ equals $R_j$, and
 \item for all $i,j,k \in [d]$, there are numbers $p_{i,j}^{k}$ such that for all $(v,w) \in R_k$,
 \[\big|\big\{x \in V \, \big\vert \, (v,x) \in R_i\ \text{ and } (w,x) \in R_j\big\}\big| = p_{i,j}^{k} \, .\]
\end{enumerate}
An association scheme is \emph{symmetric} if $R_i^{\intercal} = R_i$ for all $i \in [d]$.
With each $R_i$, we associate a \emph{constituent graph $G(R_i)$ of the association scheme}, defined as $G(R_i) \coloneqq (V,R_i \cup R_i^{\intercal})$. Note that, formally, $G(R_i)$ is directed because its edges are ordered pairs. However, since for every pair of vertices $v,w \in V$, it holds that $(v,w) \in E(G(R_i))$ if and only if $(w,v) \in E(G(R_i))$, we may also consider $G(R_i)$ as undirected, i.e., $E(G(R_i))$ as a set of unordered pairs of vertices.

Every association scheme induces a coloring on $V^{2}$, in which every $(v,w)$ is colored with the relation it is contained in.
This coloring is \emph{stable with respect to $2$-WL}, 
 which means that it is not refined by $2$-WL (when $V$ is interpreted as the vertex set of a complete directed graph).
In particular, for all $i \in [d]$ and all $v,w \in V$, it holds that $\chi_{G(R_i)}(v,v) = \chi_{G(R_i)}(w,w)$.
Conversely, every colored graph $G$ with $\chi_{G}(v,v) = \chi_{G}(w,w)$ for all $v,w \in V(G)$ induces an association scheme in which the relations $R_i$ are the color classes of the coloring $\chi_{G} = \chi_{G,2}$.

\section{One Color}
\label{sec:onecolor}

Our first goal is to prove that $2$-WL distinguishes vertex pairs that are separators in a graph from other pairs of vertices.
This result forms the central technical contribution of this paper.
We first sketch the proof strategy.
Consider two graphs $G$ and $H$ and vertices $v_1,v_2 \in V(G)$ and $w_1,w_2 \in V(H)$ such that $\chi_G(v_1,v_2) = \chi_H(w_1,w_2)$.
We want to argue that $\{v_1,v_2\}$ forms a $2$-separator in $G$ if and only if $\{w_1,w_2\}$ forms a $2$-separator in~$H$.
Since $\chi_G(v_1,v_2) = \chi_H(w_1,w_2)$, there are (possibly equal) colors $c_1, c_2$ with $\{c_1,c_2\} = \{\chi_G(v_1,v_1),\chi_G(v_2,v_2)\} = \{\chi_H(w_1,w_1),\chi_H(w_2,w_2)\}$.

As it turns out (and as we show in detail in Section \ref{sec:separators}), the structural analysis of a graph with an arbitrary number of $\chi_2$-colors reduces to the case that there are at most two $\chi_2$-colors present in the graph. The crucial observation for the reduction is that one can bypass all vertices whose $\chi_2$-color is not contained in $\{c_1,c_2\}$. Here, bypassing a vertex $u$ means removing $u$ from the graph and connecting every pair of its neighbors via an edge.  

Hence, we first focus on input graphs $G$ for which $|\{\chi_G(v,v) \mid v \in V(G)\}| \leq 2$. In this section, we start with an analysis of the graphs in which all vertices are assigned the same color by $2$-WL.
In particular, this includes all constituent graphs of association schemes.
In the next section, we will extend the techniques to graphs $G$ for which $|\{\chi_G(w,w) \mid w \in V(G)\}| = 2$ holds.
Finally, the reduction to two vertex colors is presented in Section \ref{sec:separators}. 

So, for now, suppose that all vertices of the input graph are assigned the same color by $2$-WL.
A main tool for the analysis are distance patterns of vertices.
For a graph $G$ and a vertex $v \in V(G)$, let $D(v) \coloneqq \{\!\{\dist(v,w) \mid w \in V(G)\}\!\}$.
Note that, for vertices $u, v \in V(G)$, it holds that $\chi_G(u,u) \neq \chi_G(v,v)$ whenever $D(u) \neq D(v)$, since $2$-WL knows the distance function for pairs of vertices (see, e.g., \cite[Theorem 2.6.7]{ChenP19}).

\begin{lemma}
 \label{la:distance-patterns-small}
 Let $G$ be a graph and $uv \in E(G)$. Suppose that $D(u) = D(v)$. Then
 \begin{align*}
      &\{\!\{\dist(u,w) \mid w \in V(G) \colon \dist(u,w) < \dist(v,w)\}\!\}\\
  =\; &\{\!\{\dist(v,w) \mid w \in V(G) \colon \dist(v,w) < \dist(u,w)\}\!\}.
 \end{align*}
\end{lemma}

\begin{proof}
 We have $|\dist(v,w) - \dist(u,w)| \leq 1$ for all $w \in V(G)$ since $uv \in E(G)$.
 Suppose the statement is false and let $d \in \mathbb{N}$ be the maximal number such that $d$ has distinct multiplicities in the two multisets. Let $m_1$ be the multiplicity of $d$ in the first multiset and $m_2$ be the multiplicity of $d$ in the second.
 Without loss of generality, assume that $m_1 > m_2$.
 Then
 \[m_1 = |\{w \in V(G) \mid \dist(u,w) = d \wedge \dist(v,w) = d+1\}|\]
 and
 \[m_2 = |\{w \in V(G) \mid \dist(v,w) = d \wedge \dist(u,w) = d+1\}|.\]
 But
 \begin{align*}
      &|\{w \in V(G) \mid \dist(u,w) = d+1 \wedge \dist(v,w) = d+1\}|\\
  =\; &|\{w \in V(G) \mid \dist(v,w) = d+1 \wedge \dist(u,w) = d+1\}|
 \end{align*}
 and
 \begin{align*}
      &|\{w \in V(G) \mid \dist(u,w) = d+2 \wedge \dist(v,w) = d+1\}|\\
  =\; &|\{w \in V(G) \mid \dist(v,w) = d+2 \wedge \dist(u,w) = d+1\}|,
 \end{align*}
 where the first equality is trivial and the second equality follows from the maximality of~$d$.
 However, then $D(u)$ and $D(v)$ contain the number $d+1$ in distinct multiplicities, a contradiction.
\end{proof}

Throughout the remainder of this section, if not explicitly stated otherwise, we make the following assumption.

\begin{assumption}
 \label{ass:wl-2sep}
 $G$ is a connected graph on $n$ vertices with the following properties:
 \begin{enumerate}
  \item\label{item:ass-color} $\chi_G(u,u) = \chi_G(v,v)$ for all $u,v \in V(G)$, and
  \item\label{item:ass-separator} $G$ has a $2$-separator $\{w_1,w_2\}$.
 \end{enumerate}
\end{assumption}

In the rest of this section, we analyze the structure of $G$ and ultimately prove that $G$ must be a cycle.
In particular, this completely characterizes constituent graphs of association schemes that are connected, but not $3$-connected.

Before diving into the technical details, we again give some intuition.
Let $d \coloneqq \dist(w_1,w_2)$ and let $C$ be a connected component of $G - \{w_1,w_2\}$ such that $|C| \leq \frac{n-2}{2}$.
We first argue that $\dist(v,w_i) \leq d$ for all $v \in C$ and $i \in \{1,2\}$ by analyzing distance patterns of vertices.
Intuitively speaking, if $v$ is far away from both $w_1$ and $w_2$, then more than half of the vertices of $G$ (specifically, all vertices in $V(G)\setminus C$) are closer to $w_1$ and $w_2$ than to $v$, which eventually contradicts the fact that all distance patterns are identical.

Then, we split the vertices into subclasses $C_{i,j} \coloneqq \{v \in C \mid \dist(v,w_1) = i, \ \dist(v,w_2) = j\}$ for $i,j \leq d$ and analyze the connections between those sets.
To this end, we may assume that the edge set of the input graph is maximal (among all graphs satisfying Assumption~\ref{ass:wl-2sep}).
To see how we can exploit this property, consider the simple example of four vertices $v_1,v_2,v_3,v_4 \in C$ such that $v_1v_2,v_2v_3,v_3v_4,v_4v_1 \in E(G)$ (i.e., the vertices form a $4$-cycle) and $d \geq 3$.
Also consider a pair $u_1,u_3 \in V(G)$ such that $\chi_G(u_1,u_3) = \chi_G(v_1,v_3)$.
By the properties of $2$-WL, there are $u_2,u_4 \in V(G)$ such that $u_1u_2,u_2u_3,u_3u_4,u_4u_1 \in E(G)$.
In particular, since $d \geq 3$, we know that $u_1$ and $u_3$ are either both contained in $C \cup \{w_1,w_2\}$ or they are both not contained in $C$.
Hence, by inserting all edges $u_1u_3$ for which $\chi_G(v_1,v_3) = \chi_G(u_1,u_3)$ holds into the graph, we preserve Condition \ref{item:ass-separator} from Assumption \ref{ass:wl-2sep}.
Also, $\chi_G$ remains a stable coloring of the updated graph, which implies that Condition \ref{item:ass-color} from Assumption \ref{ass:wl-2sep} is preserved.
As a result, we may assume that all cycles on four vertices induce a clique which, in particular, implies that every pair of vertices is connected by at most one shortest path of length~$2$.

By extending this argument to cycles of length strictly less than $2d$, we obtain that shortest paths of length strictly less than $d$ are unique.
This provides us with strong restrictions on the connections between the sets $C_{i,j}$.
Ultimately, this allows us to prove that $G$ is a cycle by showing the existence of a vertex of degree $2$. 
Note that Condition \ref{item:ass-color} of Assumption \ref{ass:wl-2sep} implies that $G$ is regular, i.e., $\deg(u) = \deg(v)$ for all $u,v \in V(G)$.

\begin{lemma}
 \label{la:as-no-cut-vertex}
 $G$ is $2$-connected, i.e., $G$ does not contain any cut vertex.
\end{lemma}

This is a consequence of Condition \ref{item:ass-color} in Assumption \ref{ass:wl-2sep}, since $2$-WL distinguishes cut vertices from other vertices (see \cite[Corollary 7]{kieponschwei19}) and there is no graph in which every vertex is a cut vertex. 
To see this, note that in the unique block-cut tree decomposition of the graph into its $2$-connected components (see, e.g., \cite{Tutte84}), each leaf component contains a vertex that is not a cut vertex.

The lemma implies that each of $w_1$ and $w_2$ has at least one neighbor in each of the connected components of $G - \{w_1,w_2\}$.

\begin{lemma}
 \label{la:as-distance-patterns}
 Let $C$ be the vertex set of a connected component of $G - \{w_1,w_2\}$ such that $|C| < \frac{n}{2}$ and let $v \in C$.
 Then there is no vertex $u \in N(v)$ such that $\dist(u,w_1) < \dist(v,w_1)$ and $\dist(u,w_2) < \dist(v,w_2)$.
\end{lemma}

\begin{proof}
 Suppose towards a contradiction that such a vertex $u \in N(v)$ exists.
 For all $w \in V(G)$, we have $|\dist(v,w) - \dist(u,w)| \leq 1$, since $uv \in E(G)$.
 Furthermore, we have $\sum_{w \in V(G)} \big(\dist(v,w) - \dist(u,w)\big) = 0$ because $D(v) = D(u)$ due to Condition \ref{item:ass-color} in Assumption \ref{ass:wl-2sep}.
 But $\dist(v,w) > \dist(u,w)$ for all $w \in V(G) \setminus C$, and $|V(G) \setminus C| > \frac{n}{2}$.
 This is a contradiction.
\end{proof}

\begin{lemma}
 \label{la:as-small-side-distances}
 Let $d \coloneqq \dist(w_1,w_2)$ and let $C$ be the vertex set of a connected component of $G - \{w_1,w_2\}$ such that $|C| \leq \frac{n-2}{2}$.
 Then for all $v \in C \cup \{w_1,w_2\}$ and all $i \in \{1,2\}$, it holds that $\dist(v,w_i) \leq d$.
\end{lemma}

\begin{proof}
 By symmetry, it suffices to prove $\dist(v,w_2) \leq d$.
 The statement is proved by induction on $\ell \coloneqq \dist(v,w_1)$.
 For $\ell = 0$, it holds that $v = w_1$ and $\dist(w_1,w_2) = d$.
 So suppose the statement holds for all $u \in C \cup \{w_1,w_2\}$ with $\dist(u,w_1) \leq \ell$. Obviously, the statement is true if $v = w_1$ or $v = w_2$. So pick $v \in C$ with $\dist(v,w_1) = \ell+1$.
 Let $u \in N(v)$ such that $\dist(u,w_1) \leq \ell$. Then $\dist(v,w_2) \leq \dist(u,w_2) \leq d$ by Lemma \ref{la:as-distance-patterns} and the induction hypothesis.
\end{proof}

\begin{lemma}
 \label{la:as-dist-1}
 $w_1w_2 \notin E(G)$.
\end{lemma}

\begin{proof}
 Suppose towards a contradiction that $w_1w_2 \in E(G)$.
 Let $C$ be the vertex set of a connected component of $G - \{w_1,w_2\}$ such that $|C| \leq \frac{n-2}{2}$.
 By Lemma \ref{la:as-small-side-distances}, we conclude that $C \subseteq N(w_1) \cap N(w_2)$.
 Let $v \in C$. Since $G$ is $2$-connected, the vertex $w_1$ must have at least one neighbor in $V(G) \setminus C$, in addition to being adjacent to all vertices in $C$ and to~$w_2$.
 Thus, $\deg(w_1) \geq |C|+2 > |C|-1 + |\{w_1,w_2\}| \geq \deg(v)$, which contradicts $G$ being a regular graph.
\end{proof}

\begin{lemma}
 \label{la:as-dist-2}
 Suppose that $N(w_1) \cap N(w_2) \neq \emptyset$. Then $G$ is a cycle.
\end{lemma}

\begin{proof}
 By Lemma \ref{la:as-dist-1}, it holds that $w_1w_2 \notin E(G)$. Furthermore, by the assumption of the lemma, we have $\dist(w_1,w_2) = 2$. Let $C$ be the vertex set of a connected component of $G - \{w_1,w_2\}$ such that $|C| \leq \frac{n-2}{2}$. Also let $C' \coloneqq V(G) \setminus (C \cup \{w_1,w_2\})$.
 For $i,j \geq 1$, let
 \[C_{i,j} \coloneqq \{v \in C \mid \dist(v,w_1) = i \text{ and } \dist(v,w_2) = j\}.\]
 By Lemma \ref{la:as-small-side-distances}, we conclude that $C_{i,j} = \emptyset$ unless $(i,j) \in \{(1,1),(1,2),(2,1),(2,2)\}$.
 
 Suppose there exists $v \in C_{1,2}$.
 We have $D(v) = D(w_1)$ and, by Lemma \ref{la:as-no-cut-vertex}, also $N(w_1) \cap C' \neq \emptyset$.
 Thus, there is a vertex $u' \neq w_1$ such that $\dist(w_1,u') < \dist(v,u')$ and therefore, by Lemma \ref{la:distance-patterns-small}, there is also a vertex $u \neq v$ such that $\dist(v,u) < \dist(w_1,u)$.
 Note that $u = w_2$ is not possible since $\dist(v,w_2) = 2 = \dist(w_1,w_2)$.
 Suppose $u \in C'$.
 Then no shortest path from $v$ to $u$ uses $w_1$, but this implies that $\dist(v,u) = 2 + \dist(w_2,u) = \dist(w_1,w_2) + \dist(w_2,u) \geq \dist(w_1,u)$, a contradiction.
 Hence, $u \in C$, which implies that $\dist(w_1,u) \leq 2$.
 So $\dist(v,u) = 1$, meaning that $u \in N(v)$. 
 Therefore, by Lemma \ref{la:distance-patterns-small}, for every vertex $v' \neq w_1$ with $\dist(w_1,v') < \dist(v,v')$, it holds that $v' \in N(w_1)$.
 This implies that there is no $v' \in C'$ such that $\dist(w_1,v') = 2$ since such a vertex would satisfy $v' \neq w_1$, $v' \notin N(w_1)$, and $\dist(w_1,v') = 2 < 3 = \dist(v,v')$.
 Because $w_2$ is not a cut vertex (cf.\ Lemma \ref{la:as-no-cut-vertex}), from every $v' \in C'$, there is a path to $w_1$ that does not contain $w_2$. 
 However, this is only possible if there is no vertex $v' \in C'$ such that $\dist(w_1,v') > 1$. In other words, $C' \subseteq N(w_1)$.
 Since $G$ is regular and $|N(w_1) \setminus C'| \geq 1$, it follows that $\deg(v) \geq |C'| + 1$.
 But $N(v) \subseteq (C \cup \{w_1\}) \setminus \{v\}$, which implies $\deg(v) \leq |C|$.
 The combination of both inequalities yields $|C'| + 1 \leq |C|$, which implies $n = 2 + |C| + |C'| \leq 1 + 2|C| \leq n-1$, a contradiction.
 So $C_{1,2} = \emptyset$ and, by symmetry, it also holds that $C_{2,1} = \emptyset$.
 But then $C_{2,2} = \emptyset$ by Lemma \ref{la:as-distance-patterns}.
 
 So $C = C_{1,1}$, which means that $C \subseteq N(w_1) \cap N(w_2)$.
 In particular, $\deg(w_1) \geq |C| + 1$ since $|N(w_1) \cap C'| \geq 1$.
 Since $G$ is regular, this implies that $\deg(v) \geq |C| + 1$ for every $v \in C$, which is only possible if $N[v] = C \cup \{w_1,w_2\}$.
 Because $C \neq \emptyset$, this means that there is a vertex $v \in V(G)$ such that $G[N[v]]$ contains only one non-edge.
 Now by Condition \ref{item:ass-color} in Assumption \ref{ass:wl-2sep}, this also has to hold for $w_1$, and hence, since no vertex in $C \cap N(w_1)$ is adjacent to any vertex in $C' \cap N(w_1)$, it must hold that $\deg(w_1) = 2$.
 Therefore, by regularity, all vertices in $G$ have degree $2$ and thus, being connected, $G$ is a cycle.
\end{proof}

We are now ready to show that $G$ is a cycle. 
\begin{lemma}
 \label{la:as-main}
 $G$ is a cycle.
\end{lemma}

\begin{proof}
 Note that every graph that satisfies Assumption \ref{ass:wl-2sep} has at least $n$ edges because the graph must be regular and connected. We first argue that it suffices to prove the lemma for the case that $G$ is a graph with a maximum edge set that satisfies Assumption \ref{ass:wl-2sep}.
 Indeed, if it holds for such a graph $G$, then $G$ has $n$ edges, so every graph that satisfies Assumption~\ref{ass:wl-2sep} with an edge set that is not maximal with respect to inclusion would have to have less than $n$ edges. However, there are no such graphs. So in the following, we assume that the edge set of $G$ is maximal with respect to inclusion.
 
 Let $d \coloneqq \dist(w_1,w_2)$.
 By Lemmas \ref{la:as-dist-1} and \ref{la:as-dist-2}, we can assume that $d \geq 3$.
 Let $C$ be the vertex set of a connected component of $G - \{w_1,w_2\}$ of size $|C| \leq \frac{n-2}{2}$. Also let $C' \coloneqq V(G) \setminus (C \cup \{w_1,w_2\})$.
 For $i,j \geq 1$, let
 \[C_{i,j} \coloneqq \{v \in C \mid \dist(v,w_1) = i \text{ and } \dist(v,w_2) = j\}.\]
 By Lemma \ref{la:as-small-side-distances}, we conclude that $C_{i,j} = \emptyset$ unless $i,j \leq d$.
 Furthermore, by the definition of~$d$, we have that $C_{i,j} = \emptyset$ unless $i+j \geq d$.
 This situation is also visualized in Figure \ref{fig:distance-pyramid-app}.

 In the following, the strategy is to prove that $C_{i,j} = \emptyset$ whenever $i+j > d$.
 We first show that there cannot be any vertex in $C$ of distance at least $d-1$ both to $w_1$ and $w_2$ (see Claim \ref{cl:c-d-d}).
 For the remaining sets $C_{i,j}$, we exploit the edge maximality of $G$ to reduce the complexity of vertex connections in $C$ (see Claims \ref{cl:unique-shortest-path-app} and \ref{cl:shortest-path-detour-app}). For example, we argue that shortest paths between vertices of distance strictly less than $d$ are unique.
 This severely restricts possible connections between the sets $C_{i,j}$ and eventually allows us to argue that we can assume $C_{i,j} = \emptyset$ whenever $i+j > d$ (using Claims \ref{cl:two-neighbors} and \ref{cl:tool-claim}). Finally, we argue that $|C_{i,j}| = 1$ whenever $i+j = d$, again relying on the fact that certain shortest paths are unique (see the application of Claim \ref{cl:matching-graph}).
 
 \begin{figure}
  \centering
  \scalebox{0.85}{
  \begin{tikzpicture}
   \node[tinyvertex,fill=red,label=below:{$w_1$}] (w1) at (0,0) {};
   \node[tinyvertex,fill=red,label=below:{$w_2$}] (w2) at (8,0) {};
   
   \node[circle,minimum size=32pt,fill=red,opacity=0.4,text opacity=1] (c1) at (1.8,0.8) {\large $C_{1,3}$};
   \node[circle,minimum size=32pt,fill=red,opacity=0.4,text opacity=1] (c2) at (4,1.2) {\large $C_{2,2}$};
   \node[circle,minimum size=32pt,fill=red,opacity=0.4,text opacity=1] (c3) at (6.2,0.8) {\large $C_{3,1}$};
   \node[circle,minimum size=32pt,fill=red,opacity=0.4,text opacity=1] (c4) at (0.7,2.0) {\large $C_{1,4}$};
   \node[circle,minimum size=32pt,fill=red,opacity=0.4,text opacity=1] (c5) at (2.9,2.4) {\large $C_{2,3}$};
   \node[circle,minimum size=32pt,fill=red,opacity=0.4,text opacity=1] (c6) at (5.1,2.4) {\large $C_{3,2}$};
   \node[circle,minimum size=32pt,fill=red,opacity=0.4,text opacity=1] (c7) at (7.3,2.0) {\large $C_{4,1}$};
   \node[circle,minimum size=32pt,fill=red,opacity=0.4,text opacity=1] (c8) at (1.6,3.4) {\large $C_{2,4}$};
   \node[circle,minimum size=32pt,fill=red,opacity=0.4,text opacity=1] (c9) at (4,3.6) {\large $C_{3,3}$};
   \node[circle,minimum size=32pt,fill=red,opacity=0.4,text opacity=1] (c10) at (6.4,3.4) {\large $C_{4,2}$};
   \node[circle,minimum size=32pt,fill=red,opacity=0.4,text opacity=1] (c11) at (2.7,4.6) {\large $C_{3,4}$};
   \node[circle,minimum size=32pt,fill=red,opacity=0.4,text opacity=1] (c12) at (5.3,4.6) {\large $C_{4,3}$};
   \node[circle,minimum size=32pt,fill=red,opacity=0.4,text opacity=1] (c13) at (4,5.6) {\large $C_{4,4}$};
   
   \draw[thick] (w1) edge (c1);
   \draw[thick] (w1) edge (c4);
   \draw[thick] (w2) edge (c3);
   \draw[thick] (w2) edge (c7);
   \draw[thick] (c1) edge (c2);
   \draw[thick] (c2) edge (c3);
   \draw[thick] (c1) edge (c4);
   \draw[thick] (c1) edge (c5);
   \draw[thick] (c2) edge (c5);
   \draw[thick] (c2) edge (c6);
   \draw[thick] (c3) edge (c6);
   \draw[thick] (c3) edge (c7);
   \draw[thick] (c4) edge (c5);
   \draw[thick] (c5) edge (c6);
   \draw[thick] (c6) edge (c7);
   \draw[thick] (c4) edge (c8);
   \draw[thick] (c5) edge (c8);
   \draw[thick] (c5) edge (c9);
   \draw[thick] (c6) edge (c9);
   \draw[thick] (c6) edge (c10);
   \draw[thick] (c7) edge (c10);
   \draw[thick] (c8) edge (c9);
   \draw[thick] (c9) edge (c10);
   \draw[thick] (c8) edge (c11);
   \draw[thick] (c9) edge (c11);
   \draw[thick] (c9) edge (c12);
   \draw[thick] (c10) edge (c12);
   \draw[thick] (c11) edge (c12);
   \draw[thick] (c11) edge (c13);
   \draw[thick] (c12) edge (c13);
   
   \draw[dashed] (4,0) ellipse (4.4 and 0.3);

   \draw[dashed] (8.2,-1.2) arc[start angle=-20,end angle=200, x radius=4.47, y radius=5.6];
   \node at (4,-1) {\large $C'$};
   
  \end{tikzpicture}
  }
  \caption{Visualization of the sets $C_{i,j}$ for $d = 4$ in the proof of Lemma \ref{la:as-main}.
   Each arc between two sets indicates that there may be edges connecting vertices from the two sets.}
  \label{fig:distance-pyramid-app}
 \end{figure}
 
 \begin{claim}
  \label{cl:c-d-d}
  $C_{d,d} = C_{d-1,d} = C_{d,d-1} = \emptyset$.
 \end{claim}
 \begin{claimproof}
  We argue that $C_{d,d} = C_{d-1,d} = \emptyset$. By symmetry, we also obtain that $C_{d,d-1} = \emptyset$.
  
  Suppose, for the sake of a contradiction, that $C_{d,d} \cup C_{d-1,d} \neq \emptyset$ and pick $v \in C_{\ell,d}$ for some $\ell \in \{d,d-1\}$. 
  Let $v_0,\dots,v_\ell$ be a shortest path from $w_1 = v_0$ to $v = v_\ell$.
  We first argue by induction that $v_i \in C_{i,d}$.
  In the base case $i = \ell$, the statement holds by definition.
  So suppose $i < \ell$.
  First observe that $\dist(w_1,v_i) = i$ by definition.
  By the induction hypothesis, $v_{i+1} \in C_{i+1,d}$.
  This means $\dist(w_2,v_i) \geq d$ by Lemma \ref{la:as-distance-patterns}.
  Also, $\dist(w_2,v_i) \leq d$ by Lemma~\ref{la:as-small-side-distances}.
  So $v_i \in C_{i,d}$. In particular, $C_{1,d} \neq \emptyset$.
  
  Now let $v' \in C_{1,d}$ and consider an arbitrary vertex $u \neq v'$ such that $\dist(v',u) < \dist(w_1,u)$.
  Then $u \in C$ and $\dist(v',u) \leq d-1$ by Lemma \ref{la:as-small-side-distances}.
  Since $D(v') = D(w_1)$, we conclude that for every vertex $u \neq w_1$ with $\dist(v',u) > \dist(w_1,u)$, it holds that $\dist(w_1,u) \leq d-1$ (cf.\ Lemma \ref{la:distance-patterns-small}).
  Hence, there is no vertex $u \in C'$ such that $\dist(w_1,u) = d$.
  Indeed, every such vertex would have to satisfy $\dist(w_1,u) \geq \dist(v',u)$, i.e., every shortest path from $u$ to $v'$ would have to pass through $w_2$.
  However, since $\dist(v',w_2) = d$, such a path has length at least $d+1$.
  Since there is no cut vertex in $G$, this excludes also the existence of vertices $u \in C'$ such that $\dist(w_1,u) > d$.
  Hence, we obtain that $\dist(w_1,u) \leq d-1$ holds for every $u \in C'$.
  
  Overall, this means that, on the one hand, $\dist(w_1,u) \leq d$ for all $u \in V(G)$.
  On the other hand, there is a $u \in C'$ such that $\dist(w_1,u) \geq 2$ because $d \geq 3$.
  But then $\dist(v,u) \geq d+1$ holds for every $v \in C_{d,d} \cup C_{d-1,d}$.
  So $D(w_1) \neq D(v)$, which is a contradiction.
 \end{claimproof}
 
 \begin{claim}
  \label{cl:unique-shortest-path-app}
  Let $u,v \in V(G)$ such that $\dist(u,v) < d$. Then there is a unique shortest path from $u$ to $v$.
 \end{claim}
 \begin{claimproof}
  Suppose the statement does not hold and let $\ell < d$ be the minimal number for which the claim is violated.
  Let $u,v \in V(G)$ be two vertices with $\dist(u,v) = \ell$ such that there are two paths of length $\ell$ from $u$ to $v$.
  Also let $E' \coloneqq \{u'v' \mid \chi_G(u,v) = \chi_G(u',v')\}$ and consider the graph $G' \coloneqq (V(G),E(G) \cup E')$, i.e., the graph obtained from $G$ by inserting all (undirected) edges contained in the set $E'$.
  We argue that $G'$ still satisfies Assumption \ref{ass:wl-2sep}, which contradicts the edge maximality of $G$.
  
  First, the coloring $\chi_G$ is also a stable coloring of $G'$, which implies that $\chi_{G}$ refines the coloring $\chi_{G'}$.
  In particular, Condition \ref{item:ass-color} of Assumption~\ref{ass:wl-2sep} is satisfied for the graph $G'$.
  
  Now let $u'v' \in E'$.
  Then $\dist(u',v') = \ell$ and there are at least two different walks of length $\ell$ from $u'$ to $v'$ because the same statement holds for $u$ and $v$.
  Due to the minimality of $\ell$, the two walks are internally vertex-disjoint paths.
  If $u'$ and $v'$ lie in different connected components of $G - \{w_1,w_2\}$, then one of the two paths must pass through $w_1$ while the other one passes through $w_2$, forming a cycle of length $2\ell < 2d$.
  This implies $\dist(w_1,w_2) < d$, a contradiction.
  Thus, we conclude that there is a connected component with vertex set $C$ of $G - \{w_1,w_2\}$ such that $u',v' \in C \cup \{w_1,w_2\}$.
  But this means Condition \ref{item:ass-separator} of Assumption \ref{ass:wl-2sep} is satisfied for the graph $G'$.
 \end{claimproof}
 
 \begin{claim}
  \label{cl:shortest-path-detour-app}
  Let $u,v \in V(G)$ such that $\ell \coloneqq \dist(u,v) < d$. Furthermore, suppose there is a walk $u = u_0,\dots,u_{\ell+1} = v$ of length $\ell+1$ from $u$ to $v$.
  Then there is an $i \in [\ell]$ with~$u_{i-1}u_{i+1} \in E(G)$.
 \end{claim}
 \begin{claimproof}
  Suppose the statement does not hold and let $\ell < d$ be the minimal number for which the claim is violated.
  Let $u,v \in V(G)$ such that $\ell = \dist(u,v)$ and there is a walk $u = u_0,\dots,u_{\ell+1} = v$
  of length $\ell+1$ from $u$ to $v$ such that, for all $i \in [\ell]$, it holds that $u_{i-1}u_{i+1} \notin E(G)$.
  Also, let $E' \coloneqq \{u'v' \mid \chi_G(u,v) = \chi_G(u',v')\}$ and consider the graph $G' = (V(G),E(G) \cup E')$.
  Similarly to the previous claim, we argue that $G'$ still satisfies Assumption \ref{ass:wl-2sep}, which contradicts the edge maximality of $G$.
  
  Indeed, by the same argument as in the previous claim, Condition \ref{item:ass-color} of Assumption \ref{ass:wl-2sep} is satisfied for the graph $G'$.
  
  Now let $u'v' \in E'$.
  Then $\dist(u',v') = \ell$ and there is a walk $u' = u_0',\dots,u_{\ell+1}' = v'$
  of length $\ell+1$ from $u'$ to $v'$ such that, for all $i \in [\ell]$, it holds that $u_{i-1}'u_{i+1}' \notin E(G)$.
  Note that $2$-WL knows whether such a walk exists, since the shortest path is unique by Claim~\ref{cl:unique-shortest-path-app} and the algorithm knows the number of triangles that share an edge with the shortest path.
  Due to the minimality of $\ell$, it holds that the unique shortest path from $u'$ to $v'$ and $u_0',\dots,u_{\ell+1}'$ are internally vertex-disjoint.
  Since $\dist(w_1,w_2) = d$, this implies that there is a connected component of $G - \{w_1,w_2\}$ with vertex set $C$ such that $u',v' \in C \cup \{w_1,w_2\}$ (using the same arguments as before).
  But this again means Condition \ref{item:ass-separator} of Assumption \ref{ass:wl-2sep} is also satisfied for $G'$.
 \end{claimproof}
 
These claims drastically restrict the structure of the graph $G$ and will allow us to prove that $G$ is a cycle. Intuitively speaking, the claims imply that, for each of the $w_i$, the graph induced by $w_i$ and all vertices in $G[C]$ with distance at most $d-1$ to $w_i$ is sufficiently similar to a tree.

 For $k' \in \{d,\dots,2d\}$, let
 \[\mathcal{C}_{k'} \coloneqq \bigcup_{i,j \colon i+j = k'} C_{i,j}.\]
 Let $k \in \{d,\dots,2d\}$ be the maximal number such that $\mathcal{C}_k \neq \emptyset$.
 Note that $k \leq 2d-2$ by Claim \ref{cl:c-d-d}. First suppose that $k \geq d+1$.
 
 \begin{claim}
  \label{cl:two-neighbors}
  Let $v \in C_{i,k-i}$ for $k-d+1 \leq i \leq d-1$. Then $|N(v) \cap (C_{i-1,k-i+1} \cup C_{i-1,k-i})| = 1$ and also $|N(v) \cap (C_{i,k-i-1} \cup C_{i+1,k-i-1})| = 1$.
 \end{claim}
 \begin{claimproof}
  This follows directly from Claim \ref{cl:unique-shortest-path-app}.
 \end{claimproof}
 
 \begin{claim}
  \label{cl:tool-claim}
  Suppose $C_{i,k-i} \neq \emptyset$ holds for some $i$ with $k-d+1 \leq i \leq d-1$. Then $G$ is a cycle.
 \end{claim}
 \begin{claimproof}
  Let $i$ be as described in the claim and let $v \in C_{i,k-i}$.
  By Claim \ref{cl:two-neighbors} and the maximality of $k$, we obtain that $|N(v) \setminus C_{i,k-i}| = 2$.
  If $|N(v)| = 2$, then we are done (recall that $G$ is regular and connected).
  So suppose there is a $u \in N(v) \cap C_{i,k-i}$. Let $u'$ be a neighbor of $v$ of distance $i-1$ to $w_1$ (if $i=1$, then $u' = w_1$). Since via $v$ and $u'$, there is a path from $u$ to $w_1$ of length $i+1$, we can apply Claim \ref{cl:shortest-path-detour-app} to deduce that $u' \in N(u)$. Indeed, by the definition of the $C_{i,k-i}$, no other shortcuts on the path are possible. Swapping the roles of $u$ and $v$, we obtain that $u$ and $v$ agree on their neighbors of distance $i-1$ to $w_1$. By applying the same argument also for paths from $u$ and $v$ to $w_2$, we conclude that $N(u) \setminus C_{i,k-i} = N(v) \setminus C_{i,k-i}$.
  Now let $A$ be the vertex set of the connected component containing $v$ in the graph $G[C_{i,k-i}]$.
  Then $G[A]$ is a clique. Indeed, for every pair $u,u' \in A$, there are at least two paths of length~$2$ from $u$ to $u'$, since $N(u) \setminus C_{i,k-i} = N(u') \setminus C_{i,k-i}$.
  Thus, $uu' \in E(G)$ by Claim \ref{cl:unique-shortest-path-app}.
  But now $G[N[v]]$ contains at most one non-edge, namely between the vertices in $N(v) \setminus C_{i,k-i}$.
  So the same has to be true for $w_1$, which implies that $\deg(w_1) = 2$.
 \end{claimproof}

 Hence, we can assume that $C_{i,k-i} = \emptyset$ holds for all $i$ with $k-d+1 \leq i \leq d-1$.
 Since $\mathcal{C}_k \neq \emptyset$, this means that $C_{d,k-d} \neq \emptyset$ or $C_{k-d,d} \neq \emptyset$.
 Without loss of generality, assume that $C_{k-d,d} \neq \emptyset$. Note that $k-d < d$.
 
 Let $v \in C_{k-d,d}$. Let $w \in N(v)$ such that $\dist(w_1,w) = k-d-1$. Note that $w \in C_{k-d-1,d}$ (when $k = d+1$, then $w=w_1$).
 Also, observe that $|N(v) \cap C_{k-d-1,d}| = 1$ by Claim \ref{cl:unique-shortest-path-app}. 
 Let~$A$ be the vertex set of the connected component in $G[C_{k-d,d}]$ containing $v$.
 Let $u,u' \in A$ such that $uu' \in E(G)$.
 Then $N(u) \cap C_{k-d-1,d} = N(u') \cap C_{k-d-1,d}$ by Claim \ref{cl:shortest-path-detour-app}.
 So $N(u) \cap C_{k-d-1,d} = \{w\}$ for every $u \in A$.
 Also $G[A]$ forms a clique, since there is a unique shortest path between pairs of vertices at distance $2$ by Claim \ref{cl:unique-shortest-path-app}.
 So $G[A \cup \{w\}]$ forms a clique.
 Now let $u \in N(v) \setminus (A \cup \{w\})$.
 Then $u \in C_{k-d,d-1}$ (recall that $C_{k-d+1,d-1} = \emptyset$).
 First, $uw \in E(G)$ by Claim \ref{cl:shortest-path-detour-app}.
 Also, $A \subseteq N(u)$, since there is a unique shortest path between pairs of vertices at distance $2$ by Claim \ref{cl:unique-shortest-path-app}.
 Finally, let $u' \in N(v) \setminus (A \cup \{w\})$ such that $u \neq u'$. Then $\{v,w\} \subseteq N(u) \cap N(u')$ and hence, $uu' \in E(G)$, again using Claim~\ref{cl:unique-shortest-path-app}.
 So overall $N[v]$ forms a clique and thus, the same must hold for $N[w_1]$, which is a contradiction, since $w_1$ belongs to a separator.
 
 In the other case, $k=d$.
 Observe that $C_{i,d-i} \neq \emptyset$ holds for every $i$ with $1 \leq i \leq d-1$.
 For two disjoint sets $U,W \subseteq V(G)$, let $G[U,W] \coloneqq \big(U \cup W,E(G) \cap \big\{\{u,w\} \, \big\vert \, u \in U, w \in W\big\}\big)$ be the bipartite graph induced by $U$ and $W$.
 \begin{claim}
  \label{cl:matching-graph}
  For all $i,j$ with $1 \leq i,j \leq d-1$, it holds that $|C_{i,d-i}| = |C_{j,d-j}|$. Also, the graph $G[C_{i,d-i},C_{i+1,d-i-1}]$ is a matching graph (i.e., every vertex in the graph has degree~$1$) for $1 \leq i \leq d-1$.
 \end{claim}
 \begin{claimproof}
  This follows directly from Claim \ref{cl:unique-shortest-path-app}.
 \end{claimproof}
 Since $G$ is regular, we further conclude that for every $i$ with $1 \leq i \leq d-1$, the graph $G[C_{i,d-i}]$ is a clique.
 To see this, let $v \in C_{i,d-i}$ for an arbitrary $i$ with $1 \leq i \leq d-1$.
 By Claim \ref{cl:matching-graph} and the definition of the sets $C_{i,j}$, we know $|N(v) \setminus C_{i,d-i}| \leq 2$.
 But also $\deg(v) = \deg(w_1) \geq |C_{1,d}| + 1 = |C_{i,d-i}| + 1$, since $C_{1,d} \subseteq N(w_1)$ and $|N(w_1) \cap C'| \geq 1$ and~$G$ is regular.
 Still, this is only possible if $v$ is adjacent to every vertex in $C_{i,d-i}$ (except for $v$ itself).
 Since $v \in C_{i,d-i}$ was chosen arbitrarily, it follows that $C_{i,d-i}$ forms a clique. 

 Now suppose that $|C_{i,d-i}| \geq 2$ holds for some (and therefore every) $i$ with $1 \leq i \leq d-1$.
 Then $G$ contains an induced subgraph isomorphic to $C_4$, which contradicts Claim \ref{cl:unique-shortest-path-app} (recall that $d \geq 3$).
 So $|C_{i,d-i}| = 1$ holds for all $i$ with $1 \leq i \leq d-1$ and hence, $G$ is a cycle.
\end{proof}

Reformulating the previous lemma, we obtain the following theorem.

\begin{theorem}\label{thm:onecolor}
 Let $G$ be a graph such that $\chi_G(u,u) = \chi_G(v,v)$ holds for all $u,v \in V(G)$.
 Then (exactly) one of the following holds:
 \begin{enumerate}
  \item $G$ is not connected, or
  \item $G$ is $3$-connected, or
  \item $G$ is a cycle of length $\ell \geq 4$.
 \end{enumerate}
\end{theorem}

Note that the complete graphs on $2$ and $3$ vertices are $3$-connected.
The theorem also implies that a connected constituent graph of an association scheme is either $3$-connected or a cycle.
It thus provides a generalization of Kodalen's and Martin's result in \cite{KodalenM17}, where they proved the theorem in case the graph stems from a symmetric association scheme. (For other work on the connectivity of relations in association schemes, see e.g.\ \cite{Brouwer96,BrouwerK09,BrouwerM85,EvdokimovP06}.)

\section{Two Colors}\label{sec:twocolors}

Recall that our overall goal is to prove that $2$-WL assigns special colors to $2$-separators. We will use Lemma \ref{la:as-main} to prove this in case the tuples $(u,u)$ and $(v,v)$ of a $2$-separator $\{u,v\}$ obtain the same $\chi_2$-color. 
To treat the much more difficult case that $u$ and $v$ obtain distinct colors, we intend to generalize the results of the previous section to two vertex colors.
Maybe somewhat surprisingly, we obtain a statement that is similar to Lemma \ref{la:as-main}.
However, now we require the input graphs to be $2$-connected instead of only being connected. This is a necessary condition, since, for example, the star graphs $K_{1,n}$ for $n \geq 2$ are neither $3$-connected nor cycles but still, the image of $\chi_2$ only contains two vertex colors for them. 

We shall prove the following theorem. 

\begin{theorem}
 \label{thm:two-colors-main}
 Let $G$ be a $2$-connected graph with the following properties:
 \begin{enumerate}
  \item $G$ has a $2$-separator $\{w_1,w_2\}$, and
  \item for every $v \in V(G)$, there is an $i \in \{1,2\}$ such that $\chi_G(v,v) = \chi_G(w_i,w_i)$.
 \end{enumerate}
 Then $G$ is a cycle.
\end{theorem}

The route to proving the statement is similar to the one described in Section \ref{sec:onecolor}.
Still, with two colors allowing for more complexity in the graph structure, the statements and proofs become more involved and additional cases need to be considered.
To be more precise, using an induction on the order of $G$ and additional insights and tools regarding $2$-WL, the first step of the proof is to reduce the general case to one in which the partition into the two vertex color classes forms a bipartition of the graph (i.e., there are no edges between vertices of the same color).
Let $\{U,V\}$ be the partition of $V(G)$ according to vertex colors.
To achieve the reduction, we contract the connected components of $G[U]$ and $G[V]$ to single vertices.
Here, we crucially exploit the fact that $2$-WL ``has access'' to the factor graph and that thus, the factor graph still meets the requirements of the theorem.

To treat the case that $\{U,V\}$ forms a bipartition, the strategy is essentially the same as in the previous section.
We start by adapting several of the auxiliary lemmas given in the previous section to the setting of two vertex colors.

\begin{lemma}
 \label{la:distance-patterns-dist-two}
 Let $G$ be a connected graph with $n$ vertices and suppose $\{w_1,w_2\}$ is a $2$-separator of $G$.
 Let $C$ be the vertex set of a connected component of $G - \{w_1,w_2\}$ such that $|C| < \frac{n}{2}$
 and let $v \in C$.
 Then there is no $u \in V(G)$ such that
 \begin{enumerate}
  \item $\dist(u,v) \leq 2$,
  \item $\chi_G(u,u) = \chi_G(v,v)$, and
  \item $\dist(u,w_i) \leq \dist(v,w_i)-2$ for both $i \in \{1,2\}$.
 \end{enumerate}
\end{lemma}

\begin{proof}
 The proof is analogous to the one for Lemma \ref{la:as-distance-patterns}.
\end{proof}

\begin{lemma}
 \label{la:two-colors-small-side-distances}
 Let $G = (U,V,E)$ be a $2$-connected bipartite graph with $n$ vertices and the following properties:
 \begin{enumerate}
  \item $\chi_G(u,u) = \chi_G(u',u')$ for all $u,u' \in U$,
  \item $\chi_G(v,v) = \chi_G(v',v')$ for all $v,v' \in V$, and
  \item $G$ has a $2$-separator $\{w_1,w_2\}$ with $w_1 \in U$ and $w_2 \in V$.
 \end{enumerate}
 Let $d \coloneqq \dist(w_1,w_2)$ and let $C$ be the vertex set of a connected component of $G - \{w_1,w_2\}$ with $|C| \leq \frac{n-2}{2}$.
 Then $\dist(v,w_i) \leq d+1$ for all $v \in C$ and $i \in \{1,2\}$.
\end{lemma}

\begin{proof}
 Since $G$ is bipartite, $d$ is odd.
 By symmetry, it suffices to show the statement for $i=2$.
 For $u \in U \cap C$, we prove by induction on $\ell \coloneqq \dist(u,w_1)$ that $\dist(u,w_2) \leq d$. Then the statement follows, because every $v \in V \cap C$ is connected to some $u \in U \cap C$.
 
 For $\ell = 0$, it holds that $u = w_1$ and $\dist(w_1,w_2) = d$.
 So suppose the statement holds for all $u \in U \cap C$ such that $\dist(u,w_1) \leq \ell$, and pick $u' \in U \cap C$ with $\dist(u',w_1) = \ell+2$.
 Let $u \in U \cap C$ such that $\dist(u,w_1) \leq \ell$ and $\dist(u,u') \leq 2$ hold. Then we have $\dist(u',w_2) \leq \dist(u,w_2)+1 \leq d+1$ by Lemma \ref{la:distance-patterns-dist-two} and the induction hypothesis.
 But since~$G$ is bipartite and we know that $u' \in U$ and $w_2 \in V$, we have that $\dist(u',w_2)$ is odd and thus, $\dist(u',w_2) \leq d$.
\end{proof}

\begin{lemma}
 \label{la:two-colors-dist-1}
 Let $G = (U,V,E)$ be a $2$-connected bipartite graph with $n$ vertices and the following properties:
 \begin{enumerate}
  \item\label{item:two-colors-dist-1-1} $\chi_G(u,u) = \chi_G(u',u')$ for all $u,u' \in U$,
  \item $\chi_G(v,v) = \chi_G(v',v')$ for all $v,v' \in V$, and
  \item $G$ has a $2$-separator $\{w_1,w_2\}$.
 \end{enumerate}
 Then $w_1w_2 \notin E(G)$.
\end{lemma}

\begin{proof}
 Since $G$ is bipartite and by symmetry, we only need to consider the case that $w_1 \in U$ and $w_2 \in V$.
 Suppose towards a contradiction that $w_1w_2 \in E(G)$.
 Let $C$ be the vertex set of a connected component of $G - \{w_1,w_2\}$ such that $|C| \leq \frac{n-2}{2}$.
 For $i,j \geq 1$, let
 \[C_{i,j} \coloneqq \{v \in C \mid \dist(v,w_1) = i \text{ and } \dist(v,w_2) = j\}.\]
 By Lemma \ref{la:two-colors-small-side-distances}, we conclude that $C_{i,j} = \emptyset$ unless $(i,j) \in \{(1,1),(1,2),(2,1),(2,2)\}$.
 Furthermore, $C_{1,1} = C_{2,2} = \emptyset$, since $G$ is bipartite.

 Note that $C_{1,2} \cup \{w_2\} \subseteq N(w_1)$.
 Moreover, since $G$ is $2$-connected, $w_1$ must have a neighbor in $V(G) \setminus (C \cup \{w_1,w_2\})$.
 Thus, $\deg(w_1) \geq |C_{1,2}| + 2$. Let $v \in C_{2,1} \subseteq U$.
 Then all neighbors of~$v$ are contained in $C_{1,2} \cup \{w_2\}$, since $C_{1,1} = C_{2,2} = \emptyset$.
 Hence, $\deg(v) \leq |C_{1,2}| + 1 < \deg(w_1)$, which contradicts Condition \ref{item:two-colors-dist-1-1}.
\end{proof}

We also require that $2$-WL ``has access'' to certain factor graphs of $G$, which allows us to build an inductive argument in case $G$ is not bipartite with a bipartition into the color classes.
We present the necessary tools next.

Let $G$ be a graph and let $\mathcal{C} \subseteq \{\chi_G(v_1,v_2) \mid v_1, v_2 \in V(G), v_1 \neq v_2\}$.
We define the graph~$G[\mathcal{C}]$ with vertex set
\begin{align*}
 V(G[\mathcal{C}]) {}\coloneqq{} &\{v_1 \in V(G) \mid \exists v_2 \in V(G) \colon \chi_G(v_1,v_2) \in \mathcal{C}\} {}\cup{} \\
                            &\{ v_2 \in V(G) \mid \exists v_1 \in V(G) \colon \chi_G(v_1,v_2) \in \mathcal{C}\}
\end{align*}
and edge set
\[E(G[\mathcal{C}]) {}\coloneqq{} \{v_1v_2 \mid \chi_G(v_1,v_2) \in \mathcal{C}\}.\]
Moreover, let $A_1,\dots,A_\ell$ be the vertex sets of the connected components of $G[\mathcal{C}]$.
We define~$G/\mathcal{C}$ to be the graph obtained from contracting every set $A_i$ to a single vertex.
Formally,
\[V(G/\mathcal{C}) \coloneqq \{\{v\} \mid v \in V(G) \setminus V(G[\mathcal{C}])\} \cup \{A_1,\dots,A_\ell\}\]
and
\[E(G/\mathcal{C}) \coloneqq \{X_1X_2 \mid \exists v_1 \in X_1,v_2 \in X_2\colon v_1v_2 \in E(G)\}.\]

\begin{lemma}[see {\cite[Theorem 3.1.11]{ChenP19}}]
 \label{la:factor-graph-2-wl}
 Let $G$ be a graph and $\mathcal{C} \subseteq \{\chi_G(v_1,v_2) \mid v_1, v_2 \in V(G), v_1 \neq v_2\}$.
 For all $X_1,X_2 \in V(G/\mathcal{C})$, define
 \[(\chi_G/\mathcal{C})(X_1,X_2) \coloneqq \big\{\!\big\{\chi_G(v_1,v_2) \mid v_1 \in X_1, v_2 \in X_2\big\}\!\big\}.\]
 Then $\chi_G/\mathcal{C}$ is a stable coloring of the graph $G/\mathcal{C}$ with respect to $2$-WL.
 
 Also, for all $X_1,X_2,X_1',X_2' \in V(G/\mathcal{C})$, it holds that $(\chi_G/\mathcal{C})(X_1,X_2) = (\chi_G/\mathcal{C})(X_1',X_2')$ or $(\chi_G/\mathcal{C})(X_1,X_2) \cap (\chi_G/\mathcal{C})(X_1',X_2') = \emptyset$.
\end{lemma}

We can use the lemma to deduce the following corollary, which will turn out to be useful for the induction that we perform afterwards to prove Theorem \ref{thm:two-colors-main}.

\begin{corollary}
 \label{cor:wl-factor-graph}
 Let $G$ be a graph and suppose there are colors $c_1 \neq c_2$ such that $\{\chi_G(v,v) \mid v \in V(G)\} = \{c_1,c_2\}$.
 Let $U \coloneqq \{u \in V(G) \mid \chi_G(u,u) = c_1\}$ and $V \coloneqq \{v \in V(G) \mid \chi_G(v,v) = c_2\}$.
 Also let $U_1,\dots,U_k$ be the vertex sets of the connected components of $G[U]$ and let $V_1,\dots,V_\ell$ be the vertex sets of the connected components of $G[V]$.
 Let $G'$ be the graph with $V(G') = \{U_1,\dots,U_k,V_1,\dots,V_\ell\}$ and whose edge set consists of all $U_iV_j$ for which there are $u \in U_i$, $v \in V_j$ such that $uv \in E(G)$.

 Then $\chi_{G'}(U_i,U_i) = \chi_{G'}(U_j,U_j)$ for all $i,j \in [k]$, and $\chi_{G'}(V_i,V_i) = \chi_{G'}(V_j,V_j)$ for all~$i,j \in [\ell]$. 
\end{corollary}

\begin{proof}
 Let $\mathcal{C} \coloneqq \{\chi_G(v,w) \mid vw \in E(G), \chi_G(v,v) = \chi_G(w,w)\}$.
 Observe that $vw \in E(G)$ and $\chi_G(v,v) = \chi_G(w,w)$ holds for all $v,w \in V(G)$ with $\chi_G(v,w) \in \mathcal{C}$.
 In particular, $U_1,\dots,U_k,V_1,\dots,V_\ell$ are precisely the connected components of $G[\mathcal{C}]$.
 So $\chi_{G'}(U_i,U_i) = \chi_{G'}(U_j,U_j)$ for all $i,j \in [k]$, and $\chi_{G'}(V_i,V_i) = \chi_{G'}(V_j,V_j)$ for all $i,j \in [\ell]$ by Lemma \ref{la:factor-graph-2-wl}.
\end{proof}

With this, we are now ready to prove the main result of this section.

\begin{proof}[Proof of Theorem \ref{thm:two-colors-main}]
 By Lemma \ref{la:as-main}, we can assume without loss of generality that $\chi_G(w_1,w_1) \neq \chi_G(w_2,w_2)$.
 The statement is proved by induction on the graph order $n$.
 For~$n \leq 4$, a simple case analysis among the possible graphs $G$ yields the statement.
 
 So let $n \geq 5$. Again, it suffices to prove the statement for the case that $G$ is an $n$-vertex graph with a maximum edge set that satisfies the requirements of the theorem.  
 Let
 \[U \coloneqq \big\{u \in V(G) \ \big\vert \ \chi_G(u,u) = \chi_G(w_1,w_1)\big\}\]
 and
 \[V \coloneqq \big\{v \in V(G) \ \big\vert \ \chi_G(v,v) = \chi_G(w_2,w_2)\big\}.\]
 Let $U_1,\dots,U_k$ be the vertex sets of the connected components of $G[U]$ and let $V_1,\dots,V_\ell$ be the vertex sets of the connected components of $G[V]$.
 Without loss of generality, assume that $w_1 \in U_1$ and $w_2 \in V_1$.
 Let $C$ be the vertex set of a connected component of $G - \{w_1,w_2\}$ with $|C| \leq \frac{n-2}{2}$.
 Also, let $C' \coloneqq V(G) \setminus (C \cup \{w_1,w_2\})$.
 
 \begin{claim}
  \label{cl:only-two-components-in-c}
  Suppose that $C \subseteq U_1 \cup V_1$ or $C' \subseteq U_1 \cup V_1$. Then $G$ is a cycle.
 \end{claim}
 \begin{claimproof}
  It is not hard to see that, by reasons of connectivity and since $|U_i| = |U_{j}|$ and $|V_i| = |V_{j}|$ holds for all $i, j$, it suffices to show that $|U_1| \leq 2$ and $|V_1| \leq 2$.

  Suppose $C \subseteq U_1 \cup V_1$. Then $C \cap (U_1 \cup V_1) \neq \emptyset$.
  By symmetry, we may assume that there exists a vertex $u_1 \in U_1 \cap C$ with $u_1w_1 \in E(G)$.
  We first argue that $U_1 \cap C' = \emptyset$.
  Towards a contradiction, suppose that there exists $u'_1 \in U_1 \cap C'$.
  Then $w_1$ is a cut vertex in $G[U_1]$. Since $2$-WL distinguishes arcs within connected components from arcs between different connected components, by \cite[Corollary 7]{kieponschwei19}, all vertices in $U_1$ must be cut vertices, in addition to $G[U_1]$ being connected.
  It is easy to see that there is no connected graph in which all vertices are cut vertices (see also the comment following Lemma \ref{la:as-no-cut-vertex}).
  Analogously, we cannot have that both $V_1 \cap C \neq \emptyset$ and $V_1 \cap C' \neq \emptyset$ hold. 

  Thus, $U_1 \cap C' = \emptyset$. First assume that $V_1 \cap C = \emptyset$.
  In this case, we have $C \cup \{w_1\} = U_1$ and, by regularity, $uw_2 \in E(G)$ for all $u \in U_1$.
  In particular, every vertex in $U$ has exactly one neighbor in $V$. However, for $w_1$, this neighbor must be distinct from $w_2$ because it must lie in $C'$.
  If $|U_1| \geq 3$, this results in $\chi_G(w_1,w_1) \neq \chi_G(u_1, u_1)$, a contradiction.

  Now assume that $V_1 \cap C' = \emptyset$.
  Note that for all $j \neq 1$ and all $v \in V_j$, we have that $u_1v \notin E(G)$.
  Thus, $|\{j \mid \exists v \in V_j \text{s.t.} u_1v \in E(G)\}| = 1$ by the connectivity of $G$.
  Since $\chi_G(u_1,u_1) = \chi_G(u,u)$ holds for all $u \in U$, every vertex in $U$ has neighbors in exactly one of the sets $V_j$.
  Moreover, for all $u \in U_1 \setminus \{w_1\}$, this set is $V_1$. However, for $w_1 \in U_1$, it is some~$V_j \subseteq C'$. Indeed, since $U_1 \cap C' = \emptyset$, the vertex $w_1$ is the only candidate in $U \cap (C \cup \{w_1\})$ to be adjacent to a vertex in $C'$.

  Suppose there is another vertex $u_2 \in U_1$ with $w_1 \neq u_2 \neq u_1$. Then for every~$v_1 \in V_1$, the colors $\chi_G(u_1,v_1)$ and $\chi_G(u_2,v_1)$ encode that there exist vertices $v,v' \in V_1$ (possibly equal) such that $u_1v, u_2v' \in E(G)$.
  The existence of such a $v_1$ is thus encoded in $\chi_G(u_1,u_2)$.
  However, for $w_1$ and any other vertex in $U_1$, there is no equivalent vertex in $V$. Thus, $w_1$ is not incident to any edge of color $\chi_G(u_1,u_2)$.
  Hence, $\chi_G(u_1, u_1) \neq \chi_G(w_1,w_1)$, which contradicts the assumptions.
  Therefore $|U_1| \leq 2$. Analogously, it can be shown that~$|V_1| \leq 2$.

  The case $C' \subseteq U_1 \cup V_1$ follows by symmetry. Altogether, we deduce that $G$ has to be a cycle.
 \end{claimproof}
 
 So for the rest of the proof, we assume that $C \nsubseteq U_1 \cup V_1$ and $C' \nsubseteq U_1 \cup V_1$.
 Note that, since we have assumed $w_1 \in U_1$ and $w_2 \in V_1$, this implies that both $G[U]$ and $G[V]$ consist of at least two connected components each (again using the fact that $G[U]$ and $G[V]$ do not contain any cut vertices). 

 Let $G'$ be the graph with $V(G') = \{U_1,\dots,U_k,V_1,\dots,V_\ell\}$ and whose edge sets consists of all $U_iV_j$ for which there are $u \in U_i$ and $v \in V_j$ such that $uv \in E(G)$.
 We argue that $G'$ satisfies the requirements of the theorem.
 First note that $G' - \{U_1,V_1\}$ is not connected and hence, $\{U_1,V_1\}$ forms a $2$-separator of $G'$.
 Also, $\chi_{G'}(U_i,U_i) = \chi_{G'}(U_1,U_1)$ holds for all~$i \in [k]$, and $\chi_{G'}(V_j,V_j) = \chi_{G'}(V_1,V_1)$ for all $j \in [\ell]$ by Corollary \ref{cor:wl-factor-graph}.

 Clearly, $G'$ is connected because $G$ is.
 So it remains to argue that $G'$ is $2$-connected. 
 If not, then $G'$ contains a cut vertex.
 Without loss of generality, assume that $U_i$ for a certain~$i \in [k]$ is a cut vertex.
 Then, since $2$-WL recognizes cut vertices (\cite[Corollary 7]{kieponschwei19}) and, by Corollary \ref{cor:wl-factor-graph}, every vertex $U_i$ with $i \in [k]$ is a cut vertex of $G'$.
 Considering the cut tree of~$G'$ (i.e., the tree in which every cut vertex and every $2$-connected component forms a vertex), it is not hard to see that this implies that every $V_j$ for $j \in [\ell]$ has only one neighbor.
 But this is only possible if $C \subseteq V_1$ or $C' \subseteq V_1$, and we assumed the contraries of both cases.
 
 \medskip
 Suppose that $|V(G')| < |V(G)|$.
 Then, by the induction hypothesis, the graph $G'$ is a cycle.
 Thus, every vertex in $U$ is adjacent to vertices in exactly one or two connected components of $G[V]$. 

 We first consider the subcase that both $C \cap U_1 = \emptyset$ and $C \cap V_1 = \emptyset$ hold. 

 The vertex $w_1$ is adjacent to vertices in a connected component with vertex set $V_j \subseteq C$, whereas all vertices $u \in U_1$ with $u \neq w_1$ must be adjacent to the same connected component with vertex set $V_{j'} \subseteq (C' \cup \{w_2\})$. 
 In particular, we have $j \neq j'$. Assuming $|U_1| > 2$, similarly as in the proof of Claim \ref{cl:only-two-components-in-c}, we reach a contradiction considering $\chi_G(w_1,w_1)$. By symmetry, we cannot have $|V_1| > 2$ either. The subcase that $C' \cap U_1 = \emptyset$ and $C' \cap V_1 = \emptyset$ can be treated analogously.

 Thus, suppose without loss of generality that $C \cap U_1 \neq \emptyset$ and $C' \cap V_1 \neq \emptyset$. (The case that $C' \cap U_1 \neq \emptyset$ and $C \cap V_1 \neq \emptyset$ follows by symmetry.) 

 Since we cannot have that $C = U_1 \backslash \{w_1\}$, there must be a connected component with vertex set $V_j \subset C$ for a certain $j \neq 1$ such that for every $u_1 \in U_1 \backslash \{w_1\}$, we have $N(u_1) \cap V \subseteq V_j$. However, all neighbors of $w_1$ in $V$ are contained in a $V_{j'} \subseteq C' \cup \{w_2\}$. In particular, $j \neq j'$. Again, we reach a contradiction when assuming that $|U_1| > 2$ and, by symmetry, also for the assumption $|V_1| > 2$.

 Thus, if $|V(G')| < |V(G)|$, then $G$ is a cycle.
 \medskip

 Now assume $|V(G')| = |V(G)|$. This means that $G$ is a bipartite graph with bipartition $\{U,V\}$ and all $U_i$ and all $V_j$ are singletons.
 From Lemma \ref{la:two-colors-dist-1}, we know that $w_1w_2 \notin E(G)$. Let~$d \coloneqq \dist(w_1,w_2)$.
 Note that $d$ is odd and thus, $d \geq 3$.
 
 \begin{claim}
  \label{cl:unique-shortest-path-two-colors}
  Let $u,v \in V(G)$ such that $\dist(u,v) < d$. Then there is a unique shortest path from $u$ to $v$.
 \end{claim}
 \begin{claimproof}
  Suppose the statement does not hold and let $\ell < d$ be the minimal number for which the claim is violated.
  Let $u,v \in V(G)$ be two vertices with $\dist(u,v) = \ell$ such that there are two paths of length $\ell$ from $u$ to $v$.
  Also, let $E' \coloneqq \{u'v' \mid \chi_G(u,v) = \chi_G(u',v')\}$ and consider the graph $G' = (V(G),E(G) \cup E')$.
  We argue that $G'$ still satisfies the assumptions of the theorem, which contradicts the edge maximality of $G$.
  
  First, the coloring $\chi_G$ is also a stable coloring of $G'$, which implies that $\chi_{G}$ refines the coloring $\chi_{G'}$.
  In particular, for every $v \in V(G)$, there is an $i \in \{1,2\}$ such that $\chi_G(v,v) = \chi_G(w_i,w_i)$.
  
  Now let $u'v' \in E'$.
  Then $\dist(u',v') = \ell$ and there are at least two different walks of length~$\ell$ from $u'$ to $v'$, because the same statement holds for $u$ and $v$.
  Due to the minimality of $\ell$, the two walks are vertex-disjoint paths.
  If $u'$ and $v'$ lie in different connected components of $G - \{w_1,w_2\}$, then one of the two paths must pass through $w_1$ while the other one passes through $w_2$, forming a cycle of length $2\ell < 2d$.
  This implies $\dist(w_1,w_2) < d$, a contradiction.
  Thus, we conclude that there is a connected component $C$ of $G - \{w_1,w_2\}$ such that $u',v' \in C \cup \{w_1,w_2\}$.
  But this means that $\{w_1,w_2\}$ is also a $2$-separator of the graph $G'$.
 \end{claimproof}
 
 For $i,j \geq 1$, let
 \[C_{i,j} \coloneqq \{v \in C \mid \dist(v,w_1) = i \text{ and } \dist(v,w_2) = j\}.\]
 By Lemma \ref{la:two-colors-small-side-distances}, we conclude that $C_{i,j} = \emptyset$ unless $i,j \leq d+1$. Furthermore, by the definition of $d$, it holds that $C_{i,j} = \emptyset$ unless $i+j \geq d$.
 Since $G$ is bipartite, we know $C_{i,j} = \emptyset$ whenever~$i+j$ is even.
 
 \begin{claim}
  \label{cl:degree-seven-cycle}
  Suppose $d \geq 7$. Then $G$ is a cycle.
 \end{claim}
 \begin{claimproof}
  We argue that $\deg(u) = 2$ for every $u \in U$. Analogously, one can prove that $\deg(v) = 2$ for every $v \in V$, which together means that $G$ is a cycle.
  
  First suppose that $C_{4,d-4} \neq \emptyset$ and let $u \in C_{4,d-4}$. Note that $u \in U$, since $4$ is even.
  Then $N(u) \subseteq C_{3,d-3} \cup C_{5,d-5} \cup C_{5,d-3}$.
  Moreover, $|N(u) \cap (C_{3,d-3} \cup C_{5,d-5})| = 2$, otherwise there would be two shortest paths from $u$ to $w_1$ or two shortest paths from $u$ to $w_2$, contradicting Claim \ref{cl:unique-shortest-path-two-colors}.
  Now suppose towards a contradiction that $\deg(u) > 2$.
  Then there is a $v \in N(u) \cap C_{5,d-3}$.
  We have that $N(v) \subseteq C_{4,d-4} \cup C_{4,d-2} \cup C_{6,d-4} \cup C_{6,d-2}$.
  Using Claim~\ref{cl:unique-shortest-path-two-colors}, we conclude that $|N(v) \cap (C_{4,d-4} \cup C_{4,d-2})| \leq 1$ and $|N(v) \cap (C_{4,d-4} \cup C_{6,d-4})| \leq 1$.
  Since $\deg(v) \geq 2$ and $u \in N(v) \cap C_{4,d-4} \neq \emptyset$, it follows that $N(v) \cap C_{6,d-2} \neq \emptyset$.
  But this contradicts Lemma \ref{la:distance-patterns-dist-two}.
  So $\deg(u) = 2$.
  
  Now suppose $C_{4,d-4} = \emptyset$ and hence, $C_{i,d-i} = \emptyset$ holds for all $1 \leq i \leq d-1$. Since $N(w_1) \cap C = C_{1,d-1} \cup C_{1,d+1}$, we conclude that $C_{1,d+1} \neq \emptyset$ and hence, $C_{i,d+2-i} \neq \emptyset$ holds for all $1 \leq i \leq d+1$.
  
  So pick a vertex $u \in C_{4,d-2}$.
  Then $N(u) \subseteq C_{3,d-1} \cup C_{5,d-3} \cup C_{5,d-1}$.
  Again, it holds that~$|N(u) \cap (C_{3,d-1} \cup C_{5,d-3})| = 2$, using Claim \ref{cl:unique-shortest-path-two-colors}.
  Suppose towards a contradiction that $\deg(u) > 2$.
  Then there is a vertex $v \in N(u) \cap C_{5,d-1}$.
  We have that $N(v) \subseteq C_{4,d-2} \cup C_{4,d} \cup C_{6,d-2} \cup C_{6,d}$.
  By Claim~\ref{cl:unique-shortest-path-two-colors}, we conclude that $|N(v) \cap (C_{4,d-2} \cup C_{4,d})| \leq 1$ and $|N(v) \cap (C_{4,d-2} \cup C_{6,d-2})| \leq 1$.
  Since $\deg(v) \geq 2$, it follows that $N(v) \cap C_{6,d} \neq \emptyset$.
  As before, this contradicts Lemma \ref{la:distance-patterns-dist-two}.
 \end{claimproof}
 
 \begin{figure}
  \centering
  \scalebox{0.85}{
  \begin{tikzpicture}
   \node[tinyvertex,fill=red,label=below:{$w_1$}] (w1) at (0,0) {};
   \node[tinyvertex,fill=blue,label=below:{$w_2$}] (w2) at (10,0) {};
   
   \node[circle,minimum size=32pt,fill=blue,opacity=0.4,text opacity=1] (c1) at (1.6,0.8) {\large $C_{1,4}$};
   \node[circle,minimum size=32pt,fill=red,opacity=0.4,text opacity=1] (c2) at (3.8,1.2) {\large $C_{2,3}$};
   \node[circle,minimum size=32pt,fill=blue,opacity=0.4,text opacity=1] (c3) at (6.2,1.2) {\large $C_{3,2}$};
   \node[circle,minimum size=32pt,fill=red,opacity=0.4,text opacity=1] (c4) at (8.4,0.8) {\large $C_{4,1}$};
   \node[circle,minimum size=32pt,fill=blue,opacity=0.4,text opacity=1] (c5) at (0.5,2.2) {\large $C_{1,6}$};
   \node[circle,minimum size=32pt,fill=red,opacity=0.4,text opacity=1] (c6) at (2.3,2.6) {\large $C_{2,5}$};
   \node[circle,minimum size=32pt,fill=blue,opacity=0.4,text opacity=1] (c7) at (4.1,3.0) {\large $C_{3,4}$};
   \node[circle,minimum size=32pt,fill=red,opacity=0.4,text opacity=1] (c8) at (5.9,3.0) {\large $C_{4,3}$};
   \node[circle,minimum size=32pt,fill=blue,opacity=0.4,text opacity=1] (c9) at (7.7,2.6) {\large $C_{5,2}$};
   \node[circle,minimum size=32pt,fill=red,opacity=0.4,text opacity=1] (c10) at (9.5,2.2) {\large $C_{6,1}$};
   \node[circle,minimum size=32pt,fill=blue,opacity=0.4,text opacity=1] (c11) at (2.4,4.4) {\large $C_{3,6}$};
   \node[circle,minimum size=32pt,fill=red,opacity=0.4,text opacity=1] (c12) at (4.2,4.8) {\large $C_{4,5}$};
   \node[circle,minimum size=32pt,fill=blue,opacity=0.4,text opacity=1] (c13) at (5.8,4.8) {\large $C_{5,4}$};
   \node[circle,minimum size=32pt,fill=red,opacity=0.4,text opacity=1] (c14) at (7.6,4.4) {\large $C_{6,3}$};
   \node[circle,minimum size=32pt,fill=blue,opacity=0.4,text opacity=1] (c15) at (4.2,6.4) {\large $C_{5,6}$};
   \node[circle,minimum size=32pt,fill=red,opacity=0.4,text opacity=1] (c16) at (5.8,6.4) {\large $C_{6,5}$};
   
   \draw[thick] (w1) edge (c1);
   \draw[thick] (w1) edge (c5);
   \draw[thick] (w2) edge (c4);
   \draw[thick] (w2) edge (c10);
   \draw[thick] (c1) edge (c2);
   \draw[thick] (c2) edge (c3);
   \draw[thick] (c3) edge (c4);
   \draw[thick] (c1) edge (c6);
   \draw[thick] (c2) edge (c7);
   \draw[thick] (c3) edge (c8);
   \draw[thick] (c4) edge (c9);
   \draw[thick] (c5) edge (c6);
   \draw[thick] (c6) edge (c7);
   \draw[thick] (c7) edge (c8);
   \draw[thick] (c8) edge (c9);
   \draw[thick] (c9) edge (c10);
   \draw[thick] (c6) edge (c11);
   \draw[thick] (c7) edge (c12);
   \draw[thick] (c8) edge (c13);
   \draw[thick] (c9) edge (c14);
   \draw[thick] (c11) edge (c12);
   \draw[thick] (c12) edge (c13);
   \draw[thick] (c13) edge (c14);
   \draw[thick] (c12) edge (c15);
   \draw[thick] (c13) edge (c16);
   \draw[thick] (c15) edge (c16);
   
   \draw[dashed] (5,0) ellipse (5.4 and 0.3);

   \draw[dashed] (10.2,-1.2) arc[start angle=-20,end angle=200, x radius=5.52, y radius=6.4];
   \node at (5,-1) {\large $C'$};
   
  \end{tikzpicture}
  }
  \caption{Visualization of the set $C_{i,j}$ for $d = 5$ in the proof of Theorem \ref{thm:two-colors-main}.
   There is an edge between two sets whenever there may be edges connecting vertices from the two sets.}
  \label{fig:distance-pyramid-two-colors}
 \end{figure}

 \begin{claim}
  Suppose $d = 5$. Then $G$ is a cycle.
 \end{claim}
 \begin{claimproof}
  Again, it suffices to show $\deg(u) = 2$ for every $u \in U$, since showing $\deg(v) = 2$ for every $v \in V$ works analogously. If $C_{2,3} \neq \emptyset$, we can argue similarly as in the proof of Claim~\ref{cl:degree-seven-cycle} that there are a vertex $u \in C_{2,3}$ and a vertex $v \in N(u) \cap C_{3,4}$ with $N(v) \cap C_{4,5} \neq \emptyset$, which contradicts Lemma \ref{la:distance-patterns-dist-two}.

  Thus, we can assume that $C_{2,3} = \emptyset$.
  Hence, we have $C_{i,d-i} = \emptyset$ for every $1 \leq i \leq 4$ and therefore $C_{1,6} \neq \emptyset$ and $C_{i,7-i} \neq \emptyset$ for every $1 \leq i \leq 6$.
  Analogously as in the proof of Claim \ref{cl:degree-seven-cycle}, there are vertices $u \in C_{4,3}$ and $v \in N(u) \cap C_{5,4}$, supposing $\deg(u) > 2$.
  We have $N(v) \subseteq C_{4,3} \cup C_{4,5} \cup C_{6,3} \cup C_{6,5}$. By Claim \ref{cl:unique-shortest-path-two-colors}, we have $|N(v) \cap (C_{4,3} \cup C_{6,3})| \leq 1$.
  We would like to apply the same argument to $N(v) \cap (C_{4,3} \cup C_{4,5})$.
  However, since $\dist(v,w_1) = 5 = d$, we cannot apply Claim \ref{cl:unique-shortest-path-two-colors} immediately.
  Still, note that every shortest path from $w_1$ to~$w_2$ with all internal vertices in $C$ has length $7$.
  Therefore, analogously as in the proof of Claim~\ref{cl:unique-shortest-path-two-colors}, using the edge maximality of $G$, we can show that there cannot be cycles of length at most~$10$ in $G$ and that, thus, there must be a unique shortest path from $v$ to $w_1$.
  In particular, this implies that $|N(v) \cap (C_{4,3} \cup C_{4,5})| \leq 1$.
  We can conclude the proof analogously as the one for Claim \ref{cl:degree-seven-cycle}.
 \end{claimproof}
 
 \begin{figure}
  \centering
  \scalebox{0.85}{
  \begin{tikzpicture}
   \node[tinyvertex,fill=red,label=below:{$w_1$}] (w1) at (0,0) {};
   \node[tinyvertex,fill=blue,label=below:{$w_2$}] (w2) at (8,0) {};
   
   \node[circle,minimum size=32pt,fill=blue,opacity=0.4,text opacity=1] (c1) at (2.5,1.2) {\large $C_{1,2}$};
   \node[circle,minimum size=32pt,fill=red,opacity=0.4,text opacity=1] (c2) at (5.5,1.2) {\large $C_{2,1}$};
   \node[circle,minimum size=32pt,fill=blue,opacity=0.4,text opacity=1] (c3) at (1.0,2.2) {\large $C_{1,4}$};
   \node[circle,minimum size=32pt,fill=red,opacity=0.4,text opacity=1] (c4) at (3.0,2.8) {\large $C_{2,3}$};
   \node[circle,minimum size=32pt,fill=blue,opacity=0.4,text opacity=1] (c5) at (5.0,2.8) {\large $C_{3,2}$};
   \node[circle,minimum size=32pt,fill=red,opacity=0.4,text opacity=1] (c6) at (7.0,2.2) {\large $C_{4,1}$};
   \node[circle,minimum size=32pt,fill=blue,opacity=0.4,text opacity=1] (c7) at (3.0,4.4) {\large $C_{3,4}$};
   \node[circle,minimum size=32pt,fill=red,opacity=0.4,text opacity=1] (c8) at (5.0,4.4) {\large $C_{4,3}$};
   
   \draw[thick] (w1) edge (c1);
   \draw[thick] (w1) edge (c3);
   \draw[thick] (w2) edge (c2);
   \draw[thick] (w2) edge (c6);
   \draw[thick] (c1) edge (c2);
   \draw[thick] (c1) edge (c4);
   \draw[thick] (c2) edge (c5);
   \draw[thick] (c3) edge (c4);
   \draw[thick] (c4) edge (c5);
   \draw[thick] (c5) edge (c6);
   \draw[thick] (c4) edge (c7);
   \draw[thick] (c5) edge (c8);
   \draw[thick] (c7) edge (c8);
   
   \draw[dashed] (4,0) ellipse (4.4 and 0.3);

   \draw[dashed] (8.2,-1.2) arc[start angle=-20,end angle=200, x radius=4.45, y radius=5.2];
   \node at (4,-1) {\large $C'$};
   
  \end{tikzpicture}
  }
  \caption{Visualization of the set $C_{i,j}$ for $d = 3$ in the proof of Theorem \ref{thm:two-colors-main}.
   There is an edge between two sets whenever there may be edges connecting vertices from the two sets.}
  \label{fig:distance-pyramid-two-colors-d-3}
 \end{figure}
 
 \begin{claim}
  Suppose $d = 3$. Then $G$ is a cycle.
 \end{claim}
 \begin{claimproof}
  First observe that all vertices of $C$ must be contained in one of the sets $C_{i,j}$ with $(i,j) \in \{(1,2),(2,1),(1,4),(2,3),(3,2),(4,1),(3,4),(4,3)\}$.

  If $w_2$ has only one neighbor $v$ in $C'$, then $\{w_1,v\}$ is a unicolored $2$-separator.
  Consider the graph $\widehat{G}$ with vertex set $U$ and an edge between every pair of vertices $u, v$ if there is a path from $u$ to $v$ of length $2$ (i.e., via a vertex in $V$) in $G$.
  Since $2$-WL knows whether such paths exist, we can apply Theorem \ref{thm:onecolor}.
  It is easy to see that $\widehat{G}$ is connected and not $3$-connected.
  Thus, it must be a cycle. In particular, every vertex in $\widehat G$ has degree~$2$.
  Therefore, the position of each vertex in $V$ is uniquely determined: the graph $G$ is a subdivision of $\widehat G$. Thus, $G$ is a cycle.

  Now assume 
  \begin{equation}\label{ass:two-neighbors}
    \deg(w_2) > |N(w_2) \cap C'| \geq 2.
  \end{equation}

  First suppose $C_{1,2} \neq \emptyset$ and, equivalently, $C_{2,1} \neq \emptyset$.
  For every vertex $u \in C_{2,1}$, we have $N(u) \subseteq C_{1,2} \cup C_{3,2} \cup \{w_2\}$.
  By Claim \ref{cl:unique-shortest-path-two-colors}, it holds that $|N(u) \cap C_{1,2}| = 1$.
  Thus, since $uw_2 \in E(G)$, we have $|N(u) \cap C_{3,2}| = r_U - 2$, where $r_U = \deg(u) = \deg(w_1)$.
  Furthermore, again by Claim \ref{cl:unique-shortest-path-two-colors}, for $u' \in C_{2,1}$ with $u' \neq u$, we have $N(u) \cap N(u') = \{w_2\}$ and thus, $|N(C_{2,1}) \cap C_{3,2}| = |C_{2,1}| \cdot (r_U - 2)$.
  Furthermore, by a similar reasoning, we have that $|N(C_{4,1}) \cap C_{3,2}| = |C_{4,1}| \cdot (r_U - 1)$. Every vertex in $C_{3,2}$ has distance $2$ to $w_2$ and thus has a neighbor in $C_{2,1} \cup C_{4,1}$.
  Note that $N(w_2) \cap C = C_{2,1} \cup C_{4,1}$.
  Now since, by Claim \ref{cl:unique-shortest-path-two-colors}, we have $\left(N(C_{2,1}) \cap N(C_{4,1})\right) \setminus \{w_2\} = \emptyset$, we obtain that $|C_{3,2}| = |C_{2,1}| \cdot (r_U - 2) + |C_{4,1}| \cdot (r_U - 1) < (|C_{2,1}| + |C_{4,1}|) \cdot (r_U - 1) \leq (r_V - 2) \cdot (r_U - 1)$ by \eqref{ass:two-neighbors}, where $r_V = \deg(w_2)$.

  Let $v \in C_{1,2}$. Similarly as above, $|N(v) \cap C_{2,3}| = r_V - 2$.
  Furthermore, for every vertex $u \in N(v) \cap C_{2,3}$, by Claim \ref{cl:unique-shortest-path-two-colors} and Lemma \ref{la:distance-patterns-dist-two}, we have that $N(u) \subseteq C_{1,2} \cup C_{3,2}$ and $|N(u) \cap C_{1,2}| = 1$.
  Again using Claim \ref{cl:unique-shortest-path-two-colors}, every two vertices in $N(v) \cap C_{2,3}$ must have disjoint neighborhoods in $C_{3,2}$.
  Thus, we obtain that $|C_{3,2}| \geq (r_V - 2)(r_U - 1)$.

  Altogether, $(r_V - 2)(r_U - 1) \leq |C_{3,2}| < (r_V - 2)(r_U - 1)$, which yields a contradiction.
  Therefore, $C_{1,2} = C_{2,1} = \emptyset$. Hence, it must hold that $C_{1,4} \neq \emptyset$ and consequently, also $C_{i,5-i} \neq \emptyset$ for every $i \in [4]$.
  Using the same arguments as before, we have $|C_{3,2}| \leq (r_V-2)(r_U-1)$.
  Every vertex $v' \in C_{3,2}$ has at least one neighbor in the set $C_{4,1}$ and at least one neighbor in $C_{2,3}$.
  Since for every $u' \in C_{4,3}$, there is a $v' \in C_{3,2} \cap N(u')$, this implies that $|C_{4,3}| \leq (r_V-2)(r_U-1)(r_V-2)$.

  Let $v \in C_{1,4}$ (recall that $C_{1,4} \neq \emptyset$).
  Then $|N(v) \cap C_{2,3}| = r_V -1$ and, for every vertex $u' \in N(v) \cap C_{2,3}$, we have by Claim \ref{cl:unique-shortest-path-two-colors} that $|N(u') \cap (C_{3,4} \cup C_{3,2})| = r_U-1$.
  Again by Claim \ref{cl:unique-shortest-path-two-colors}, for every such $u'$ and every $u'' \in N(v) \cap C_{2,3}$ with $u'' \neq u'$, it must hold that 
  \begin{equation}\label{eq:lower-bound-on-c43}
    N(u') \cap (C_{3,4} \cup C_{3,2}) \cap N(u'') = \emptyset.
  \end{equation}
  Let $S \coloneqq N(N(v)) \cap (C_{3,4} \cup C_{3,2})$.
  Note that every shortest path from $w_1$ to $w_2$ with all internal vertices in $C$ has length $5$.
  Therefore, using the edge maximality of $G$, analogously as in the proof of Claim \ref{cl:unique-shortest-path-two-colors}, there cannot be any cycles of length $6$ in $G$.
  In particular, this implies that for every pair $v', v'' \in S$ with $v'\neq v''$ we have $N(v') \cap N(v'') \cap C_{4,3} = \emptyset$.

  Thus, Equation \eqref{eq:lower-bound-on-c43} implies that 
  $|C_{4,3}| \geq (r_V-1)(r_U-1)(r_V-2)$, since every vertex in $S$ has at most two neighbors not contained in $C_{4,3}$.

  Altogether, we obtain $(r_V-1)(r_U-1)(r_V-2) \leq |C_{4,3}| \leq (r_V-2)(r_U-1)(r_V-2)$, which is a contradiction.
  This concludes the proof.
 \end{claimproof}
 The claim completes the proof of the theorem.
\end{proof}

\section{Detecting Decompositions with Weisfeiler and Leman}\label{sec:separators}

In this section, we show that $2$-WL implicitly computes the decomposition of a graph into its $3$-connected components.
As a main step, we first complete the proof that $2$-WL distinguishes $2$-separators from other pairs of vertices, building on the results from Sections \ref{sec:onecolor} and \ref{sec:twocolors}.
Afterwards, we derive several consequences of this result by falling back on ideas developed in \cite{kieponschwei19}.

Let $S$ be a set of vertex colors.
We say a path $u_0,\dots,u_\ell$ \emph{avoids $S$} if $\chi_G(u_i,u_i) \notin S$ holds for all $i \in [\ell-1]$. Note that we impose no restriction on the colors of the endpoints of the path.
It is easy to see that, given two vertices $u,v \in V(G)$, the algorithm $2$-WL knows if there is a path from $u$ to $v$ that avoids $S$.

We write $G|_S$ for the graph $(V,E)$ with $V = \{v \in V(G) \mid \chi_G(v,v) \in S\}$ and
\[E = \{uv \mid \text{there is a path from $u$ to $v$ in $G$ that avoids $S$}\}. \]

\begin{lemma}
 \label{la:wl-color-subset}
 Let $G$ be a graph and let $S \subseteq \{\chi_G(v,v) \mid v \in V(G)\}$. 
 Then $\chi_{G|_S}(u,u) = \chi_{G|_S}(v,v)$ for all $u,v \in V$ with $\chi_G(u,u) = \chi_G(v,v)$.
\end{lemma}

\begin{proof}
 Since $2$-WL knows whether there is a path between two vertices that avoids $S$, there is a set of colors $T$ such that $uv \in E$ if and only if $\chi_G(u,v) \in T$.
 Hence, every refinement performed by $2$-WL in $G|_S$ will also be done in $G$.
\end{proof}

\begin{theorem}\label{thm:separators}
 Let $G$ and $H$ be two graphs such that $G$ is $2$-connected.
 Let $w_1,w_2 \in V(G)$ such that $\{w_1,w_2\}$ forms a $2$-separator in $G$.
 Also let $v_1,v_2 \in V(H)$ and suppose $\chi_G(w_1,w_2) = \chi_H(v_1,v_2)$.
 Then $\{v_1,v_2\}$ forms a $2$-separator in $H$.
\end{theorem}

\begin{proof}
 Since the multisets of colors that the $2$-WL algorithm computes in two graphs are either equal or disjoint (see, e.g., \cite[Lemma 3.4]{kiefer:phd}), the assumptions of the theorem imply that $2$-WL does not distinguish between $G$ and $H$.
 Let $S \coloneqq \{\chi_G(w_1,w_1),\chi_G(w_2,w_2)\}$ and let $G' \coloneqq G|_S$.
 Clearly, the graph $G'$ is connected.
 We argue that $G'$ is $2$-connected.
 Suppose towards a contradiction that there is a cut vertex $w$ in $G'$.
 Let $C$ and $C'$ be the vertex sets of two connected components of $G' - \{w\}$.
 Let $v \in C$ and $v' \in C'$.
 Suppose there is a path~$P$ from $v$ to $v'$ in $G$ that does not pass $w$.
 Then there is a corresponding path $P'$ in $G'$, which simply skips all inner vertices of $P$ whose $\chi_2$-color is not contained in $S$.
 In particular,~$P'$ connects $v$ and $v'$, but avoids $w$. This contradicts $w$ being a cut vertex in $G'$. Hence, $G'$ is $2$-connected.
 
 First suppose that $|V(G')| = 2$. That is, $\{w_1,w_2\} = \{w \in V(G) \mid \chi_G(w,w) \in S\}$. Thus, $2$-WL knows that $G - \{w_1,w_2\}$ is disconnected. Let $A \coloneqq \{v \in V(H) \mid \chi_H(v,v) \in S\}$.
 Then $|A| = 2$ and thus $A = \{v_1, v_2\}$.
 So $H - A$ is disconnected because $G - \{w_1,w_2\}$ is disconnected.
 Hence, $\{v_1,v_2\}$ forms a $2$-separator in $H$.
 
 Now assume $|V(G')| \geq 3$ and suppose there is a vertex set $C$ of a connected component of $G - \{w_1,w_2\}$ such that $V(G') \subseteq C \cup \{w_1,w_2\}$.
 Let $C'$ be the vertex set of a second connected component of $G - \{w_1,w_2\}$ and let $v \in C'$.
 Then $w_1$ and $w_2$ are the only vertices with color in~$S$ that can be reached from $v$ via a path that avoids $S$.
 Since $\chi_G(w_1,w_2) = \chi_H(v_1,v_2)$, there is a $u \in V(H)$ such that $\chi_G(w_1,v) = \chi_H(v_1,u)$ and $\chi_G(w_2,v) = \chi_H(v_2,u)$.
 Because $2$-WL knows whether there is a path between two vertices that avoids $S$, it follows that $v_1$ and $v_2$ can be reached from $u$ via a path that avoids $S$.
 Moreover, it follows that $\chi_G(v,v) = \chi_H(u,u)$.
 Using again the fact that $2$-WL knows whether there is a path between two vertices that avoids $S$, this implies that there are exactly two vertices with color in $S$ that can be reached from $u$ via a path that avoids $S$.
 Together, this means that $v_1$ and $v_2$ are the only vertices with color in $S$ that can be reached from $u$ via a path that avoids $S$.
 Since $|V(G')| \geq 3$, there is a $u' \in V(H)$ such that $v_1 \neq u' \neq v_2$ and $\chi_H(u',u') \in S$, because, in order not to be distinguished, unions of color classes with colors in $S$ must have the same cardinality in both graphs.
 But then $\{v_1,v_2\}$ separates $u$ from $u'$ in $H$ and thus, $\{v_1,v_2\}$ forms a $2$-separator in~$H$.
 
 In the other case, $\{w_1,w_2\}$ forms a $2$-separator in $G'$.
 Hence, $G'$ is a cycle by Lemma~\ref{la:wl-color-subset} and Theorem \ref{thm:two-colors-main}.
 Note that $|V(G')| \geq 4$ and $w_1w_2 \notin E(G')$.
 It follows that $H|_S$ is also a cycle, since otherwise, $2$-WL would distinguish the graphs.
 Also, $|V(H|_S)| \geq 4$ and $v_1v_2 \notin E(H|_S)$.
 So $\{v_1,v_2\}$ forms a $2$-separator in $H|_S$. With analogous arguments as given for cut vertices above, every separator in $H|_S$ is also a separator of equal size in $H$.
 Thus, $\{v_1,v_2\}$ forms a $2$-separator in $H$.
\end{proof}

\begin{corollary}
 \label{cor:separators}
 Suppose $k \geq 2$.
 Let $G$ and $H$ be connected graphs. Assume $\{w_1,\dots,w_k\} \subseteq V(G)$ is a $k$-separator in $G$.
 Let $\{v_1,\dots,v_k\} \subseteq V(H)$ and suppose $\chi_{G,k}(w_1,\dots,w_k) = \chi_{H,k}(v_1,\dots,v_k)$.
 Then $\{v_1,\dots,v_k\}$ forms a $k$-separator in $H$.
\end{corollary}

\begin{proof}
 First suppose $k = 2$. If $G$ and $H$ are $2$-connected, the statement follows from Theorem~\ref{thm:separators}.
 If either $G$ or $H$ is not $2$-connected, then that graph contains a cut vertex, while the other graph does not.
 By \cite[Corollary 7]{kieponschwei19}, the presence of the cut vertex is encoded in every vertex color and thus, the multisets $\{\!\{\chi_G(u,v) \mid u,v \in V(G)\}\!\}$ and $\{\!\{\chi_H(u,v) \mid u,v \in V(H)\}\!\}$ are disjoint.
 Therefore, the statement trivially holds.

 Suppose both $G$ and $H$ are not $2$-connected.
 The statement is obviously true if $w_1$ or $w_2$ is a cut vertex in $G$.
 If this is not the case, then $w_1$ and $w_2$ must lie in a common $2$-connected component of $G$, otherwise they form no $2$-separator.
 Let $C_G$ denote the vertex set of this component.
 By \cite[Theorem 6]{kieponschwei19}, the same must hold for $v_1$ and $v_2$ in $H$, i.e., there is a vertex set $C_H$ of a $2$-connected component of $H$ such that $v_1,v_2 \in C_H$.
 Furthermore, again by \cite[Theorem 6]{kieponschwei19}, the algorithm $2$-WL distinguishes arcs from vertices $w \in C_G$ and from~$v \in C_H$ to vertices in the same $2$-connected component from arcs to other vertices.
 We claim that this implies $\chi_{G[C_G]}(w_1,w_2) = \chi_{H[C_H]}(v_1,v_2)$.
 Suppose towards a contradiction this is not the case.
 Then Spoiler wins the bijective $k$-pebble game $\BP_3(G[C_G],H[C_H])$ from position $((w_1,w_2),(v_1,v_2))$.
 This strategy can be turned into a winning strategy in the game $\BP_3(G,H)$ from position $((w_1,w_2),(v_1,v_2))$ by maintaining the property that all pebbles are placed on $C_G$ and $C_H$.
 Indeed, if Duplicator chooses a bijection $f\colon V(G) \rightarrow V(H)$ such that $f(C_G) \neq C_H$, then Spoiler can place pebbles on $w' \in V(G)$ and $v' \in V(H)$ such that~$w' \in C_G$ and $v' \notin C_H$ (or the other way around).
 Also, without loss of generality, the current position is not empty, i.e., there are pebbles already placed on some $w \in C_G$ and $v \in C_H$.
 But now $\chi_G(w,w') \neq \chi_H(v,v')$ by the above comment.
 So Spoiler wins from such a position by Theorem \ref{thm:eq-wl-pebble-tuples} which is a contradiction.
 Overall, this means $\chi_{G[C_G]}(w_1,w_2) = \chi_{H[C_H]}(v_1,v_2)$.
 Hence, applying Theorem \ref{thm:separators}, we deduce that $\{v_1,v_2\}$ forms a $2$-separator in $H[C_H]$ and therefore also in $H$.

 Finally, consider the case $k > 2$.
 Let $G$ and $H$ be graphs and suppose $\{w_1,\dots,w_k\} \subseteq V(G)$ is a $k$-separator in $G$.
 Also let $\{v_1,\dots,v_k\} \subseteq V(H)$. Furthermore, suppose that $\chi_{G,k}(w_1,\dots,w_k) = \chi_{H,k}(v_1,\dots,v_k)$.
 By Theorem \ref{thm:eq-wl-pebble-tuples}, Duplicator wins $\BP_{k+1}(G,H)$ from the initial position $\big((w_1,\dots,w_k),(v_1,\dots,v_k)\big)$.
 Now, let $G' \coloneqq G - \{w_1, \dots, w_{k-2}\}$ and $H' \coloneqq H - \{v_1, \dots, v_{k-2}\}$.
 We claim that $\chi_{G'}(w_{k-1},w_k) = \chi_{H'}(v_{k-1},v_k)$.
 Indeed, if this is not the case, then Spoiler wins the game $\BP_3(G',H')$ from the initial position $\big((w_{k-1},w_k),(v_{k-1},v_k)\big)$, again by Theorem \ref{thm:eq-wl-pebble-tuples}.
 However, such a winning strategy can easily be translated to a winning strategy for Spoiler in $\BP_{k+1}(G,H)$ from the initial position $\big((w_1,\dots,w_k),(v_1,\dots,v_k)\big)$ by never touching the $k-2$ first pebbles.
 This is a contradiction.
 So $\chi_{G'}(w_{k-1},w_k) = \chi_{H'}(v_{k-1},v_k)$.
 
 If $G'$ is not connected, the same must hold for $H'$ since $2$-WL knows whether a graph is connected.
 First, suppose there are vertex sets $C_G,C_G'$ of two connected components of $G'$ such that $C_G,C_G' \nsubseteq \{w_{k-1},w_k\}$.
 Then there are also vertex sets $C_H,C_H'$ of distinct connected components of $H'$ such that $C_H,C_H' \nsubseteq \{v_{k-1},v_k\}$.
 In particular, $\{v_1,\dots,v_k\}$ forms a $k$-separator in $H$.

 Otherwise, $G'$ has exactly two connected components with vertex sets $C_G,C_G'$ and $|C_G| = 1$ and $C_G \subseteq \{w_{k-1},w_k\}$.
 Without loss of generality, assume that $C_G = \{w_{k-1}\}$.
 Then $H'$ has exactly two connected components with vertex sets $C_H,C_H'$, one of which, say $C_H$, equals $\{v_{k-1}\}$.
 Let $G'' \coloneqq G' - \{w_{k-1}\}$ and $H'' \coloneqq H' - \{v_{k-1}\}$.
 Since $w_{k-1}$ and~$v_{k-1}$, respectively, is the only isolated vertex in $G'$ and $H'$, respectively, it follows that $\chi_{G''}(w_k,w_k) = \chi_{H''}(v_k,v_k)$.
 But now, $w_k$ is a cut vertex in $G''$.
 By \cite[Corollary 7]{kieponschwei19}, it holds that $v_k$ is a cut vertex in $H''$.
 This means that $\{v_1,\dots,v_k\}$ forms a $k$-separator in $H$.
 
 In the last case, both $G'$ and $H'$ are connected and $\{w_{k-1},w_k\}$ forms a $2$-separator in the graph $G' = G - \{w_1, \dots, w_{k-2}\}$.
 Thus, $\{v_{k-1},v_k\}$ forms a $2$-separator in the graph $H - \{v_1, \dots, v_{k-2}\}$ by the arguments given above.
 So $\{v_1,\dots,v_k\}$ forms a $k$-separator in~$H$.
\end{proof}

Following \cite{kieponschwei19}, we say that $k$-WL \emph{determines orbits} on a graph class $\mathcal{G}$ if for all arc-colored graphs $(G, \lambda)$, $(G', \lambda')$ with $G, G' \in \mathcal{G}$, and for all $v \in V(G)$ and $v' \in V(G')$, there exists an isomorphism from $(G, \lambda)$ to $(G', \lambda')$ mapping $v$ to $v'$ if and only if it holds that $\chi_{G,k}(v,\dots,v) = \chi_{G',k}(v',\dots,v')$.

Using the corollary, we can prove a strengthened version of \cite[Theorem 13]{kieponschwei19}.
The theorem says that on every minor-closed graph class $\mathcal{G}$, it holds that if $3$-WL determines vertex orbits on all arc-colored $3$-connected graphs in $\mathcal{G}$, then $3$-WL also distinguishes all pairs of non-isomorphic graphs in $\mathcal{G}$.
We show that the same statement actually holds already when replacing the occurrences of $3$-WL with $2$-WL (see Theorem \ref{thm:reduction-to-3-connected}).
Roughly speaking, the main step is to argue that $2$-WL implicitly computes the decomposition into triconnected components (in the sense of Tutte).
We can then prove the above statement using standard dynamic-programming arguments to decide isomorphism in a bottom-up fashion along the decomposition into triconnected components.

We start by formally defining what it means for $2$-WL to \emph{implicitly compute} a decomposition. 
Let $G$ be a graph and let $(T,\beta)$ be a tree decomposition of $G$.
For the remainder of this section, we restrict ourselves to \emph{rooted tree decompositions}, i.e., one node $r \in V(T)$, the \emph{root} of $T$, is designated and all edges are directed away from the root. Also, we assume that $\beta(s) \not\supseteq \beta(t)$ holds for all $(s,t) \in E(T)$.
For every $v \in V(G)$, define $t_v \in V(T)$ to be the unique node closest to the root such that $v \in \beta(t_v)$.
Note that, by the above assumption, for every $t \in V(T)$, there is a $v \in V(G)$ such that $t = t_v$.

\begin{definition}
 Let $G$ be a graph and let $(T,\beta)$ be a rooted tree decomposition of $G$.
 For~$k \geq 2$, we say that $k$-WL \emph{implicitly computes} $(T,\beta)$ if there are two sets $\mathcal{C}_{\sf bags},\mathcal{C}_{\sf edges} \subseteq \{\chi_{G,k}(v,w,\dots,w) \mid v,w \in V(G)\}$ of colors such that
 \[\chi_{G,k}(v,w,\dots,w) \in \mathcal{C}_{\sf bags} \;\;\Leftrightarrow\;\; w \in \beta(t_v)\]
 and
 \[\chi_{G,k}(v,w,\dots,w) \in \mathcal{C}_{\sf edges} \;\;\Leftrightarrow\;\; (t_v,t_w) \in E(T).\]
\end{definition}

For simplicity, we restrict ourselves to the case $k=2$ in the following.
We start with some basic observations.
Let $G$ be a graph and suppose $2$-WL implicitly computes a rooted tree decomposition $(T,\beta)$ of $G$.
Then there is also a set of colors $\mathcal{C}_{\sf eq} \subseteq \{\chi_G(v,w) \mid v,w \in V(G)\}$ such that
\[\chi_G(v,w) \in \mathcal{C}_{\sf eq} \;\;\Leftrightarrow\;\; \beta(t_w) = \beta(t_v).\]
Hence, each connected component of the graph $G[\mathcal{C}_{\sf eq}]$ (recall the definition in the paragraph before Lemma \ref{la:factor-graph-2-wl}) corresponds to one bag of the tree decomposition $(T,\beta)$.
Moreover, the rooted tree $T$ can be reconstructed by considering the graph with vertex set $V(G/\mathcal{C}_{\sf eq})$ and edge set
\[\big\{(B_1,B_2) \ \big\vert \ (\chi_G/\mathcal{C}_{\sf eq})(B_1,B_2) \subseteq \mathcal{C}_{\sf edges}\big\}.\]
In particular, the rooted tree decomposition $(T,\beta)$ can be reconstructed from the sets $\mathcal{C}_{\sf bags}$ and $\mathcal{C}_{\sf edges}$.

\begin{definition}
Let $\mathcal{G}$ be a graph class. The \emph{WL dimension} of $\mathcal{G}$, denoted by $\dimWL(\mathcal{G})$, is the minimum $k \in \mathbb{N}$ such that $k$-WL identifies every graph $G \in \mathcal{G}$. (If there is no such $k$, set $\dimWL(\mathcal{G})$ to be $\infty$.)
\end{definition}

Intuitively speaking, arguing that $2$-WL implicitly computes a certain type of tree decomposition can be a powerful tool for showing that the WL dimension of a graph class is bounded.
It essentially suffices to show that $k$-WL can handle the graphs that are induced by the bags.

As an important example, we now argue that $2$-WL implicitly computes the decomposition into triconnected components and thereby reduce the task of determining $\dimWL(\mathcal{G})$ of a minor-closed graph class $\mathcal{G}$ to dealing with all $3$-connected graphs in $\mathcal{G}$.
Recall that, by our definition, all graphs of order at most $2$ are $3$-connected.

For $X \subseteq V(G)$, define the \emph{torso} $G[[X]]$ to be the graph with $V(G[[X]]) = X$ and
\[E(G[[X]]) = E(G[X]) \cup \{vw \mid v,w \in N(Z) \text{ for a connected component $Z$ of $G - X$}\}.\]

\begin{theorem}\label{thm:implicitly-computes-decomposition}
 Let $G$ be a graph.
 Then there is a rooted tree decomposition $(T,\beta)$ of $G$ such that
 \begin{enumerate}
  \item $|\beta(s) \cap \beta(t)| \leq 2$ for all $(s,t) \in E(T)$, and
  \item for every $t \in V(T)$, the torso $G[[\beta(t)]]$ is $3$-connected or a cycle, and
  \item $2$-WL implicitly computes $(T,\beta)$.
 \end{enumerate}
\end{theorem}

The decomposition described in the theorem is the decomposition of a graph into triconnected components in the sense of Tutte.
Hence, the only challenge is to prove that $2$-WL implicitly computes this decomposition.
The proof of the theorem uses the same arguments as the one for \cite[Theorem 13]{kieponschwei19} and it can be found in the appendix.

Building on a bottom-up dynamic-programming strategy along the decomposition into triconnected components, we can strengthen Theorem 13 from \cite{kieponschwei19}. 

\begin{theorem}
 \label{thm:reduction-to-3-connected}
 Let $\mathcal{G}$ be a minor-closed graph class and assume $k \geq 2$.
 Suppose the $k$-dimensional WL algorithm determines orbits on the class of all arc-colored $3$-connected graphs in $\mathcal{G}$.
 Then the $k$-di\-men\-sion\-al WL algorithm distinguishes all non-isomorphic graphs in $\mathcal{G}$.
\end{theorem}

Again, the proof of the theorem follows the lines of the proof of \cite[Theorem 13]{kieponschwei19}, exploiting the improved bound on the dimension of the WL algorithm required to distinguish $2$-separators from other pairs of vertices stated in Theorem \ref{thm:separators}.
For the sake of completeness, all details can be found in the appendix.

In \cite{kieponschwei19}, it was shown that the WL dimension of the class of planar graphs is either $2$ or $3$.
As an application of Theorem \ref{thm:reduction-to-3-connected}, we obtain that the precise value only depends on the dimension needed to determine orbits on the class of all arc-colored $3$-connected planar graphs.

\medskip

As another consequence of our results, we show that $2$-WL also knows sizes of connected components after removing a $2$-separator.
This particular result forms the basis for improving the upper bound on a parametrization of the WL dimension by the treewidth of the input graph in Section \ref{sec:treewidth}.

For a graph $G$ and $v_1,v_2,v_3 \in V(G)$, we define $s_G(v_1,v_2,v_3) \coloneqq |C|$, where $C$ is the vertex set of the connected component of $G - \{v_1,v_2\}$ that contains $v_3$
(if $v_3 \in \{v_1,v_2\}$, then we set $s_G(v_1,v_2,v_3) \coloneqq 0$).

\begin{theorem}\label{thm:sizes-two-connected}
 Let $G$ and $H$ be $2$-connected graphs.
 Also suppose $v_1,v_2,v_3 \in V(G)$ and $w_1,w_2,w_3 \in V(H)$ such that $\chi_G(v_i,v_j) = \chi_H(w_i,w_j)$ for all $i,j \in \{1,2,3\}$.
 Then $s_G(v_1,v_2,v_3) = s_H(w_1,w_2,w_3)$.
\end{theorem}

\begin{proof}
 The statement trivially holds if $v_3 = v_1$ or $v_3 = v_2$. Thus, assume $v_1 \neq v_3 \neq v_2$.

 If $\{v_1,v_2\}$ is not a $2$-separator, the statement follows easily from Theorem \ref{thm:separators}.
 Thus, we may assume that $\{v_1,v_2\}$ and $\{w_1,w_2\}$ are $2$-separators.
 Recall the notation $G|_S$ (preceding Lemma \ref{la:wl-color-subset}).
 Let $S \coloneqq \{\chi_G(v_1,v_1),\chi_G(v_2,v_2)\}$ and let $G' \coloneqq G|_S$.
 As in the proof of Theorem~\ref{thm:separators}, the graph $G'$ is $2$-connected.
 
 First suppose that $|V(G')| = 2$.
 Let $A \coloneqq \{v \in V(H) \mid \chi_H(v,v) \in S\}$.
 Then $A = \{w_1,w_2\}$.
 Moreover, $s_G(v_1,v_2,v_3)$ is the number of vertices that are reachable from $v_3$ without ever visiting a vertex with a color from $S$.
 This is encoded in the color $\chi_G(v_3,v_3)$.
 Since $\chi_G(v_3,v_3) = \chi_H(w_3,w_3)$, it follows that $s_G(v_1,v_2,v_3) = s_H(w_1,w_2,w_3)$.
 
 Next, suppose there is a vertex set $C$ of a connected component of $G - \{v_1,v_2\}$ such that $V(G') \subseteq C \cup \{v_1,v_2\}$.
 If $v_3 \notin C$, then $v_1$ and $v_2$ are the only vertices in $V(G')$ that $v_3$ can reach via paths that avoid $S$.
 Also, as before, $s_G(v_1,v_2,v_3)$ is the number of vertices that are reachable from $v_3$ without ever visiting a vertex with a color in $S$.
 The same has to hold for $w_1$ and $w_2$ with respect to $w_3$, since $\chi_G(v_3,v_3) = \chi_H(w_3,w_3)$. Thus, $s_G(v_1,v_2,v_3) = s_H(w_1,w_2,w_3)$.
 In the other case, $v_3 \in C$ and $s_G(v_1,v_2,v_3) = |C|$.
 But $|C| = n - 2 - \sum_{C' \neq C} |C'|$, where $C'$ ranges over all vertex sets of connected components of $G - \{v_1,v_2\}$.
 As discussed above, the sizes of these sets $C'$ are encoded in the vertex colors of the graph $G$ and the same is true for $H$.
 
 Otherwise, $\{v_1,v_2\}$ forms a $2$-separator in $G'$ and hence, $G'$ is a cycle by Lemma \ref{la:wl-color-subset} and Theorem \ref{thm:two-colors-main}.
 Note that $|V(G')| \geq 4$ and $v_1v_2 \notin E(G')$.
 Using Lemma \ref{la:wl-color-subset}, it follows that~$H|_S$ is also a cycle.
 Also, $|V(H|_S)| \geq 4$ and $w_1w_2 \notin E(H|_S)$.
 Note that every vertex can reach only two vertices in $S$ via paths that avoid $S$.
 Thus, the colors $\chi_G(v,v')$ for $vv' \in E(G')$ encode the number of vertices with color not in $S$ for which $v$ and $v'$ are the only vertices reachable via paths that avoid $S$.
 Moreover, $2$-WL knows the arc-colored cycle $G'$ (and the analogous arc-colored cycle $H'$ for $H$).
 Thus, it suffices to consider these two cycles, for which the statement is easy to see.
\end{proof}

The last theorem can also be formulated in terms of the expressive power of the $3$-variable fragment ${\sf C}^{3}$ of first-order logic with counting quantifiers of the form $\exists^{\geq k} x \varphi(x)$.
Indeed, it implies that for all $n,s \in \mathbb{N}$, there is a formula $\varphi_{n,s}(x_1,x_2,x_3) \in {\sf C}^{3}$ such that,
for every $2$-connected $n$-vertex graph $G$ and $v_1,v_2,v_3 \in V(G)$, it holds that $G \models \varphi_{n,s}(v_1,v_2,v_3)$ if and only if $s_G(v_1,v_2,v_3) = s$ (for details about the connection between the WL algorithm and counting logics, see, e.g., \cite{CaiFI92,Grohe17,IL90}).

\section{New Bounds for Graphs of Treewidth $\boldsymbol{k}$}
\label{sec:treewidth}

As an application of the results presented so far, we investigate the WL dimension of the class of graphs of treewidth at most $k$.
Up to this point, the best known upper bound on the WL dimension of such graphs has been $k+2$, i.e., $(k+2)$-WL identifies every graph of treewidth at most $k$ \cite{GroheM99}. In this section, we present new upper and lower bounds, as stated in the following theorem.

\begin{theorem}
 \label{thm:wl-treewidth-bounds}
 Let $k \geq 2$. Then $\left\lceil\frac{k}{2}\right\rceil - 2 \leq \dimWL(\mathcal{T}_k) \leq k$, where $\mathcal{T}_k$ denotes the class of graphs of treewidth at most $k$.
\end{theorem}

Roughly speaking, for the upper bound, given two graphs $G$ and $H$ where $G$ has treewidth at most $k$, we use the correspondence from Corollary \ref{cor:eq-wl-pebble} and force the players in the game $\BP_k(G,H)$ to descend along a tree decomposition of width $k$ in $G$. 
For this, we build on an extension of Theorem \ref{thm:sizes-two-connected} to arbitrary separator sizes.

Meanwhile, the lower bound in Theorem \ref{thm:wl-treewidth-bounds} is an independent insight, i.e., it does not build on the results from the previous sections. Here, we define a family of difficult instances by applying a construction introduced by Cai, Fürer, and Immerman \cite{CaiFI92} to modified grid graphs.

\subsection{Upper Bound}

As indicated, the basic idea for proving a new upper bound is to provide a winning strategy for Spoiler in the corresponding bijective pebble game and it works similarly to the proof that $(k+2)$-WL identifies every graph of treewidth at most $k$ given in \cite{GroheM99}.
The main difference is a much more careful implementation of the general strategy in order to get by with the desired number of pebbles.
As a major ingredient, we exploit that separators can be detected using fewer pebbles.

For a $(k+1)$-tuple $(v_1, \dots, v_{k+1})$ of vertices of a graph $G$, we define
\[s_G(v_1,\dots,v_{k+1}) \coloneqq |C|,\]
where $C$ is the vertex set of the unique connected component of $G - \{v_1,\dots,v_k\}$ with $v_{k+1} \in C$
(as before, $s_G(v_1,\dots,v_k,v_{k+1}) = 0$ if $v_{k+1} \in \{v_1,\dots,v_k\}$).

\begin{lemma}
 \label{lem:separator-component-sizes-wl}
 Suppose $k \geq 2$.
 Let $G$ and $H$ be two graphs and consider two tuples $(v_1,\dots,v_{k+1}) \in \big(V(G)\big)^{k+1}$ and $(w_1,\dots,w_{k+1}) \in \big(V(H)\big)^{k+1}$ for which $s_G(v_1,\dots,v_{k+1}) \neq s_H(w_1,\dots,w_{k+1})$.
 Then Spoiler wins the game $\BP_{k+1}(G,H)$ from the initial position $\big((v_1,\dots,v_{k+1}),(w_1,\dots,w_{k+1})\big)$.
\end{lemma}

\begin{proof}
 It is easy to see that it suffices to prove the statement for the case that the following conditions hold:
 \begin{itemize}
  \item $v_{k+1} \notin \{v_1, \dots, v_k\}$ and $w_{k+1} \notin \{w_1, \dots, w_k\}$, and
  \item $G$ and $H$ are connected and of the same order. 
 \end{itemize}
 Also note that, by Corollary \ref{cor:separators}, the graph $G$ is $2$-connected if and only if $H$ is $2$-connected.

 First suppose $k = 2$.
 Assume the graphs are connected, but not $2$-connected.
 We are going to use the correspondence from Theorem \ref{thm:eq-wl-pebble-tuples}.
 The proof basically exploits the fact that the decomposition into $2$-connected components has a tree-like structure and that $2$-WL recognizes cut vertices (see \cite[Corollary 7]{kieponschwei19}), thus being able to ``transport information'' from one side of a cut vertex to the other.
 Note that $v_i$ is a cut vertex if and only if $w_i$ is a cut vertex.

 First, suppose that exactly one vertex in $\{v_1,v_2\}$ is a cut vertex, say $v_1$.
 Then we can ignore the vertex $v_2$. Also, for every $v' \in V(G)$, the color triple formed by $\chi_G(v_1,v'), \chi_G(v',v_3)$, and $\chi_G(v_1,v_3)$ encodes how many vertices are contained in the same connected component of the graph $G - \{v_1\}$ as $v_3$, using the methods from \cite[Theorem 6]{kieponschwei19}.
 Thus, $\chi_G(v_1,v_3)$ encodes the size of the connected component of $G - \{v_1\}$ that contains $v_3$.
 Also, by the above observation, we know whether $v_2$ is in this connected component or not.
 Thus, we know the size of the connected component of $G - \{v_1,v_2\}$ that contains $v_3$.

 Now assume that both $v_1$ and $v_2$ are cut vertices.
 Then, as above, we can determine whether $v_3$ lies in the same connected component of $G - \{v_1\}$ as $v_2$ and whether $v_3$ lies in the same connected component of $G - \{v_2\}$ as $v_1$, and compute the corresponding sizes.
 These numbers determine $s_G$ and thus also $s_H$.

 Otherwise, neither $v_1$ nor $v_2$ is a cut vertex.
 The only problematic case is that $\{v_1,v_2\}$ forms a $2$-separator (implying that $v_1$ and $v_2$ lie in a common $2$-connected component of $G$).
 In this case, we can proceed analogously as in the proof of Theorem \ref{thm:sizes-two-connected}.

 For the general case, i.e., $k > 2$, let $\widehat G \coloneqq G - \{v_1,\dots,v_{k-2}\}$ and $\widehat H \coloneqq H - \{v_1,\dots,v_{k-2}\}$.
 Then $s_{\widehat G}(v_{k-1},v_k,v_{k+1}) \neq s_{\widehat H}(w_{k-1},w_k,w_{k+1})$ and hence,
 Spoiler wins the game $\BP_3(\widehat G,\widehat H)$ from the initial position $\big((v_{k-1},v_k,v_{k+1}),(w_{k-1},w_k,w_{k+1})\big)$ by Theorem \ref{thm:eq-wl-pebble-tuples} and the first part of the proof.
 But then Spoiler also wins the game $\BP_{k+1}(G,H)$ from the initial position $\big((v_1,\dots,v_{k+1}),(w_1,\dots,w_{k+1})\big)$ by simply never moving the first $k-2$ pebbles.
\end{proof}

To build Spoiler's strategy along a given tree decomposition, we use the following characterization of treewidth.
Let $G$ be a graph of treewidth $k$.
For a $k$-separator $S \subseteq V(G)$ and the vertex set $C$ of a connected component of $G - S$, we define $G(S,C)$ to be the graph on vertex set $S \cup C$ obtained by inserting a clique between the vertices in $S$ into $G[S \cup C]$.

\begin{lemma}[Arnborg et al.\ \cite{ArnborgCP87}]
 \label{la:inductive-characterization-treewidth}
 Suppose $G(S,C)$ has at least $k+2$ vertices.
 Then $G(S,C)$ has treewidth at most $k$ if and only if there exists $v \in C$ such that for every connected component $A$ of $G[C \setminus \{v\}]$,
 there is a $k$-separator $S_A \subseteq S \cup \{v\}$ such that
 \begin{enumerate}
  \item no vertex in $A$ is adjacent to the unique element from $S \setminus S_A$, and
  \item $G\big(S_A,V(A)\big)$ has treewidth at most $k$.
 \end{enumerate}
\end{lemma}

Suppose $G(S,C)$ has treewidth at most $k$.
Let $D_G(S,C)$ denote the set of possible vertices~$v \in C$ that satisfy Lemma \ref{la:inductive-characterization-treewidth}.

\begin{theorem}
 \label{thm:wl-treewidth-upper-bound}
 Suppose $k \geq 2$.
 Let $G$ be a graph of treewidth at most $k$.
 Then the $k$-dimensional Weisfeiler-Leman algorithm identifies $G$.
\end{theorem}

\begin{proof}
 Let $G$ be a connected graph of treewidth $k$ and suppose $H$ is a second connected graph with $H \not\cong G$.
 Let $(T,\beta)$ be a tree decomposition of $G$ of width $k$.
 For a $k$-separator $S \subseteq V(G)$ and an integer $m \in \mathbb{N}$, we define 
 \begin{align*}
  \mathcal{C}_G(S,m) \coloneqq \big\{C \subseteq V(G) \ \big\vert \ &C \text{ is the vertex set of a connected}\\
                                                         &\text{component of $G - S$ of size $m$}\big\}.
 \end{align*}
 Moreover,
 \[G(S,m) \coloneqq G\left[S \cup \bigcup_{C \in \mathcal{C}_G(S,m)} C\right].\]
 An \emph{ordered separator} is a tuple $\bar a = (a_1,\dots,a_k)$ where the underlying set $\{a_1,\dots,a_k\}$ is a separator.
 In this proof, slightly abusing notation, we do not distinguish between ordered separators and their underlying unordered separators.
 For two ordered separators $\bar a \in \big(V(G)\big)^{k}$ and $\bar b \in \big(V(H)\big)^{k}$,
 we define $m(\bar a,\bar b)$ to be the minimal number $m \geq 1$ such that $(G(\bar a,m),\bar a) \not\cong (H(\bar b,m),\bar b)$.
 Here, $(G(\bar a,m),\bar a)$ denotes the graph $G(\bar a,m)$ where, additionally, each vertex in $\bar a$ is individualized.
 More formally, it holds that $(G(\bar a,m),\bar a) \cong (H(\bar b,m),\bar b)$ if there is an isomorphism $\varphi\colon G(\bar a,m) \cong H(\bar b,m)$ such that for $\bar a = (a_1,\dots,a_k)$ and $\bar b = (b_1,\dots,b_k)$, it holds that $\varphi(a_i) = b_i$ for all $i \in [k]$.
 
 We now argue that Spoiler wins the game $\BP_{k+1}(G,H)$.
 Suppose the game is in a position $(\bar a,\bar b) \in \big(V(G)\big)^{k} \times \big(V(H)\big)^{k}$ where there is an edge $st \in E(T)$ with $\beta(s) \cap \beta(t) \subseteq \bar a$ (as sets).
 We shall prove by induction on $m \coloneqq m(\bar a,\bar b)$ that Spoiler wins the game from the initial position $(\bar a,\bar b)$.
 In each case, Spoiler wishes to play another pebble.
 Let $f\colon V(G) \rightarrow V(H)$ be the bijection chosen by Duplicator.
 Using Lemma \ref{lem:separator-component-sizes-wl}, we can assume that $f$ maps the vertex set of $G(\bar a,m)$ to the vertex set of $H(\bar b,m)$.
 Now let $C \in \mathcal{C}_G(\bar a,m)$ be such that
 \begin{align*}
     &|\{C' \in \mathcal{C}_G(\bar a,m) \mid (G(\bar a,C'),\bar a) \cong (G(\bar a,C),\bar a)\}|\\
  >\;&|\{C' \in \mathcal{C}_H(\bar b,m) \mid (H(\bar b,C'),\bar b) \cong (G(\bar a,C),\bar a)\}|.
 \end{align*}
 Also let
 \[D \coloneqq \Big\{v \in D_G(\bar a, C') \mid C' \in \mathcal{C}_G(\bar a,m) \text{ and } \big(G(\bar a,C'),\bar a\big) \cong (G(\bar a,C),\bar a)\Big\}.\]
 Then there exist a $v \in D$ and $C_G \in \mathcal{C}_G(\bar a,m)$ and $C_H \in \mathcal{C}_H(\bar b,m)$ with $v \in C_G$ and $f(v) \in C_H$ and
 \begin{equation}
  \label{eq:non-isomorphic-components}
  \big(G[C_G \cup \bar a],\bar a,v\big) \not\cong \big(G[C_H \cup \bar b],\bar b,f(v)\big).
 \end{equation}
 Now Spoiler places pebbles on $(v,w)$ with $w = f(v)$.
 
 For the base case of the induction, suppose $m = 1$.
 This means $C_G = \{v\}$ and $C_H = \{w\}$ and thus, Spoiler wins immediately. So assume $m > 1$.
 Let $A_1,\dots,A_\ell \subseteq C_G$ be the vertex sets of the connected components of $G[C_G \setminus \{v\}]$.
 Note that $|A_i| \leq m-1$ for every $i \in [\ell]$.
 Also let $B_1,\dots,B_{\ell'} \subseteq C_H$ be the vertex sets of the connected components of $H[C_H \setminus \{w\}]$.
 Due to Equation (\ref{eq:non-isomorphic-components}), there is an $A \in \{A_1,\dots,A_\ell\}$ such that
 \begin{align*}
     &\big|\big\{i \in [\ell] \, \big\vert \, G[A_i \cup \bar a \cup \{v\}],\bar a,v) \cong G[A \cup \bar a \cup \{v\}],\bar a,v)\big\}\big|\\
  >\;&\big|\big\{i \in [\ell'] \, \big\vert \, H[B_i \cup \bar b \cup \{w\}],\bar b,w) \cong G[A \cup \bar a \cup \{v\}],\bar a,v)\big\}\big|.
 \end{align*}
 We pick such a set $A \in \{A_1,\dots,A_\ell\}$ with minimal cardinality (i.e., there is no set $A' \in \{A_1,\dots,A_\ell\}$ strictly smaller than $A$ and satisfying the above condition).
 Now suppose $\bar a = (a_1,\dots,a_k)$ and $\bar b = (b_1,\dots,b_k)$.
 Pick $i \in [k]$ such that no vertex in $A$ is adjacent to $a_i$ (cf.\ Lemma \ref{la:inductive-characterization-treewidth}).
 Now Spoiler removes the pair of pebbles $(a_i,b_i)$.
 Let $\bar a' \coloneqq (a_1,\dots,a_{i-1},a_{i+1},\dots,a_k,v)$ and $\bar b' \coloneqq (b_1,\dots,b_{i-1},b_{i+1},\dots,b_k,w)$.
 Observe that $(\bar a',\bar b')$ is the current position of the game.
 Now let $m' = |A| < m$.
 Note that $A \in \mathcal{C}_G(\bar a' ,m')$.
 
 \begin{claim}
  $\big(G(\bar a',m'),\bar a'\big) \not\cong \big(H(\bar b',m'),\bar b'\big)$.
 \end{claim}
 \begin{claimproof}
  To prove the claim, it suffices to argue that $|\mathcal{A}| > |\mathcal{B}|$ where
  \[\mathcal{A} = \Big\{A' \in \mathcal{C}_G(\bar a',m') \mid \big(G(\bar a',A'),\bar a'\big) \cong \big(G(\bar a',A),\bar a'\big)\Big\}\]
  and
  \[\mathcal{B} = \Big\{B' \in \mathcal{C}_H(\bar b',m') \mid \big(H(\bar b',B'),\bar b'\big) \cong \big(G(\bar a',A),\bar a'\big)\Big\}.\]
  Let $\mathcal{A}' \coloneqq \{A' \in \mathcal{A} \mid A' \subseteq C_G\}$ and $\mathcal{A}'' \coloneqq \mathcal{A} \setminus \mathcal{A}'$.
  Similarly, define $\mathcal{B}' \coloneqq \{B' \in \mathcal{B} \mid B' \subseteq C_H\}$ and $\mathcal{B}'' \coloneqq \mathcal{B} \setminus \mathcal{B}'$.
  From the definition of the set $A$, it follows that $|\mathcal{A}'| > |\mathcal{B}'|$.
  Now define
  \[G' \coloneqq G\!\left[\bar a \cup \{v\} \cup \bigcup_{m'' \leq m'} \mathcal{C}_G(\bar a,m'') \cup \bigcup_{i \in [\ell]\colon |A_i| < m'} A_i\right]\]
  and
  \[H' \coloneqq H\!\left[\bar b \cup \{w\} \cup \bigcup_{m'' \leq m'} \mathcal{C}_H(\bar b,m'') \cup \bigcup_{i \in [\ell']\colon |B_i| < m'} B_i\right].\]
  From the definitions of the number $m$ and the set $A$, it follows that $(G',\bar a,v) \cong (H',\bar b,w)$.
  
  Now let $A' \in \mathcal{A}''$ and let $C$ be the vertex set of a connected component of $G - \bar a$ such that $C \neq C_G$ and $A' \cap C \neq \emptyset$.
  Then $C \subseteq A'$, since $C$ is a connected set in the graph $G - \bar a'$.
  Since $|A'| = m'$, this implies $A' \subseteq V(G')$.
  By the same argument, $B' \subseteq V(H')$ holds for all $B' \in \mathcal{B}''$.
  But this means that $|\mathcal{A}''| = |\mathcal{B}''|$ because $(G',\bar a,v) \cong (H',\bar b,w)$.
  Overall, this proves the claim.
 \end{claimproof}
 Since $m' < m$, Spoiler wins from the initial position $(\bar a',\bar b')$ by the induction hypothesis.
 Using the induction principle, this completes the proof.
\end{proof}

\subsection{Lower Bound}

For the lower bound, we use a construction introduced by Cai, F\"urer, and Immerman \cite{CaiFI92} and start by reviewing it. For a non-empty finite set $S$, we define the \emph{CFI gadget} $X_S$ to be the following graph.
For each $w \in S$, there are vertices $a(w)$ and $b(w)$, and for every $A \subseteq S$ such that $|A|$ is even, there is a vertex $m_A$.
For every $A \subseteq S$ such that $|A|$ is even, there are edges $\{a(w),m_A\} \in E(X_S)$ for all $w \in A$ and $\{b(w),m_A\} \in E(X_S)$ for all $w \in S\setminus A$.
As an example, the graph $X_3 \coloneqq X_{[3]}$ is depicted in Figure \ref{fig:cfi-gadget}.
The graph is colored so that~$\{m_A\mid A \subseteq S \text{ and } |A| \text{ is even}\}$ forms a color class and so that~$\{a(w),b(w)\}$ forms a color class for each $w \in S$.

\begin{figure}
 \centering
 \begin{tikzpicture}
  \node[style=normalvertex,fill=red,label=right:{$a(1)$}] (a1) at (6,2) {};
  \node[style=normalvertex,fill=red,label=right:{$b(1)$}] (b1) at (6,3) {};
  
  \node[style=normalvertex,fill=blue,label=left:{$a(2)$}] (a2) at (0,0) {};
  \node[style=normalvertex,fill=blue,label=left:{$b(2)$}] (b2) at (0,1) {};
  
  \node[style=normalvertex,fill=darkpastelgreen,label=left:{$a(3)$}] (a3) at (0,4) {};
  \node[style=normalvertex,fill=darkpastelgreen,label=left:{$b(3)$}] (b3) at (0,5) {};
  
  \node[style=normalvertex,label=85:{{\small $m_\emptyset$}}] (v1) at (3,4) {};
  \node[style=normalvertex,label=85:{{\small $m_{\{2,3\}}$}}] (v2) at (3,3) {};
  \node[style=normalvertex,label=85:{{\small $m_{\{1,3\}}$}}] (v3) at (3,2) {};
  \node[style=normalvertex,label=85:{{\small $m_{\{1,2\}}$}}] (v4) at (3,1) {};
  
  \path
   (v1) edge (b1)
   (v1) edge (b2)
   (v1) edge (b3)
   (v2) edge (b1)
   (v2) edge (a2)
   (v2) edge (a3)
   (v3) edge (a1)
   (v3) edge (b2)
   (v3) edge (a3)
   (v4) edge (a1)
   (v4) edge (a2)
   (v4) edge (b3);
  
 \end{tikzpicture}
 \caption{The Cai-F\"{u}rer-Immerman gadget $X_3$.}
 \label{fig:cfi-gadget}
\end{figure}

Let $G$ be a connected graph of minimum degree $2$.
For $T \subseteq E(G)$, we define $\CFI_T(G)$ to be the vertex-colored graph obtained from $G$ in the following way.
Each $v \in V(G)$ is replaced with a gadget $X_{E(v)}$, where $E(v) = \{(v,w) \mid vw \in E(G)\}$ denotes the set of (directed) edges incident with $v$.
Additionally, the following edges are inserted between the gadgets.
For every $vw \in E(G) \setminus T$, there are edges from $a(v,w)$ to $a(w,v)$ and from $b(v,w)$ to $b(w,v)$.
Also, for every $vw \in T$, there are edges from $a(v,w)$ to $b(w,v)$ and from $b(v,w)$ to $a(w,v)$.
Finally, we define the vertex coloring $\lambda$ of $\CFI_T(G)$ as follows:
\begin{itemize}
 \item $\lambda\big(a(v,w)\big) = \chi\big(b(v,w)\big) \coloneqq (1,v,w)$ for all $vw \in E(G)$, and
 \item $\lambda(m_A) \coloneqq (0,v)$ for every $v \in V(G)$ and every $A \subseteq E(v)$ such that $|A|$ is even.
\end{itemize}

\begin{lemma}[\cite{CaiFI92}]
 Let $G$ be a connected graph of minimum degree $2$ and $S,T \subseteq E(G)$.
 Then $\CFI_S(G) \cong \CFI_T(G)$ if and only if $|S| \equiv |T| \mod 2$.
\end{lemma}

Hence, applying the above construction to a specific graph $G$ yields a pair of non-isomorphic graphs $\CFI(G) = \CFI_{\emptyset}(G)$ and $\twCFI(G) = \CFI_{\{e\}}(G)$ for some $e \in E(G)$.

\begin{theorem}[Dawar and Richerby \cite{DawarR07}]
 \label{thm:cfi-equivalence-tree-width}
 Let $G$ be a connected graph such that $\tw(G) \geq k+1$ and $\deg(v) \geq 2$ for all $v \in V(G)$.
 Then $\CFI(G) \simeq_k \twCFI(G)$.
\end{theorem}

The strategy to obtain a good lower bound is to find graphs $G$ for which we can show a sufficiently strong upper bound on the treewidths of $\CFI(G)$ and $\twCFI(G)$.
Here, we first provide an analysis for the simple case that $G$ is an $(n \times n)$-grid.
This almost achieves the lower bound stated in Theorem \ref{thm:wl-treewidth-bounds}.
Then, we argue how to modify the grid to obtain the desired improvement.

For $n \geq 2$, let $G_{n,n}$ be the $(n \times n)$-grid.
Moreover, let $G_{n,n}^{+}$ be the $(n \times n)$-grid in which each edge is replaced with a path of length $3$.
Formally,
$V(G_{n,n}) = [n] \times [n]$
and
$E(G_{n,n}) = \{(i,j)(i',j') \mid (i = i' \wedge |j - j'| = 1) \vee (j = j' \wedge |i-i'| = 1)\}$.
Moreover,
$V(G_{n,n}^{+}) = V(G_{n,n}) \cup \{(v,w) \mid vw \in E(G_{n,n})\}$
and 
$E(G_{n,n}^{+}) = \{v(v,w) \mid v \in V(G_{n,n}), \ vw \in E(G_{n,n})\} \cup \{(v,w)(w,v) \mid vw \in E(G_{n,n})\}$.

\begin{lemma}
 \label{la:tree-decomposition-subdivided-grid}
 Let $n \geq 2$.
 Then there is a tree decomposition $(T,\beta)$ of $G_{n,n}^{+}$ of width $n+2$ such that
 \begin{enumerate}
  \item $|\beta(t) \cap V(G_{n,n})| \leq 1$ for every $t \in V(T)$, and
  \item if $|\beta(t) \cap V(G_{n,n})| = 1$, then there exists a $v \in V(G_{n,n})$ such that $\beta(t) = E(v) \cup \{v\}$, where $E(v) = \{(v,w) \mid vw \in E(G_{n,n})\}$.
   In this case, $t$ is a leaf of $T$ and $\beta(s) \cap V(G_{n,n}) = \emptyset$ for the unique $s \in V(T)$ with $st \in E(T)$.
 \end{enumerate}
\end{lemma}

\begin{proof}
 In order to describe the bags of the tree decomposition, we start by defining several sets $A_{i,j},B_{i,j},C_{i,j} \subseteq V(G_{n,n}^{+})$ for $i,j \in [n]$ (see also Figure \ref{fig:grid-decomposition}).
 Let
 \begin{align*}
  A_{i,j} {}\coloneqq{} &\{((i',j),(i',j+1)) \mid 1 \leq i' \leq i\} {}\cup{}\\
            &\{((i',j),(i',j-1)) \mid i \leq i' \leq n\}{}\cup{}\\
            &\{((i,j),(i+1,j)), ((i,j),(i-1,j))\},\\[2ex]
  B_{i,j} {}\coloneqq{} &\{((i',j),(i',j+1)) \mid 1 \leq i' \leq i\} {}\cup{}\\
            &\{((i',j),(i',j-1)) \mid i < i' \leq n\} {}\cup{}\\
            &\{((i,j),(i+1,j)), ((i+1,j),(i,j))\},\\[2ex]
  C_{i,j} {}\coloneqq{} &\{((i',j),(i',j-1)) \mid 1 \leq i' \leq i\} {}\cup{}\\
           &\{((i',j-1),(i',j)) \mid i \leq i' \leq n\}.
 \end{align*}
 (Formally, the sets defined above may also contain elements outside of $V(G_{n,n}^{+})$ if some index is not contained in the set $[n]$.
 In this case, we simply do not include the corresponding element in the set.)
 \begin{figure}
 \centering
 \scalebox{0.8}{
 \begin{tikzpicture}
  \draw (0.5,0) edge (0.5,4.5);
  \draw (2.0,0) edge (2.0,4.5);
  \draw (3.5,0) edge (3.5,4.5);
  
  \draw (0,0.0) edge (4,0.0);
  \draw (0,1.5) edge (4,1.5);
  \draw (0,3.0) edge (4,3.0);
  \draw (0,4.5) edge (4,4.5);
  
  \foreach \i in {0,...,9}{
   \node[tinyvertex] (c1-\i) at (0.5,0.5*\i) {};
   \node[tinyvertex] (c2-\i) at (2.0,0.5*\i) {};
   \node[tinyvertex] (c3-\i) at (3.5,0.5*\i) {};
  }
  
  \foreach \j in {0,...,8}{
   \node[tinyvertex] (r1-\j) at (0.5*\j,0.0) {};
   \node[tinyvertex] (r2-\j) at (0.5*\j,1.5) {};
   \node[tinyvertex] (r3-\j) at (0.5*\j,3.0) {};
   \node[tinyvertex] (r4-\j) at (0.5*\j,4.5) {};
  }
  
  \draw[thick,dotted,rounded corners] (2.75,4.75) -- (2.75,3.00) -- (1.75,2.00) -- (1.75,-0.25) -- (1.25,-0.25) -- (1.25,3.00) -- (2.25,4.00) -- (2.25,4.75) -- cycle;
  \draw[thick,dotted] (2,3) circle (0.25);
  
  \node at (0.5,5) {{\small $j-1$}};
  \node at (2,5) {{\small $j$}};
  \node at (-0.5,3) {{\small $i$}};
  \node at (-0.6,1.5) {{\small $i+1$}};
  \node at (1.5,-0.6) {$A_{i,j}$};
  
 \end{tikzpicture}
 }
 \hspace*{15pt}
 \scalebox{0.8}{
 \begin{tikzpicture}
  \draw (0.5,0) edge (0.5,4.5);
  \draw (2.0,0) edge (2.0,4.5);
  \draw (3.5,0) edge (3.5,4.5);
  
  \draw (0,0.0) edge (4,0.0);
  \draw (0,1.5) edge (4,1.5);
  \draw (0,3.0) edge (4,3.0);
  \draw (0,4.5) edge (4,4.5);
  
  \foreach \i in {0,...,9}{
   \node[tinyvertex] (c1-\i) at (0.5,0.5*\i) {};
   \node[tinyvertex] (c2-\i) at (2.0,0.5*\i) {};
   \node[tinyvertex] (c3-\i) at (3.5,0.5*\i) {};
  }
  
  \foreach \j in {0,...,8}{
   \node[tinyvertex] (r1-\j) at (0.5*\j,0.0) {};
   \node[tinyvertex] (r2-\j) at (0.5*\j,1.5) {};
   \node[tinyvertex] (r3-\j) at (0.5*\j,3.0) {};
   \node[tinyvertex] (r4-\j) at (0.5*\j,4.5) {};
  }
  
  \draw[thick,dotted,rounded corners] (2.75,4.75) -- (2.75,2.50) -- (1.75,1.50) -- (1.75,-0.25) -- (1.25,-0.25) -- (1.25,2.00) -- (2.25,3.00) -- (2.25,4.75) -- cycle;
  
  \node at (0.5,5) {{\small $j-1$}};
  \node at (2,5) {{\small $j$}};
  \node at (1.5,-0.6) {$B_{i,j}$};
  
 \end{tikzpicture}
 }
 \hspace*{15pt}
 \scalebox{0.8}{
 \begin{tikzpicture}
  \draw (0.5,0) edge (0.5,4.5);
  \draw (2.0,0) edge (2.0,4.5);
  
  \draw (0,0.0) edge (2.5,0.0);
  \draw (0,1.5) edge (2.5,1.5);
  \draw (0,3.0) edge (2.5,3.0);
  \draw (0,4.5) edge (2.5,4.5);
  
  \foreach \i in {0,...,9}{
   \node[tinyvertex] (c1-\i) at (0.5,0.5*\i) {};
   \node[tinyvertex] (c2-\i) at (2.0,0.5*\i) {};
  }
  
  \foreach \j in {0,...,5}{
   \node[tinyvertex] (r1-\j) at (0.5*\j,0.0) {};
   \node[tinyvertex] (r2-\j) at (0.5*\j,1.5) {};
   \node[tinyvertex] (r3-\j) at (0.5*\j,3.0) {};
   \node[tinyvertex] (r4-\j) at (0.5*\j,4.5) {};
  }
  
  \draw[thick,dotted,rounded corners] (1.75,4.75) -- (1.75,3.00) -- (1.25,2.00) -- (1.25,-0.25) -- (0.75,-0.25) -- (0.75,3.00) -- (1.25,4.00) -- (1.25,4.75) -- cycle;
  
  \node at (0.5,5) {{\small $j-1$}};
  \node at (2,5) {{\small $j$}};
  \node at (1,-0.6) {$C_{i,j}$};
  
 \end{tikzpicture}
 }
 \caption{Visualization of the sets $A_{i,j}$, $B_{i,j}$ and $C_{i,j}$ constructed in the proof of Lemma \ref{la:tree-decomposition-subdivided-grid}.}
 \label{fig:grid-decomposition}
 \end{figure}
 Now define
 \[V(T) \coloneqq \{t_{i,j}^{A},t_{i,j}^{B},t_{i,j}^{C},t_{i,j}^{D} \mid i,j \in [n]\}.\]
 Also set
 \begin{align*}
  \beta(t_{i,j}^{A}) &\coloneqq A_{i,j},\\
  \beta(t_{i,j}^{B}) &\coloneqq B_{i,j},\\
  \beta(t_{i,j}^{C}) &\coloneqq C_{i,j},\\
  \beta(t_{i,j}^{D}) &\coloneqq E((i,j)) \cup \{(i,j)\} = N_{G^+_{n,n}}[(i,j)].
 \end{align*}
 Observe that each bag contains at most $n+3$ elements.
 It remains to define the edges of the tree $T$.
 The following edges are added to the set $E(T)$:
 \begin{itemize}
  \item $t_{i,j}^{C}t_{i+1,j}^{C}$ for all $i \in [n-1], j \in [n]$,
  \item $t_{n,j}^{C}t_{1,j}^{A}$ for all $j \in [n]$,
  \item $t_{i,j}^{A}t_{i,j}^{B}$ for all $i,j \in [n]$,
  \item $t_{i,j}^{A}t_{i,j}^{D}$ for all $i,j \in [n]$,
  \item $t_{i,j}^{B}t_{i+1,j}^{A}$ for all $i \in [n-1], j \in [n]$, and
  \item $t_{n,j}^{B}t_{1,j+1}^{C}$ for all $j \in [n-1]$.
 \end{itemize}
 It can be verified in a straight-forward manner that $(T,\beta)$ defines a tree decomposition of the graph $G_{n,n}^{+}$ with the desired properties.
\end{proof}

\begin{lemma}
 \label{la:tree-width-cfi-grid}
 Let $n \geq 2$. Then $\tw(\CFI(G_{n,n})) \leq 2n + 5$ and $\tw(\twCFI(G_{n,n})) \leq 2n + 5$.
\end{lemma}

\begin{proof}
 We need to define a tree decomposition for the graphs $\CFI(G_{n,n})$ and $\twCFI(G_{n,n})$.
 Fix $n \geq 2$ and let $(T,\beta)$ be the tree decomposition described in Lemma \ref{la:tree-decomposition-subdivided-grid} for $G_{n,n}^{+}$.
 Now a tree decomposition $(T',\beta')$ for the graphs $\CFI(G_{n,n})$ and $\twCFI(G_{n,n})$ can be obtained as follows.
 For each $t \in V(T)$ such that $\beta(t) \cap V(G_{n,n}) = \emptyset$, it also holds that~$t \in V(T')$ and
 \[\beta'(t) = \big\{a(v,w),b(v,w) \ \big\vert \ (v,w) \in \beta(t)\big\}.\]
 Note that $|\beta'(t)| = 2 \cdot |\beta(t)|$.
 Also, for every $t_1t_2 \in E(T)$ with $\beta(t_i) \cap V(G_{n,n}) = \emptyset$, there is an edge $t_1t_2 \in E(T')$.
 
 Otherwise, $|\beta(t) \cap V(G_{n,n})| = 1$ and $\beta(t) = E(v) \cup \{v\}$ for a certain $v \in V(G_{n,n})$.
 Also,~$t$ is a leaf of $T$ and $\beta(s) \cap V(G_{n,n}) = \emptyset$ for the unique $s \in V(T)$ with $st \in E(T)$.
 For every $A \subseteq E(v)$ such that $|A|$ is even, there is a vertex $t_A \in V(T')$.
 We define
 \[\beta'(t_A) \coloneqq \big\{m_A\big\} \cup \big\{a(v,w),b(v,w) \ \big\vert \ vw \in E(G_{n,n})\big\}.\]
 Note that $|\beta'(t_A)| \leq 9$, since $\deg_{G_{n,n}}(v) \leq 4$ for every $v \in V(G_{n,n})$.
 Also, there are edges $\{t_A,s\} \in E(T')$ for every $A \subseteq E(v)$ such that $|A|$ is even.
 It is easy to check that $(T',\beta')$ is a tree decomposition for the graphs $\CFI(G_{n,n})$ and $\twCFI(G_{n,n})$.
 Also,
 \[\width(T',\beta') \leq \max\{9,2(\width(T,\beta)+1)\} - 1 \leq \max\{9,2(n+3)\} - 1 = 2n+5.\]
 This concludes the proof.
\end{proof}

Combining Lemma \ref{la:tree-width-cfi-grid} and Theorem \ref{thm:cfi-equivalence-tree-width}, we almost obtain the lower bound stated in Theorem \ref{thm:wl-treewidth-bounds}.
By slightly modifying the constructions presented above, we are able to improve on the upper bound of the treewidth of certain Cai-Fürer- Immerman graphs.
Towards this end, let $G_{n,n}'$ be the graph obtained from $G_{n,n}$ by replacing each vertex of degree $4$ with two vertices of degree $3$ as indicated in Figure \ref{fig:modified-grid}.
Formally,
\[V(G_{n,n}') \coloneqq \big(\{1,n\} \times [n]\big) \cup \big([n] \times \{1,n\}\big) \cup \big\{(i,j,1),(i,j,4) \, \big\vert \, i,j \in \{2,\dots,n-1\}\big\}\]
and
\begin{align*}
 E(G_{n,n}') \coloneqq\;\;&\{(i,j)(i',j') \mid (i = i' \wedge |j - j'| = 1) \vee (j = j' \wedge |i-i'| = 1)\}\\
                    \cup\;&\{(1,i)(2,i,1) \mid i \in \{2,\dots,n-1\}\}\\
                    \cup\;&\{(i,1)(i,2,1) \mid i \in \{2,\dots,n-1\}\}\\
                    \cup\;&\{(n,i)(n-1,i,4) \mid i \in \{2,\dots,n-1\}\}\\
                    \cup\;&\{(i,n)(i,n-1,4) \mid i \in \{2,\dots,n-1\}\}\\
                    \cup\;&\{(i,j,1)(i,j,4) \mid i,j \in \{2,\dots,n-1\}\}\\
                    \cup\;&\{(i,j,4)(i+1,j,1) \mid i \in \{2,\dots,n-2\},j \in \{2,\dots,n-1\}\}\\
                    \cup\;&\{(i,j,4)(i,j+1,1) \mid i \in \{2,\dots,n-1\},j \in \{2,\dots,n-2\}\}.
\end{align*}
Note that $G_{n,n}$ is a minor of $G_{n,n}'$ and thus, $\tw(G_{n,n}') \geq n$.
Also, let $G_{n,n}^*$ be the graph obtained from $G_{n,n}'$ by replacing each edge with a path of length $3$.
Formally, we construct the graph $G_{n,n}^*$ from $G_{n,n}^+$ as follows.
For $i,j \in \{2, \dots, n-1\}$, let $V_{i,j} \coloneqq N_{G_{n,n}^{+}}[(i,j)] \subseteq V(G_{n,n}^{+})$.
Note that
\[V_{i,j} = \Big\{(i,j), \big((i,j),(i-1,j)\big),\big((i,j),(i+1,j)\big),\big((i,j),(i,j-1)\big),\big((i,j),(i,j+1)\big)\Big\}.\]
We replace the graph $G_{n,n}^{+}[V_{i,j}]$ with the modified graph $\widehat{G}_{i,j}$ with
\[V(\widehat{G}_{i,j}) = \big(V_{i,j} \setminus \{(i,j)\}\big) \cup \big\{(i,j,d) \, \vert \, d \in [4]\big\}\]
and edges of the following forms
\begin{itemize}
 \item $((i,j),(i-1,j))(i,j,1)$,
 \item $((i,j),(i,j-1))(i,j,1)$,
 \item $(i,j,1)(i,j,2)$,
 \item $(i,j,2)(i,j,3)$,
 \item $(i,j,3)(i,j,4)$,
 \item $((i,j),(i+1,j))(i,j,4)$, and
 \item $((i,j),(i,j+1))(i,j,4)$.
\end{itemize}
The graph resulting from applying the modification for all $i,j \in \{2, \dots, n-1\}$ is $G_{n,n}^*$.

\begin{figure}
 \centering
 \scalebox{0.8}{
 \begin{tikzpicture}
  \draw (0.0,0) edge (0.0,4.8);
  \draw (1.6,0) edge (1.6,4.8);
  \draw (3.2,0) edge (3.2,4.8);
  \draw (4.8,0) edge (4.8,4.8);
  
  \draw (0,0.0) edge (4.8,0.0);
  \draw (0,1.6) edge (4.8,1.6);
  \draw (0,3.2) edge (4.8,3.2);
  \draw (0,4.8) edge (4.8,4.8);
  
  \foreach \i in {0,1,2,3}{
   \foreach \j in {0,1,2,3}{
    \node[tinyvertex] (v-\i-\j) at (1.6*\i,1.6*\j) {};
   }
  }
  
 \end{tikzpicture}
 }
 \hspace*{15pt}
 \scalebox{0.8}{
 \begin{tikzpicture}
  \draw (0.0,0) edge (0.0,4.8);
  \draw (4.8,0) edge (4.8,4.8);
  
  \draw (0,0.0) edge (4.8,0.0);
  \draw (0,4.8) edge (4.8,4.8);
  
  \foreach \i in {0,3}{
   \foreach \j in {0,1,2,3}{
    \node[tinyvertex] (v-\i-\j) at (1.6*\i,1.6*\j) {};
   }
  }
  \foreach \i in {0,1,2,3}{
   \foreach \j in {0,3}{
    \node[tinyvertex] (v-\i-\j) at (1.6*\i,1.6*\j) {};
   }
  }
  \foreach \i in {1,2}{
   \foreach \j in {1,2}{
    \node[tinyvertex] (v-\i-\j-1) at ($(1.6*\i,1.6*\j)-(0.2,-0.2)$) {};
    \node[tinyvertex] (v-\i-\j-4) at ($(1.6*\i,1.6*\j)+(0.2,-0.2)$) {};
    \draw (v-\i-\j-1) edge (v-\i-\j-4);
   }
  }
  
  \path
   (v-1-3) edge (v-1-2-1)
   (v-2-3) edge (v-2-2-1)
   (v-0-1) edge (v-1-1-1)
   (v-0-2) edge (v-1-2-1)
   (v-1-2-4) edge (v-1-1-1)
   (v-1-2-4) edge (v-2-2-1)
   (v-1-1-4) edge (v-2-1-1)
   (v-1-1-4) edge (v-1-0)
   (v-2-2-4) edge (v-2-1-1)
   (v-2-2-4) edge (v-3-2)
   (v-2-1-4) edge (v-2-0)
   (v-2-1-4) edge (v-3-1);
 \end{tikzpicture}
 }
 \caption{Visualization of the graphs $G_{4,4}$ on the left and $G_{4,4}'$ on the right.}
 \label{fig:modified-grid}
\end{figure}

\begin{lemma}
 \label{la:tree-decomposition-subdivided-modified-grid}
 Let $n \geq 2$.
 Then there is a tree decomposition $(T,\beta)$ of $G_{n,n}^{*}$ of width $n+1$ such that
 \begin{enumerate}
  \item $|\beta(t) \cap V(G_{n,n}')| \leq 1$ for every $t \in V(T)$, and
  \item if $|\beta(t) \cap V(G_{n,n}')| = 1$, then there exists a $v \in V(G_{n,n}')$ such that $\beta(t) = N_{G_{n,n}^{*}}[v]$.
   In this case, $t$ is a leaf of $T$, and $\beta(s) \cap V(G_{n,n}) = \emptyset$ for the unique $s \in V(T)$ with $st \in E(T)$.
 \end{enumerate}
\end{lemma}

\begin{proof}
 We only describe how to modify the decomposition from Lemma \ref{la:tree-decomposition-subdivided-grid}.
 We replace the nodes $t_{i,j}^{A}$ and $t_{i,j}^{D}$ with five nodes $t_{i,j,1}^{A}$, $t_{i,j,23}^{A}$, $t_{i,j,4}^{A}$, $t_{i,j,1}^{D}$, $t_{i,j,4}^{D}$.
 We define
  \begin{align*}
   \beta(t_{i,j,1}^A) &\coloneqq \big(\beta(t_{i,j}^{A}) \setminus V_{i,j}\big) \cup \Big\{\big((i,j),(i-1,j)\big),\big((i,j),(i,j-1)\big),(i,j,2)\Big\},
 \\\beta(t_{i,j,4}^A) &\coloneqq \big(\beta(t_{i,j}^{A}) \setminus V_{i,j}\big) \cup \Big\{\big((i,j),(i+1,j)\big),\big((i,j),(i,j+1)\big),(i,j,3)\Big\},
 \\\beta(t_{i,j,23}^A) &\coloneqq \big(\beta(t_{i,j}^{A}) \setminus V_{i,j}\big) \cup \Big\{(i,j,2),(i,j,3)\Big\},
 \\\beta(t_{i,j,1}^D) &\coloneqq \Big\{(i,j,1),\big((i,j),(i-1,j)\big),\big((i,j),(i,j-1)\big),(i,j,2)\Big\},
 \\\beta(t_{i,j,4}^D) &\coloneqq \Big\{(i,j,4),\big((i,j),(i+1,j)\big),\big((i,j),(i,j+1)\big),(i,j,3)\Big\}.
 \end{align*}
 Also, we replace the edges $t_{i,j}^{A}t_{i,j}^{D}$, $t_{i,j}^{A}t_{i,j}^{B}$, $t_{i,j}^{A}t_{i-1,j}^{B}$ with the edges
 \begin{itemize}
  \item $t_{i,j,1}^{A}t_{i,j,1}^{D}$,
  \item $t_{i,j,2}^{A}t_{i,j,4}^{D}$,
  \item $t_{i,j,1}^{A}t_{i,j,23}^{A}$,
  \item $t_{i,j,23}^{A}t_{i,j,4}^{A}$,
  \item $t_{i,j,1}^{A}t_{i-1,j}^{B}$, and
  \item $t_{i,j,4}^{A}t_{i,j}^{B}$.\qedhere
 \end{itemize}
\end{proof}

\begin{lemma}
 \label{la:tree-width-cfi-grid-modified}
 Let $n \geq 2$. Then $\tw(\CFI(G_{n,n}')) \leq 2n + 3$ and $\tw(\twCFI(G_{n,n}')) \leq 2n + 3$.
\end{lemma}

The proof is literally the same as for Lemma \ref{la:tree-width-cfi-grid} using Lemma \ref{la:tree-decomposition-subdivided-modified-grid} instead of Lemma \ref{la:tree-decomposition-subdivided-grid}.

\begin{theorem}
 \label{thm:wl-treewidth-lower-bound}
 For every $k \geq 2$, there are non-isomorphic graphs $G_k$ and $H_k$ of treewidth at most $2k+5$ such that $G_k \simeq_k H_k$.
\end{theorem}

\begin{proof}
 Let $G_k \coloneqq \CFI(G_{k+1,k+1}')$ and $H_k \coloneqq \twCFI(G'_{k+1,k+1})$. 
 Then the statement follows from Theorem \ref{thm:cfi-equivalence-tree-width} and Lemma \ref{la:tree-width-cfi-grid-modified}.
\end{proof}

\begin{remark}
 The graphs $G_k$ and $H_k$ constructed in Theorem \ref{thm:wl-treewidth-lower-bound} are vertex-colored.
 It is possible to turn them into uncolored graphs by using simple gadgets to encode the vertex colors (see, e.g., \cite{BoothC79}).
 We remark that these gadgets can be constructed in such a way that the treewidth of $G_k$ and $H_k$ as well as indistinguishability by $k$-WL is preserved.
 Hence, Theorem \ref{thm:wl-treewidth-lower-bound} remains valid even when we restrict ourselves to uncolored graphs.
\end{remark}

As a corollary from Theorems \ref{thm:wl-treewidth-upper-bound} and \ref{thm:wl-treewidth-lower-bound}, we obtain Theorem \ref{thm:wl-treewidth-bounds}.

\section{Conclusion}
\label{sec:conclusion}

We have proved that for $k \geq 2$, the algorithm $k$-WL implicitly computes the decomposition of its input graph into its triconnected components.
As a by-product, we found that every connected constituent graph of an association scheme is either a cycle or $3$-connected.

We have applied this insight to improve on the upper bound on a parametrization of the WL dimension of a graph by its treewidth 
and have also provided a lower bound that is asymptotically only a factor of $2$ away from the upper bound.
Naturally, it would be nice to further close the gap between the upper and the lower bound.
We remark that some improvements on the lower bounds are possible for small values of $k$ as indicated in Table \ref{tab:wl-treewidth-bounds}.
For $k \geq 2$, the WL dimension of the class $\mathcal{T}_k$ is at least $2$, since $1$-WL does not distinguish between the cycle $C_6$ of length $6$ and the disjoint union of two triangles $2K_3$.
Observe that $\tw(C_6) = \tw(2K_3) = 2$.
Applying the Cai-Fürer-Immerman construction to the complete graph $K_4$ and calculating the treewidth of the resulting graph by computer \cite{Tamaki19}, we also found that $2$-WL does not identify every graph of treewidth $8$.

\begin{table}[htbp]
 \caption{Upper and lower bounds on the WL dimension of graphs of treewidth $k$.
 The bounds marked with an asterix are not derived from Theorem \ref{thm:wl-treewidth-bounds}.}
 \label{tab:wl-treewidth-bounds}
 \begin{center}
 \begin{tabular}{|c|c|c||c|c|c|}
   \hline
   $k$ & lower bound & upper bound & $k$ & lower bound & upper bound\\
   \hline
   $1$ & $1^*$ & $\hspace{1.6mm}1^*$ & $7$ & $2$ & $7$\\
   $2$ & $2^*$ & $2$ & $8$ & $\hspace{1.6mm}3^*$ & $8$\\
   $3$ & $2^*$ & $3$ & $9$ & $3$ & $9$\\
   $4$ & $2^*$ & $4$ & $10$ & $3$ & $10$\\
   $5$ & $2^*$ & $5$ & $11$ & $4$ & $11$\\
   $6$ & $2^*$ & $6$ & $12$ & $4$ & $12$\\
   \hline
 \end{tabular}
 \end{center}
\end{table}

It is an interesting open problem to determine $\dimWL(\mathcal{T}_3)$.
A potential approach for improving the upper bound could be to argue that $2$-WL implicitly computes a tree decomposition with certain nice properties.
Observe that, by Theorem \ref{thm:reduction-to-3-connected}, it essentially suffices to solve the problem for $3$-connected graphs. 

Indeed, more generally, a natural use case of our results may be determining the WL dimension of certain graph classes that satisfy the requirements of Theorem \ref{thm:reduction-to-3-connected}.
For example, we conjecture that $2$-WL identifies every planar graph.
Using the results of this paper, it essentially suffices to show this for $3$-connected planar graphs.

\bibliography{literature}

\appendix

\section{Proof of Theorem \ref{thm:implicitly-computes-decomposition}}

In this section, we prove Theorem \ref{thm:implicitly-computes-decomposition}. The theorem states that for every graph $G$, there is a rooted tree decomposition of $G$ with the following three properties.
\begin{enumerate}
 \item $|\beta(s) \cap \beta(t)| \leq 2$ for all $(s,t) \in E(T)$.
 \item For every $t \in V(T)$, the torso $G[[\beta(t)]]$ is $3$-connected or a cycle.
 \item $2$-WL implicitly computes $(T,\beta)$.
\end{enumerate}

In fact, we show that the decomposition $(T,\beta)$ and the color sets to prove that $2$-WL implicitly computes $(T,\beta)$ can be chosen such that $\mathcal{C}_{\sf bags}$ and $\mathcal{C}_{\sf edges}$ only depend on the equivalence class of the graph $G$ with respect to $2$-WL.
More precisely, for a second graph~$G'$ which $2$-WL does not distinguish from $G$, we show that we can take the same sets $\mathcal{C}_{\sf bags}$ and $\mathcal{C}_{\sf edges}$ as for $G$. 

We fall back on the concepts and notation described in \cite[Section 3]{kieponschwei19}, which we recapitulate for completeness.

For a graph~$G$, define the set~$P(G)$ to consist of the pairs~$(S,K)$, where~$S$ is a separator of~$G$ of minimum cardinality and~$K$ is the vertex set of a connected component of~$G-S$. 
Let~$P_0(G)$ be the set of minimal elements of~$P(G)$ with respect to the following (well-defined) partial order $\leq$ on~$P(G)$:
\[(S,K) \leq (S',K') \iff K \subseteq K'.\]

The following lemma summarizes useful properties of $P(G)$ when certain bounds on the connectivity and the vertex degrees hold in $G$.

\begin{lemma}[{\cite[Lemma 3]{kieponschwei19}}]
 \label{lem:minimum:degree:implies:disjoin:min:seps}
 For $k \in \mathbb{N}$, let~$G$ be a graph that is not~$(k+1)$-connected and has minimum degree at least~$\frac{3k-1}{2}$. 
 \begin{enumerate}
  \item \label{item:one:min:comp:trivial:int:contained} If~$(S,K)\in P_0(G)$ and~$(S',K')\in P(G)$ are distinct, then~$K\subseteq K'$ or~$(K\cup S)\cap (K'\cup S') = S\cap S'$.
  \item \label{item:min:comp:trivial:int} If~$(S,K),(S',K') \in P_0(G)$ are distinct, then~$(K\cup S)\cap (K'\cup S') = S\cap S'$.
  \item \label{item:char:min:comp} A pair~$(S,K)\in P(G)$ is contained in~$P_0(G)$ if and only if there is no separator~$S'$ of~$G$ of minimum cardinality with~$S'\cap K \neq \emptyset$.
 \end{enumerate}
\end{lemma}

We intend to prove Theorem \ref{thm:implicitly-computes-decomposition} via an induction on the order of~$G$. To this end, we need a method to remove the vertices appearing in the second component of pairs in~$P_0(G)$ from the graph in an isomorphism-invariant fashion. 

For a set~$S \subseteq V(G)$, define~$G^S_{\top}$ as the graph on the vertex set
\[V'= S \cup \bigcup\limits_{(S,K)\in P_0(G)} K\]
with edge set~$E' = E(G[V']) \cup \{\{s,s'\}\mid s,s'\in S \text{ and }s\neq s'\}$.

These concepts also play a crucial role in the proof of Theorem \ref{thm:reduction-to-3-connected}.
In order to allow for an induction, we use the initial coloring $\lambda$ of $G$ to define also colorings of the graphs $G^S_\top$.

For an arc coloring~$\lambda$ of~$G$, we obtain an arc coloring~$\lambda^S_\top$ for~$G^S_{\top}$ as follows.
\[\lambda^S_\top(v_1,v_2)  \coloneqq
\begin{cases}
(0,0)& \text{if } \{v_1, v_2\}\subseteq  S \text{ and } \{v_1, v_2\} \notin E(G)\\ 
(\lambda(v_1,v_2),1)& \text{if } \{v_1, v_2\}\subseteq  S \text{ and } \{v_1, v_2\} \in E(G) \\ 
(\lambda(v_1,v_2),2) & \text{otherwise.} 
\end{cases}
\]

If~$G$ is a vertex-colored graph with vertex coloring~$\lambda'$, in order to obtain a coloring of~$G^S_{\top}$, we define an arc coloring~$\lambda$ as~$\lambda (v_1,v_2) \coloneqq \lambda'(v_1)$ and let~$\lambda^S_\top$ be as above.

For~$(S,K) \in P_0(G)$, we also let~$G^{(S,K)}_{\top}\coloneqq G^S_{\top}[S \cup K]$. We let~$\lambda^{(S,K)}_\top$ be the restriction of~$\lambda^S_\top$ to pairs~$(v_1,v_2)$ for which~$\{v_1,v_2\}\subseteq S\cup K$.

Finally, we define~$G_{\bot}$ to be the graph with vertex set~\[V_{\bot} \coloneqq V(G) \setminus \bigg(\bigcup_{(S,K)\in P_0(G)}K\bigg)\]
and edge set 
\[E_{\bot} \coloneqq E(G[V_{\bot}]) \, \cup \, \big\{ \{s_1,s_2\} \mid \exists (S,K)\in P_0(G) \text{ s.t.\ } s_1,s_2\in S, s_1\neq s_2\big\}.\]

Given an arc coloring~$\lambda$ of a graph~$G$ that is not $3$-connected, we define an arc coloring~$\lambda_\bot$ of~$G_{\bot}$ as follows. Assume that~$v_1, v_2 \in V(G_\bot)$ (possibly with~$v_1 = v_2$). Let~$S \coloneqq \{v_1,v_2\}$. If~$S$ is a~$2$-separator of~$G$, but~$S \notin E(G)$, we set
\begin{align*}
\lambda_\bot(v_1,v_2) &\coloneqq 
\left(0, \ISOTYPE\left(\left(G^{S}_{\top},\lambda^{S}_\top\right)_{(v_1,v_2)}\right)\right) .
\intertext{Furthermore, if~$v_1 = v_2$ or if~$\{v_1,v_2\} \in E(G)$, we set}
\lambda_\bot(v_1,v_2) &\coloneqq \left(\lambda(v_1,v_2), \ISOTYPE\left(\left(G^S_{\top},\lambda^S_\top\right)_{(v_1,v_2)}\right)\right), 
\end{align*}
where~$ \ISOTYPE((G^S_{\top},\lambda^S_\top)_{(v_1,v_2)})$ is the isomorphism class of the colored graph $(G^S_{\top},\lambda^S_\top)_{(v_1,v_2)}$ which is obtained from~$(G^S_\top,\lambda^S_\top)$ by individualizing~$v_1$ and~$v_2$. Therefore, the graphs $(G^{\{v_1,v_2\}}_{\top},\lambda^{\{v_1,v_2\}}_\top)_{(v_1,v_2)}$ and~$(G'^{\{v'_1,v'_2\}}_{\top},\lambda'^{\{v'_1,v'_2\}}_\top)_{(v'_1,v'_2)}$ have the same isomorphism type if and only if there is an isomorphism from the first to the second graph mapping~$v_1$ to~$v'_1$ and~$v_2$ to~$v'_2$. 

The following lemma guarantees that the designed coloring of the $\bot$-graph maintains all isomorphism-invariant information about the removed connected components.

\begin{lemma}[{\cite[Lemma 4]{kieponschwei19}}]
Let $k \in \{1,2\}$. If~$G$ and~$G'$ are~$k$-connected graphs that are not~$(k+1)$-connected and that are of minimum degree at least~$\frac{3k-1}{2}$ with arc colorings~$\lambda$ and~$\lambda'$, respectively, then
it holds that $(G, \lambda) \cong (G', \lambda')$ if and only if $(G_{\bot}, \lambda_\bot) \cong (G'_{\bot},\lambda'_\bot)$. 
\end{lemma} 

We define the rooted tree decomposition from Theorem \ref{thm:implicitly-computes-decomposition} by induction on the order of $G$.
If $G$ is $3$-connected or a cycle (which includes the case $|V(G)| = 1$), the graph $T$ consists of a single node $t$ and we let $\beta(t) \coloneqq V(G)$.
Clearly, $2$-WL implicitly computes $(T,\beta)$ and the sets $\mathcal{C}_{\sf bags}$ and $\mathcal{C}_{\sf edges}$ are canonical for every $2$-WL-equivalence class of graphs ($\mathcal{C}_{\sf bags}$ consists of all colors in the image of $\chi_G$ and $\mathcal{C}_{\sf edges}$ is empty).

If $G$ is not $3$-connected, $G$ contains at least one separator of size at most $2$. To construct~$(T,\beta)$ and show that $2$-WL implicitly computes it, we distinguish cases according to the vertex degrees and connectivity in $G$.

First suppose that $G$ is not $2$-connected, i.e.\ $G$ contains a cut vertex. Since $G_\bot$ is smaller than $G$, there is a rooted tree decomposition $(T_\bot, \beta_\bot)$ for it with the desired properties by the induction hypothesis. To construct $(T, \beta)$ from $(T_\bot, \beta_\bot)$, consider for an $(\{s\},K) \in P_0(G)$ the graph $G^{(\{s\},K)}_\top$. Again, by the induction hypothesis, there is a rooted tree decomposition $(T^{(\{s\},K)}_\top, \beta^{(\{s\},K)}_\top)$ for $G^{(\{s\},K)}_\top$ with the desired properties. 

In $T_\top^{(\{s\},K)}$, pick as $t_s \coloneqq t^{(\{s\},K)}$ the node $t$ with $s \in \beta^{(\{s\},K)}_\top(t)$ that is closest to the root and orient all edges away from $t_s$ so as to change the root of the tree. Call the new rooted tree $\widetilde{T}^{(\{s\},K)}_\top$. We let 
\[V(T) \coloneqq V(T_\bot) \cup \bigcup_{(\{s\},K) \in P_0(G)} V(\widetilde{T}^{(\{s\},K)}_\top)\]
and set 
\begin{align*}
  E(T) \coloneqq E(T_\bot) &\cup \bigcup_{(\{s\},K) \in P_0(G)} E(\widetilde{T}^{(\{s\},K)}_\top) \\&\cup \{(t,t^{(\{s\},K)}) \mid (\{s\},K) \in P_0(G) \text{ and $t$ is the unique node of $V(T_\bot)$}
  \\&\text{\phantom{$\cup \{(t,t^{(\{s\},K)}) \mid $ } that is closest to the root and satisfies $s \in \beta_\bot(t)$}\}.
\end{align*}
Finally, set
\[\beta(t) \coloneqq \begin{cases} 
\beta_\bot(t) &\text{ if $t \in V(T_\bot)$} \\ \beta_\top^{(\{s\},K)}(t) &\text{ if $t \in V(G^{(\{s\},K)}_\top)$}.
\end{cases}\]
Clearly, the rooted tree decomposition $(T,\beta)$ satisfies the first two conditions from Theorem~\ref{thm:implicitly-computes-decomposition}. We now show that $2$-WL implicitly computes it.

By \cite[Lemma 8]{kieponschwei19}, for all $u \in V(G_\bot)$ and $v \in V(G) \setminus V(G_\bot)$, the inequality~$\chi_G(u,u) \neq \chi_G(v,v)$ holds. We conclude that $\chi_G|_{(V(G_\bot))^2}$ is a stable coloring of $G_\bot$ (with respect to $2$-WL). Hence, by induction, since $G_\bot$ is smaller than $G$, there are sets $\mathcal{C}^\bot_{\sf bags}, \mathcal{C}^\bot_{\sf edges} \subseteq \{\chi_G(v,w) \mid v,w \in V(G)\}$ such that
\begin{align*}
\chi_G(v,w) \in \mathcal{C}^\bot_{\sf bags} \;\;&\Leftrightarrow\;\; w \in \beta(t_v)
\intertext{and}
\chi_G(v,w) \in \mathcal{C}^\bot_{\sf edges} \;\;&\Leftrightarrow\;\; (t_v,t_w) \in E(T)
 \end{align*}
 holds for all $v,w \in V(G_\bot) \subseteq V(G)$. 

Also, $2$-WL knows the set of cut vertices appearing in the first components of elements in~$P_0(G)$ (because they form the neighborhood of $V(G) \setminus V(G_\bot)$ in $G$).
Therefore, the color set computed by $2$-WL for vertices contained in elements of $P_0(G)$ is disjoint from the set of colors for other vertices.
Theorem 6 in \cite{kieponschwei19} says that pairs of vertices that share a $2$-connected component obtain different colors with respect to $2$-WL than other pairs of vertices.
Thus, we can use the theorem to deduce that for every $(\{s\},K) \in P_0(G)$, the coloring $\chi_G|_{(V(G^{(\{s\},K)}_\top))^2}$ is a stable coloring of $G^{(\{s\},K)}_\top$ (with respect to $2$-WL).
Hence, using the induction hypothesis, for every $(\{s\},K) \in P_0(G)$, there are sets $\mathcal{C}^{(\{s\},K)}_{\sf bags}, \mathcal{C}^{(\{s\},K)}_{\sf edges} \subseteq \{\chi_G(v,w) \mid v,w \in V(G)\}$ such that
\begin{align*}
\chi_G(v,w) \in \mathcal{C}^{(\{s\},K)}_{\sf bags} \;\;&\Leftrightarrow\;\; w \in \beta_\top^{(\{s\},K)}(t_v)
\intertext{and}
\chi_G(v,w) \in \mathcal{C}^{(\{s\},K)}_{\sf edges} \;\;&\Leftrightarrow\;\; (t_v,t_w) \in E(T_\top^{(\{s\},K)})
\end{align*}
hold for all $v,w \in V(G^{(\{s\},K)}_\top) \subseteq V(G)$.

To obtain $\widetilde{T}_\top^{(\{s\},K)}$, the orientation of an edge $(t_v,t_w)$ in $T_\top^{(\{s\},K)}$ must be inverted if there is a path from $t_w$ to $t_s$ in $T_\top^{(\{s\},K)}$, which is a property that $2$-WL knows in $G$. 
We conclude that, again by the induction hypothesis, there are $\widetilde{\mathcal{C}}^{(\{s\},K)}_{\sf bags}, \widetilde{\mathcal{C}}^{(\{s\},K)}_{\sf edges} \subseteq \{\chi_G(v,w) \mid v,w \in V(G)\}$ such that
\begin{align*}
\chi_G(v,w) \in \widetilde{\mathcal{C}}^{(\{s\},K)}_{\sf bags} \;\; &\Leftrightarrow\;\; w \in \beta(t_v)
\intertext{and}
\chi_G(v,w) \in \widetilde{\mathcal{C}}^{(\{s\},K)}_{\sf edges} \;\; &\Leftrightarrow\;\; (t_v,t_w) \in E(T)
\end{align*}
hold for all $v,w \in V(G^{(\{s\},K)}_\top) \subseteq V(G)$. 

Furthermore, we know also by the induction hypothesis that these color sets can be chosen consistently among the different $(\{s\},K)$.
That is, for $(\{s\},K) \neq (\{s'\},K')$ with $(\{s\},K),(\{s'\},K') \in P_0(G)$, either the corresponding sets $\mathcal{C}_{\sf bags}$ are equal and the same holds for the sets $\mathcal{C}_{\sf edges}$, or 
\[\big\{\chi_G(v,w) \ \big\vert \ v,w \in K \cup \{s\}, v \neq w\big\} \cap \big\{\chi_G(v',w') \ \big\vert \ v',w' \in K' \cup \{s'\}, v' \neq w'\big\} = \emptyset.\] 

Altogether, we can let 
\[\mathcal{C}_{\sf bags} \coloneqq \mathcal{C}^\bot_{\sf bags} \cup \bigcup_{(\{s\},K) \in P_0(G)} \widetilde{\mathcal{C}}^{(\{s\},K)}_{\sf bags}.\]

For the arc colors, we can also extend $\mathcal{C}^\bot_{\sf edges}$ by the sets $\widetilde{\mathcal{C}}^{(\{s\},K)}_{\sf edges}$.
Then it remains to consider the edges in $T$ of the form $(t,t^{(\{s\},K)})$. For $v \in V(G_\bot)$ and $w \in V(G) \setminus V(G_\bot)$, it never holds that $(t_w,t_v) \in E(T)$. Moreover,
\begin{align*}
 (t_v,t_w) \in E(T) \;\;\Leftrightarrow\;\; \exists u \big(t_u = t_v \wedge u \in \beta(t_w) \big).
\end{align*}
So we can express the existence of an edge between $t_v$ and $t_w$ by a formula with $3$ variables where the satisfaction of each atomic formula can be checked by a membership query for a set of $2$-WL color classes. 
Hence, using the well-known equivalence between $2$-WL and the $3$-variable fragment of first-order logic with counting quantifiers (see, e.g., \cite{CaiFI92,Grohe17,IL90}),
we conclude the existence of a set of colors $\mathcal{C}_{\sf edges}$ with the desired properties. Note that $\mathcal{C}_{\sf bags}$ and $\mathcal{C}_{\sf edges}$ are solely determined by the equivalence class of the graph $G$ with respect to~$2$-WL. 
This completes the analysis for the case that $G$ is not $2$-connected.

\medskip

Now assume that $G$ is $2$-connected, but not $3$-connected, that is, every minimal separator of $G$ has size $2$.
Similar to the analysis in \cite{kieponschwei19}, we need to distinguish two further cases depending on the minimum degree in $G$.
More precisely, in order to be able to apply Lemma~\ref{lem:minimum:degree:implies:disjoin:min:seps}, we require that $G$ has minimum degree $3$.
Hence, we first cover the case that $G$ has minimum degree $2$. (Note that $G$ does not contain a vertex of degree $1$ since $G$ is $2$-connected.)

Let $D \coloneqq \{v \in V(G) \mid \deg(v) = 2\}$ and, overwriting $G_\bot$ for this case, define~$G_\bot \coloneqq G[[V(G) - D]]$ to be the torso on the remaining vertices.
By the induction hypothesis, there is a tree decomposition $(T_\bot,\beta_\bot)$ for the graph $G_\bot$ with the desired properties.
We construct~$(T,\beta)$ from $(T_\bot,\beta_\bot)$ as follows.
For each connected component $C$ of $G[D]$, we introduce a new node~$t_C$ with $\beta(t_C) \coloneqq C \cup N_G(C)$ and insert the edge $(t'_C,t_C)$, where~$t_C' \in V(T_\bot)$ is the unique node that is closest to the root and satisfies $N_G(C) \subseteq \beta_\bot(t_C')$. Observe that $|N_G(C)| = 2$ and the corresponding pair of vertices in $N_G(C)$ forms an edge in $G_\bot$, which guarantees the existence of the node $t_C'$.

Clearly, the decomposition $(T,\beta)$ satisfies the first two conditions of Theorem \ref{thm:implicitly-computes-decomposition}.
Hence, it remains to argue that $2$-WL implicitly computes $(T,\beta)$.
First observe that $\chi_G(v,v) \neq \chi_G(u,u)$ holds for all $v \in D$ and $u \in V(G_\bot)$.
Since $2$-WL knows which edges are inserted to obtain the torso $G[[V(G) - D]]$, the coloring $\chi_G|_{(V(G_\bot))^2}$ is a stable coloring of $G_\bot$ with respect to~$2$-WL. Hence, by the induction hypothesis, there are $\mathcal{C}^\bot_{\sf bags}, \mathcal{C}^\bot_{\sf edges} \subseteq \{\chi_G(v,w) \mid v,w \in V(G)\}$ such that
\begin{align*}
\chi_G(v,w) \in \mathcal{C}^\bot_{\sf bags} \;\;&\Leftrightarrow\;\; w \in \beta(t_v)
\intertext{and}
\chi_G(v,w) \in \mathcal{C}^\bot_{\sf edges} \;\;&\Leftrightarrow\;\; (t_v,t_w) \in E(T)
\intertext{hold for all $v,w \in V(G_\bot)$.
Moreover, there is a set $\mathcal{C}^\top_{\sf bags} \subseteq \{\chi_G(v,w) \mid v,w \in V(G)\}$ such that}
\chi_G(v,w) \in \mathcal{C}^\top_{\sf bags} \;\;&\Leftrightarrow\;\; v \in D \wedge w \in C_v \cup N_G(C_v)
\end{align*}
for all $v,w \in V(G)$, where $C_v$ denotes the connected component of $G[D]$ that contains $v$ (assuming $v \in D$).
Hence, setting $\mathcal{C}_{\sf bags} \coloneqq \mathcal{C}^\bot_{\sf bags} \cup \mathcal{C}^\top_{\sf bags}$ gives the desired set of colors for the bags of $(T,\beta)$.

So it remains to express the edges $(t'_C,t_C)$ for connected components $C$ of $G[D]$ via membership in a set of color classes with respect to $2$-WL.
Let $v \in V(G_\bot) = V(G) \setminus D$ and~$w \in D$.
Then
\begin{align*}
 (t_v,t_w) \in E(T) \;\;\Leftrightarrow\;\; \exists u \Big(t_u = t_v \wedge \exists v \big(\{u,v\} = N_G(C_w) \wedge v \in \beta(t_u)\big)\Big).
\end{align*}
As in the first case, we conclude the existence of a set of colors $\mathcal{C}_{\sf edges}$ with the desired properties. Note that the sets $\mathcal{C}_{\sf bags}$ and $\mathcal{C}_{\sf edges}$ are solely determined by the equivalence class of the graph $G$ with respect to $2$-WL. 
This completes the case that $G$ is $2$-connected and has minimum degree $2$.

\medskip

Now suppose that $G$ is $2$-connected, but not $3$-connected, and has minimum degree at least $3$. We first define the rooted tree decomposition $(T,\beta)$ described in Theorem \ref{thm:implicitly-computes-decomposition}. Since~$G_\bot$ is smaller than $G$, there is a rooted tree decomposition $(T_\bot, \beta_\bot)$ for it with the desired properties by the induction hypothesis. To construct $(T, \beta)$ from $(T_\bot, \beta_\bot)$, we introduce for each $(S,K) \in P_0(G)$ a node $t_{(S,K)}$. Let $M$ be the collection of all these $t_{(S,K)}$.
Now we define $V(T) \coloneqq V(T_\bot) \uplus M$ and
\begin{align*}
  E(T) \coloneqq E(T_\bot) &\cup \big\{(t,t_{(S,K)}) \ \big\vert \ t_{(S,K)} \in M \text{ and $t$ is the unique node of $V(T_\bot)$}
  \\&\text{\phantom{$\cup \{(t,t^{(\{s\},K)}) \mid$}that is closest to the root and satisfies $S \subseteq \beta_\bot(t)$}\big\}
\end{align*}
and set
\[\beta(t) \coloneqq \begin{cases} 
\beta_\bot(t) &\text{ if $t \in V(T_\bot)$} \\ S \cup K &\text{ if $t = t_{(S,K)}$ for some $t_{(S,K)} \in M$}.
\end{cases}\]

It is straightforward to see that this decomposition satisfies the first two conditions from Theorem \ref{thm:implicitly-computes-decomposition}.
To show that $2$-WL implicitly computes it, we again need to find sets $\mathcal{C}_{\sf bags}, \mathcal{C}_{\sf edges} \subseteq \{\chi_G(v,w) \mid v,w \in V(G)\}$ such that
\begin{align*}
\chi_G(v,w) \in \mathcal{C}_{\sf bags} \;\;&\Leftrightarrow\;\; w \in \beta(t_v)
\intertext{and}
\chi_G(v,w) \in \mathcal{C}_{\sf edges} \;\;&\Leftrightarrow\;\; (t_v,t_w) \in E(T)
\end{align*}
hold for all $v,w \in V(G)$.

We first show statements analogous to Theorem 6 and Corollary 7 from \cite{kieponschwei19} for $2$-connected graphs.
For $k \geq 2$, $\ell \leq k$, and $v_1,\dots,v_\ell \in V(G)$, we use the shorthand $\chi_{G,k}(v_1,\dots,v_\ell)$ for the color $\chi_{G,k}(v_1,\dots,v_\ell,v_\ell,\dots,v_\ell)$.

\begin{lemma}
 \label{lem:edges_conncomp}
 Assume~$k \geq 2$ and let~$G$ and $G'$ be~$2$-connected graphs of minimum degree at least~$3$ that are not~$3$-connected.
 Suppose $(S,K) \in P_0(G)$ and $(S',K') \in P_0(G')$ and that~$u,v \in K$ and $u' \in K', v' \notin K'$.
 Then~$\chi_{G,k}(u,v) \neq \chi_{G',k}(u',v')$.
\end{lemma}

\begin{proof}
 By Corollary \ref{cor:separators}, we can assume that vertices contained in a $2$-separator have different $\chi_k$-colors than vertices which are not contained in any $2$-separator. Thus, vertex pairs connected via a path which does not contain any vertices of $2$-separators obtain different colors than vertex pairs for which every connecting path contains a vertex of a $2$-separator.
 Since $K$ is the vertex set of a connected component of $G - S$, there is a path in $G[K]$ from $u$ to $v$. Moreover, by Part \ref{item:char:min:comp} of Lemma \ref{lem:minimum:degree:implies:disjoin:min:seps}, no vertex on the path is contained in a $2$-separator of $G$. On the other hand, in $G'$, all paths from $u'$ to $v'$ must include a vertex of~$S'$, since~$v' \notin K'$. We conclude that~$\chi_{G,k}(u,v) \neq \chi_{G',k}(u',v')$.
\end{proof}

\begin{lemma}
 \label{lem:bot:vertices:are:different:3con}
 Assume~$k \geq 2$ and let~$G$ and $G'$ be~$2$-connected graphs of minimum degree at least~$3$ that are not~$3$-connected. 
 Then for all~$v \in V(G_\bot)$ and~$v' \in V(G') \setminus V(G'_\bot)$, it holds that~$\chi_{G,k}(v,v) \neq \chi_{G',k}(v',v')$.
\end{lemma}

\begin{proof}
 Let $G, G', v, v'$ be as described in the prerequisites of the lemma.
 By Part~\ref{item:char:min:comp} of Lemma~\ref{lem:minimum:degree:implies:disjoin:min:seps}, the vertex $v'$ is not contained in any $2$-separator of $G'$.
 If $v$ is contained in a $2$-separator of $G$, by Corollary \ref{cor:separators}, the algorithm $k$-WL distinguishes $v$ from $v'$.
 Hence, we can assume otherwise.
 
 We show that there are exactly two vertices each of which is contained in a $2$-separator and can be reached from $v'$ via a path that does not use any other vertices contained in $2$-separators, while this is not the case for $v$.
 
 Since $v' \in V(G') \setminus V(G'_\bot)$, there must be a $2$-separator $S'$ and a connected component of $G' - S'$ with vertex set $K'$ with $(S',K') \in P_0(G')$ and $v' \in K'$.
 Furthermore, since $G'$ is connected and $S'$ is a separator of minimal size, there must be a path in $G'[S' \cup K']$ from~$v'$ to each of the two vertices contained in $S'$.
 Each path can be chosen to avoid the other vertex in $S'$: otherwise, there would be a cut vertex in $G'$, contradicting its $2$-connectedness.
 This means that for the two chosen paths, all but the last vertex are contained in $K'$.
 In particular, by Part \ref{item:char:min:comp} of Lemma \ref{lem:minimum:degree:implies:disjoin:min:seps}, none of these inner path vertices is contained in a $2$-separator.
 Note that all vertices reachable from $v'$ in $G' - S'$ are in $K'$ and thus, there are exactly two vertices that each are contained in a $2$-separator and can be reached from $v'$ via a path that does not use any other vertices contained in $2$-separators. 
 
 We now show that there are at least three vertices that are contained in a $2$-separator and can be reached from $v$ via a path that does not use any other vertices contained in $2$-separators.
 Choose $(S,K) \in P(G)$ such that $v \in K$.
 Let $\tilde{s}_1$ be one of the vertices in $S$.
 Then there is a path from $v$ to $\tilde{s}_1$.
 Let $s_1$ be the first vertex on the path that is contained in a $2$-separator.
 Its existence is guaranteed because $\tilde{s}_1$ is a candidate.
 
 Choose $\tilde{s}_2$ such that $\{s_1,\tilde{s}_2\}$ forms a $2$-separator of $G$.
 Since $s_1$ is not a cut vertex, there must be a path from $v$ to $\tilde{s}_2$ that avoids $s_1$.
 Again, let $s_2$ be the first vertex on the path that is contained in a $2$-separator.
 Now, two cases can occur. In the first case, $\{s_1,s_2\}$ is not a $2$-separator. 
 Then there is a vertex $\tilde{s}_3$ such that $\{s_1, \tilde{s}_3\}$ is a $2$-separator and such that there is a path from $v$ to $\tilde{s}_3$ in $G - \{s_1,s_2\}$.
 Choose $s_3$ as the first vertex on the path that is contained in a $2$-separator. 
 
 In the second case, $\{s_1,s_2\}$ is a $2$-separator.
 Let $\tilde{K}$ be the vertex set of the connected component of $G - \{s_1,s_2\}$ that contains $v$.
 Then $(\{s_1,s_2\},\tilde{K}) \notin P_0(G)$, since $v \in V(G_\bot)$.
 Due to Part \ref{item:char:min:comp} of Lemma \ref{lem:minimum:degree:implies:disjoin:min:seps}, there is a $2$-separator $\tilde{S}$ such that $\tilde{S} \cap \tilde{K} \neq \emptyset$.
 Choose $\tilde{s}_3 \in \tilde{S} \cap \tilde{K}$.
 Since $G[\tilde{K}]$ is connected, there is a path in $G[\tilde{K}]$ from $v$ to $\tilde{s}_3$.
 Let $s_3$ be the first vertex on the path contained in a $2$-separator.

 In both cases, we have found three vertices $s_1,s_2,s_3$ that each are contained in a $2$-separator and such that every $s_i$ is reachable from $v$ via a path that avoids all vertices contained in $2$-separators except for $s_i$.
 
 This property distinguishes $v$ and $v'$ and is detectable by $k$-WL.
 As a result, we obtain that~$\chi_{G,k}(v,v) \neq \chi_{G',k}(v',v')$.
\end{proof}

By Lemma \ref{lem:bot:vertices:are:different:3con}, the color set computed by $2$-WL for vertices in $V(G_\bot)$ is disjoint from the set of colors for vertices in $V(G) \setminus V(G_\bot)$. Since $2$-WL also knows which edges are added to obtain the torso $G_\bot = G[[V(G_\bot)]]$, we conclude that $\chi_G|_{(V(G_\bot))^2}$ is a stable coloring of~$G_\bot$ (with respect to $2$-WL). Thus, by the induction hypothesis, since $G_\bot$ is smaller than~$G$, there are sets $\mathcal{C}^\bot_{\sf bags}, \mathcal{C}^\bot_{\sf edges} \subseteq \{\chi_G(v,w) \mid v,w \in V(G)\}$ such that
\begin{align*}
\chi_G(v,w) \in \mathcal{C}^\bot_{\sf bags} \;\;&\Leftrightarrow\;\; w \in \beta(t_v)
\intertext{and}
\chi_G(v,w) \in \mathcal{C}^\bot_{\sf edges} \;\;&\Leftrightarrow\;\; (t_v,t_w) \in E(T)
\intertext{hold for all $v,w \in V(G_\bot) \subseteq V(G)$. Moreover, by Lemmas \ref{lem:edges_conncomp} and \ref{lem:bot:vertices:are:different:3con}, there is a set $\mathcal{C}^\top_{\sf bags} \subseteq \{\chi_G(v,w) \mid v,w \in V(G)\}$ such that}
\chi_G(v,w) \in \mathcal{C}^\top_{\sf bags} \;\;&\Leftrightarrow\;\; w \in\big(V(G) \setminus V(G_\bot)\big) \wedge v \in \big( K_w \cup N_G(K_w)\big)
\end{align*}
holds for all vertices $v,w \in V(G)$, where $K_w$ is chosen such that there are vertices $u',v'$ for which~$(\{u',v'\}, K_w) \in P_0(G)$ and $w \in K_w$. 
Hence, $\mathcal{C}_{\sf bags} \coloneqq \mathcal{C}^\bot_{\sf bags} \cup \mathcal{C}^\top_{\sf bags}$ gives the desired set of colors for the bags of $(T,\beta)$.

It remains for us to express the edge relation in $T$ via colors computed by $2$-WL. By the existence of $\mathcal{C}^\bot_{\sf edges}$ and the definition of the decomposition, we only need to consider the case that $v \in V(G_\bot)$ and $w \in V(G) \setminus V(G_\bot)$.
Then
\begin{align*}
 \hspace{0.25cm}(t_v,t_w) \in E(T) \;\;&\Leftrightarrow\;\; \exists u \Big(t_u = t_v \wedge \exists v \big(\{u,v\} = N_G(K_w) \wedge v \in \beta(t_u)\big)\Big),
\end{align*}
Just as before, we conclude the existence of a set of colors $\mathcal{C}_{\sf edges}$ with the desired properties. Note, again, that the sets $\mathcal{C}_{\sf bags}$ and $\mathcal{C}_{\sf edges}$ are solely determined by the equivalence class of the graph $G$ with respect to $2$-WL. 
This completes the case that $G$ is $2$-connected, but not $3$-connected, and has minimum degree at least $3$. Thus, the proof of Theorem \ref{thm:implicitly-computes-decomposition} is complete.

\section{Proof of Theorem \ref{thm:reduction-to-3-connected}}

We present the proof of Theorem \ref{thm:reduction-to-3-connected}, which says that, for $k \geq 2$, if $k$-WL determines orbits on the subclass of all arc-colored $3$-connected graphs in a given minor-closed graph class $\mathcal{G}$, then $k$-WL distinguishes all non-isomorphic graphs in $\mathcal{G}$.
To this end, we first show that the coloring computed by $2$-WL is at least as fine as the $\bot$-coloring.

Recall that, for $k \geq 2$, $\ell \leq k$, and vertices $v_1,\dots,v_\ell \in V(G)$, we set $\chi_{G,k}(v_1,\dots,v_\ell) \coloneqq \chi_{G,k}(v_1,\dots,v_\ell,v_\ell,\dots,v_\ell)$.

\begin{lemma}\label{lem:one:component:hanging:sep:pairs}
Let~$\mathcal{G}$ be a minor-closed graph class. For~$k \geq 2$, suppose that $k$-WL determines orbits on the class of all arc-colored~$3$-connected graphs in~$\mathcal{G}$.
Suppose that~$(G, \lambda)$ and~$(G', \lambda')$  with $G, G' \in \mathcal{G}$ are arc-colored~$2$-connected graphs of minimum degree at least~$3$ that are not $3$-connected. Assume that~$(\{s_1,s_2\},K) \in P_0(G)$ and~$(\{s'_1, s'_2\},K') \in P_0(G')$.

Suppose that no isomorphism from the graph~$\big(G_{\top}^{(\{s_1, s_2\},K)}, \lambda_\top^{(\{s_1, s_2\},K)}\big)$ to the graph $\big(G_{\top}'^{(\{s'_1, s'_2\},K')},\lambda_\top'^{(\{s'_1, s'_2\},K')}\big)$ maps~$s_1$ to~$s'_1$. Then
\[
\big\{ \chi_{G,k}(s_1,v) \ \big\vert \ v \in K\big\} \cap \big\{ \chi_{G',k}(s'_1,v) \ \big\vert \ v \in K'\big\} = \emptyset.
\]
\end{lemma}

\begin{proof}
 This proof is an adaption of the proof of \cite[Lemma 7.28]{kiefer:phd}.
 For readability and since the dimension of the algorithm is fixed in this proof, we drop the subscript $k$, i.e., we write~$\chi_{G}$ instead of $\chi_{G,k}$.
 
 If~$\chi_{G}(s_1) \neq \chi_{G'}(s'_1)$, then the conclusion of the lemma is obvious.
 Thus, we can assume otherwise. We have already seen with Corollary~\ref{cor:separators} that~$2$-separators obtain other colors than other pairs of vertices.
 Thus, with Part~\ref{item:char:min:comp} of Lemma~\ref{lem:minimum:degree:implies:disjoin:min:seps}, we can assume that~$G$ and~$G'$ are already colored in a way such that the pairs $(s_1, s_2)$,~$(s_2, s_1)$,~$(s'_1, s'_2)$,~$(s'_2, s'_1)$ have colors different from the colors of pairs of vertices~$(t_1,t_2)$ with~$\{t_1, t_2\} \cap (K \cup K') \neq \emptyset$.
 Also, we may presuppose without loss of generality that~$s_1, s_2, s'_1, s'_2$ have colors different from colors of vertices that are not contained in any~$2$-separator of any of the graphs. 
 
 Moreover, by Lemma~\ref{lem:edges_conncomp}, we may assume that pairs of vertices that are both contained in $K$ have other colors than pairs of vertices for which just one vertex is contained in $K$ (and similarly for $K'$ in $G'$). 

 For better readability, in the rest of this proof, we drop the superscripts~$(\{s_1, s_2\},K)$ and~$(\{s'_1, s'_2\},K')$.
 That is, whenever we write $G_\top$, we mean $G_{\top}^{(\{s_1, s_2\},K)}$, and whenever we write $G'_\top$, we actually mean $G_{\top}'^{(\{s'_1, s'_2\},K')}$.
 Similarly, $\lambda_\top$ and $\lambda'_\top$ stand for their versions with the superscripts~$(\{s_1, s_2\},K)$ and~$(\{s'_1, s'_2\},K')$.   

 \begin{claim}
  For all~$i \geq 0$, all~$u,v \in K \cup \{s_1, s_2\}$ and all~$u',v' \in K' \cup \{s'_1, s'_2\}$ with $\{u,v\} \not\subseteq \{s_1,s_2\}$ and~$\{u',v'\} \not\subseteq \{s'_1,s'_2\}$, the following implication holds: 
  \begin{align}\label{imp:induction2} 
   \chi^i_{G_\top}(u,v) \neq \chi^i_{G'_\top}(u',v') \ \Rightarrow \ \chi^i_G(u,v) \neq \chi^i_{G'}(u',v').
  \end{align}
 \end{claim}
 \begin{claimproof}
  We proceed via induction on the number $i$ of iterations. For the base step $i = 0$, suppose that $\chi^0_{G_\top}(u,v) \neq \chi^0_{G'_\top}(u',v')$.
  Then either there is no isomorphism from~$G_\top[\{u,v\}]$ to $G'_\top[\{u',v'\}]$ mapping~$u$ to~$u'$ and~$v$ to~$v'$, or~$\lambda_\top(u,v) \neq \lambda'_\top(u',v')$. 
 
  In the first case, by definition of the graphs~$G_\top$ and~$G'_\top$ and their colorings, we immediately obtain~$\chi^0_G(u,v) \neq \chi^0_{G'}(u',v')$, since an isomorphism from~$G[\{u,v\}]$ to~$G'[\{u',v'\}]$ that maps~$u$ to~$u'$ and~$v$ to~$v'$ would induce an isomorphism from the graph $G_\top[\{u,v\}]$ to~$G'_\top[\{u',v'\}]$ with the same mappings.
  Thus, consider the second case, i.e., that~$\lambda_\top(u,v) \neq \lambda'_\top(u',v')$.
  Since we have that both~$\{u,v\} \not\subseteq \{s_1,s_2\}$ and $\{u',v'\} \not\subseteq \{s'_1,s'_2\}$ hold, we obtain that $(\lambda(u,v),2) = \lambda_{\top}(u,v) \neq \lambda'_{\top}(u',v') = (\lambda'(u',v'),2)$.
  Hence, $\lambda(u,v) \neq \lambda'(u',v')$, which implies that $\chi_G^0(u,v) \neq \chi_{G'}^0(u',v')$.
 
  For the induction step from $i$ to $i+1$, assume there exist vertices~$u,v \in K \cup \{s_1,s_2\}$ and~$u',v' \in K' \cup \{s'_1, s'_2\}$ that satisfy~$\{u,v\} \not\subseteq \{s_1, s_2\}$ and~$\{u',v'\} \not\subseteq \{s'_1, s'_2\}$, and furthermore~$\chi^i_{G_\top}(u,v) = \chi^i_{G'_\top}(u',v')$ and~$\chi^{i+1}_{G_\top}(u,v) \neq \chi^{i+1}_{G'_\top}(u',v')$. 
 
  Just like in the proof of \cite[Lemma 7.28]{kiefer:phd}, we can assume $\chi^i_{G_\top}(u,u) = \chi^i_{G'_\top}(u',u')$ and $\chi^i_{G_\top}(v,v) = \chi^i_{G'_\top}(v',v')$. Thus, there must be a color tuple~$(c_1,c_2)$ such that the sets
\begin{align*}
M \coloneqq{} & \big\{x \ \big\vert \ x \in V(G_\top) \setminus \{u,v\},
\\
& \phantom{\big\{x \mid{}} \left(\chi^i_{G_\top}(x,v), \chi^i_{G_\top}(u,x)\right) = (c_1,c_2)\big\}
\intertext{and} 
M' \coloneqq{} & \big\{x' \ \big\vert \ x' \in V(G'_\top) \setminus \{u',v'\}, \\
& \phantom{\{x' \mid{}} \big(\chi^i_{G'_\top}(x',v'), \chi^i_{G'_\top}(u',x')\big) = (c_1,c_2)\big\}
\intertext{do not have the same cardinality. Let}
D \coloneqq{} & \Big\{\big(\chi^i_{G}(x,v), \chi^i_{G}(u,x)\big) \ \Big\vert \ x \in M\Big\} \, \cup{} 
\\
& \Big\{\big(\chi^i_{G'}(x',v'), \chi^i_{G'}(u',x')\big) \ \big\vert \ x' \in M'\big\}.
\end{align*}
  We show that
  \begin{align*}
   \Big\{x \ \big\vert \ x \in V(G) \setminus \{u,v\}, \big(\chi^i_{G}(x,v), \chi^i_{G}(u,x)\big) \in D\Big\} &= M.
  \end{align*}
  The inclusion ``$\supseteq$'' is clear by the definitions of $M$ and $D$.
  For the inclusion ``$\subseteq$'', let $x \in V(G) \setminus \{u,v\}$ be a vertex with $\big(\chi^i_{G}(x,v), \chi^i_{G}(u,x)\big) \in D$.
  Then, by the definition of~$D$, there must be a vertex $x' \in M$ with $\big(\chi^i_{G}(x',v), \chi^i_{G}(u,x')\big) = \big(\chi^i_{G}(x,v), \chi^i_{G}(u,x)\big)$.
  Thus, $\left(\chi^i_{G_\top}(x',v), \chi^i_{G_\top}(u,x')\right) = (c_1,c_2)$ by the definition of $M$.
  Now, using the induction assumption for $i$ stated in \eqref{imp:induction2}, we obtain that $\chi^i_{G_\top}(x,v) = \chi^i_{G_\top}(x',v) = c_1$ and $\chi^i_{G_\top}(u,x) = \chi^i_{G_\top}(u,x') = c_2$.
  This implies $x \in M$.
 
  Similarly, we have
  \begin{align*}
   \Big\{x' \ \big\vert \ x' \in V(G') \setminus \{u',v'\}, \big(\chi^i_{G'}(x',v'), \chi^i_{G'}(u',x')\big) \in D\Big\} &= M'.
  \end{align*}
  Hence, these sets do not have equal cardinalities. Thus,~$\chi^{i+1}_G(u,v) \neq \chi^{i+1}_{G'}(u',v')$. 
 \end{claimproof}

 Having shown the claim, it suffices to show that 
 \[\big\{ \chi_{G_\top}(s_1,v) \ \big\vert \ v \in K\big\} \cap \big\{ \chi_{G'_\top}(s'_1,v') \ \big\vert \ v' \in K'\big\} = \emptyset.\]
 For this, it suffices to prove that~$\chi_{G_\top}(s_1) \neq \chi_{G'_\top}(s'_1)$ holds. 

 The graphs~$(G_{\top}^{(\{s_1,s_2\},K)}, \lambda_\top^{(\{s_1,s_2\},K)})$ and~$(G_{\top}'^{(\{s'_1,s'_2\},K')}, \lambda_\top'^{(\{s'_1,s'_2\},K')})$ are~$3$-con\-nected. Thus, since we have assumed that~$k$-WL determines orbits on the class of all arc-colored~$3$-connected graphs in~$\mathcal{G}$, if~$\chi_{G_\top}(s_1) = \chi_{G'_\top}(s'_1)$ held, there would have to be an isomorphism from~$(G_{\top}^{(\{s_1, s_2\},K)}, \lambda_\top^{(\{s_1, s_2\},K)})$ to~$(G_{\top}'^{(\{s'_1, s'_2\}',K')},\lambda_\top'^{(\{s'_1, s'_2\},K')})$ that maps~$s_1$ to~$s'_1$, a contradiction to our assumption.
\end{proof}

Using the lemma, we can show the following.

\begin{lemma}\label{lem:distinct_separatingpairs}
Let~$\mathcal{G}$ be a minor-closed graph class. Assume~$k \geq 2$ and suppose that $k$-WL determines orbits on the class of all arc-colored~$3$-connected graphs in~$\mathcal{G}$. Suppose~$(G,\lambda)$ and $(G',\lambda')$ with $G, G' \in \mathcal{G}$ are arc-colored~$2$-connected graphs of minimum degree at least~$3$ that are not $3$-connected and let~$\{s_1,s_2\} \subseteq V(G)$ and~$\{s'_1, s'_2 \} \subseteq V(G')$ be~$2$-separators of~$G$ and~$G'$, respectively. 

Suppose no isomorphism from $\big(G_{\top}^{\{s_1, s_2\}}, \lambda_\top^{\{s_1, s_2\}}\big)$ to $\big(G_{\top}'^{\{s'_1, s'_2\}},\lambda_\top'^{\{s'_1, s'_2\}}\big)$ maps $s_1$ to~$s'_1$. Then~$\chi_{G,k}(s_1,s_2) \neq \chi_{G',k}(s'_1,s'_2)$. 
\end{lemma}

\begin{proof} 
 This proof is similar to the proof of \cite[Lemma 10]{kieponschwei19}.

 If there is no $K$ with $(\{s_1,s_2\},K) \in P_0(G)$ or there is no $K'$ with $(\{s'_1,s'_2\},K') \in P_0(G')$, the statement of the lemma is trivial.

Otherwise, suppose that
\begin{align*}
\{K_1,\dots,K_t\} & = \big\{K \ \big\vert \ (\{s_1,s_2\},K) \in P_0(G)\big\}
\intertext{and that}
\{K'_1,\dots,K'_{t'}\} & = \big\{K' \ \big\vert \ (\{s'_1,s'_2\},K') \in P_0(G')\big\}.
\end{align*}
Since there is no isomorphism from~$\big(G_{\top}^{\{s_1, s_2\}}, \lambda_\top^{\{s_1, s_2\}}\big)$ to~$\big(G_{\top}'^{\{s'_1, s'_2\}}, \lambda_\top'^{\{s'_1, s'_2\}}\big)$ that maps~$s_1$ to~$s'_1$, there is an arc-colored graph~$(H,\lambda_H)$ such that the sets
\begin{align*}
I & \coloneqq \left\{i \ \bigg\vert \ \left(G_{\top}^{(\{s_1,s_2\},K_i)}, \lambda_\top^{(\{s_1,s_2\},K_i)}\right)_{(s_1)} \cong (H,\lambda_H)\right\}
\intertext{and}
I' & \coloneqq \left\{j \ \bigg\vert \ \left({G'}_{\top}^{(\{s'_1,s'_2\},K'_j)}, \lambda_\top'^{(\{s'_1,s'_2\},K'_j)}\right)_{(s'_1)} \cong (H,\lambda_H)\right\}
\end{align*}
have different cardinalities. Note that all~$K_i$ with~$i\in I$ and all~$K'_j$ with~$j\in I'$ have the same cardinality. By Lemma~\ref{lem:one:component:hanging:sep:pairs}, we know that for~$v\in K_i$ with~$i\in I$ and~$v'\in K'_{j}$ with~$j\notin I'$, we have~$\chi_{G,k}(s_1,v)\neq \chi_{G',k}(s'_1,v')$. Furthermore, using Lemma 
\ref{lem:bot:vertices:are:different:3con}, it holds for $v \in V\big(G^{\{s_1,s_2\}}_\top\big) \setminus \{s_1,s_2\}$ and $v' \notin V\big(G'^{\{s'_1,s'_2\}}_\top\big) \setminus \{s'_1,s'_2\}$ that $\big(\chi_{G,k}(s_1,v),\chi_{G,k}(v,s_2)\big) \neq \big(\chi_{G',k}(s'_1,v'),\chi_{G',k}(v',s'_2)\big)$. Indeed, $v$ is a vertex in $V(G) \setminus V(G_\bot)$ from which there is, for every~$i \in \{1,2\}$, a path to $s_i$ that does not contain any vertex in a $2$-separator apart from $s_i$, while the analogous statement does not hold for $v'$. 

Letting
\[C \coloneqq \big\{\big(\chi_{G,k}(s_1,v),\chi_{G,k}(v,s_2)\big) \ \big\vert \ \exists i \in I \text{ s.t.\ } v\in K_i \big\},\]
the sets
\begin{align*}
\big\{v \in V(G) \ &\vert \ \big(\chi_{G,k}(s_1,v),\chi_{G,k}(v,s_2)\big)\in C\big\}
\intertext{and}
\big\{v' \in V(G') \ &\vert \ \big(\chi_{G',k}(s'_1,v'),\chi_{G',k}(v',s'_2)\big)\in C\big\}
\end{align*}
do not have the same cardinality. We conclude that~$\chi_{G,k}(s_1,s_2) \neq \chi_{G',k}(s'_1,s'_2)$.
\end{proof}

We can now deduce that the coloring computed by $2$-WL is at least as fine as the~$\bot$-coloring.

\begin{corollary}\label{cor:induced:on:bot:graphs:is:finer:3con}
 Let~$\mathcal{G}$ be a minor-closed graph class. Assume~$k \geq 2$ and suppose that $k$-WL determines orbits on the class of all arc-colored~$3$-connected graphs in~$\mathcal{G}$. Suppose~$(G,\lambda)$ and $(G',\lambda')$ with $G, G' \in \mathcal{G}$ are arc-colored~$2$-connected graphs of minimum degree at least~$3$ that are not $3$-connected. If, for vertices $v_1,v_2 \in V(G_\bot)$ and~$v'_1,v'_2\in V(G'_\bot)$, it holds that~$\chi_{G_\bot,k}(v_1,v_2) \neq \chi_{G'_\bot,k}(v'_1,v'_2)$, then~$\chi_{G,k}(v_1,v_2) \neq \chi_{G',k}(v'_1,v'_2)$. 
\end{corollary}

\begin{proof}
Without loss of generality, we may assume that $\lambda(v_1,v_1) = \lambda'(v'_1,v'_1)$ and $\lambda(v_2,v_2) = \lambda'(v'_2,v'_2)$ hold. Furthermore, by Lemma~\ref{lem:bot:vertices:are:different:3con}, with respect to the colorings~$\chi_{G,k}$ and~$\chi_{G',k}$, the vertices in~$V(G_{\bot})$ and~$V(G'_{\bot})$ have other colors than the vertices in~$V(G) \setminus V(G_{\bot})$ and~$V(G') \setminus V(G'_{\bot})$. Thus, it suffices to show that the colorings~$\chi_{G,k}$ and~$\chi_{G',k}$ refine the colorings~$\lambda_\bot$ and~$\lambda'_\bot$, respectively, on the domains of those. Thus, by the definition of~$\lambda_\bot$ and~$\lambda'_\bot$ and since these maps are only defined on tuples $(u,v)$ for which either $u=v$ or~$\{u,v\}$ forms a $2$-separator or an edge, it suffices to show the following two statements.

 \begin{enumerate}
  \item If~$\{v_1,v_2\},\,\{v'_1,v'_2\}$ are $2$-separators and
   \[\ISOTYPE\big((G^{\{v_1,v_2\} }_{\top},\lambda^{\{v_1,v_2\}}_\top)_{(v_1)}\big) \neq \ISOTYPE\big((G'^{\{v'_1,v'_2\} }_{\top},\lambda'^{\{v'_1,v'_2\}}_\top)_{(v'_1)}\big),\]
   then~$\chi_{G,k}(v_1,v_2) \neq \chi_{G',k}(v'_1,v'_2)$. 
  \item If~$v_1=v_2$ and~$v'_1=v'_2$, or~$\{v_1,v_2\}\in E(G)$ and~$\{v'_1,v'_2\}\in E(G')$, but~$\{v_1,v_2\}$ and~$\{v'_1,v'_2\}$ are not $2$-separators, then, if~$\lambda_\bot(v_1,v_2) \neq \lambda'_\bot(v'_1,v'_2)$, we have $\chi_{G,k}(v_1,v_2) \neq \chi_{G',k}(v'_1,v'_2)$.
 \end{enumerate}

 The first part is exactly Lemma~\ref{lem:distinct_separatingpairs}.
 For the second part, from the definition of~$\lambda_\bot$ and~$\lambda'_\bot$, we obtain~$\lambda(v_1,v_2) \neq \lambda'(v'_1,v'_2)$, which implies~$\chi_{G,k}(v_1,v_2) \neq \chi_{G',k}(v'_1,v'_2)$. 
\end{proof}

We finally reach the proof of Theorem \ref{thm:reduction-to-3-connected}. 

\begin{proof}[Proof of Theorem \ref{thm:reduction-to-3-connected}]
 Let $\mathcal{G}$ be a minor-closed graph class and let $k \geq 2$. Let $G, G' \in \mathcal{G}$ be graphs with arc colorings $\lambda$ and $\lambda'$, respectively, such that $(G, \lambda) \not\cong (G', \lambda')$ and suppose that $k$-WL determines orbits on the class of all arc-colored $3$-connected graphs in $\mathcal{G}$. To show that $G$ and $G'$ are distinguished, we proceed just as outlined in \cite[Section 5]{kieponschwei19}, improving the lower bound on the dimension $k$ of the WL algorithm stated there from $3$ to $2$ by using our new results. 
 
 Namely, we prove the theorem by induction on~$|V(G)|+|V(G')|$.
 The base case~$|V(G)|+|V(G')| = 2$ is trivial.
 For the induction step, if both graphs are~$3$-connected, then the statement follows directly from the assumptions since ``determines orbits'' is a stronger assumption than ``distinguishes'' (see \cite[Observation 3.16]{kiefer:phd}).
 If exactly one of the graphs is~$3$-connected, then exactly one of them has a~$2$-separator and the statement follows from Corollary~\ref{cor:separators}.
 
 Thus, suppose that both graphs are not~$3$-connected and assume all pairs of arc-colored graphs~$(H,\lambda_H)$ and~$(H',\lambda_{H'})$ with~$|V(H)| + |V(H')| < |V(G)| + |V(G')|$ are distinguished.
 By \cite[Theorem 5]{kieponschwei19}, with the graph class~$\mathcal{G}$ from the theorem being the class of graphs containing every graph isomorphic to~$G$,~$G'$, or a minor of one of them, it suffices to show the statement for the case that~$G$ and~$G'$ are~$2$-connected.

 If~$G$ and~$G'$ do not have the same minimum degree, they are distinguished by their degree sequences.
 Now first suppose both~$G$ and~$G'$ have minimum degree at least $3$.

 By Corollary \ref{cor:induced:on:bot:graphs:is:finer:3con}, the partition of the vertices and arcs induced by the coloring~$\chi_{G,k}$ restricted to~$V(G_{\bot})$ is finer than the partition induced by~$\lambda_\bot$.
 Also, by Lemma \ref{lem:bot:vertices:are:different:3con}, vertices in $V(G_{\bot})$ obtain different colors than vertices in $V(G) \setminus V(G_{\bot})$.
 Similar statements hold in~$G'$. 

 By \cite[Lemma 4]{kieponschwei19}, the arc-colored graphs $(G_{\bot}, \lambda_\bot)$ and $(G'_{\bot},\lambda'_\bot)$ are not isomorphic and thus distinguished by~$k$-WL by the induction assumption.
 Hence, since it computes finer colorings on $(G, \lambda)$ and $(G', \lambda')$ than on $(G_\bot, \lambda_\bot)$ and on $(G'_\bot, \lambda'_\bot)$, respectively, the algorithm~$k$-WL also distinguishes~$(G, \lambda)$ from~$(G', \lambda')$. 
 
 By \cite[Lemma 15]{kieponschwei19}, the case that~$G$ and~$G'$ do not have minimum degree at least~$3$ reduces to the case of minimum degree at least~$3$, letting~$\mathcal{G}$ from the lemma be the class of graphs containing every graph isomorphic to~$G$,~$G'$, or a minor of one of them.
\end{proof}

\end{document}